\documentclass[11pt]{article}
\usepackage[margin=1in]{geometry}
\usepackage[ruled]{algorithm2e} 

\SetAlFnt{\small}
\SetAlCapFnt{\small}
\SetAlCapNameFnt{\small}
\SetAlCapHSkip{0pt}
\IncMargin{-\parindent}

\usepackage{amsmath}
\usepackage{amsfonts}
\usepackage{bbm}
\usepackage{amsthm}
\usepackage{cancel}
\usepackage{hyperref}
\hypersetup{colorlinks=true,breaklinks=true,bookmarks=true,urlcolor=darkblue,citecolor=darkblue,linkcolor=darkblue,bookmarksopen=false,draft=false}
\usepackage{mathtools}
\usepackage{thmtools} 
\usepackage{thm-restate}
\usepackage{caption}
\usepackage{subcaption}
\usepackage{graphicx}
\usepackage{bm}
\usepackage{amssymb}
\graphicspath{{./figures/}}

\newtheorem{condition}{Condition}
\theoremstyle{definition}

\newcommand{\reals}{\mathbb{R}}
\newcommand{\E}{\mathop{\mathbb{E}}}
\newcommand{\ones}{\mathbf{1}}

\newcommand{\D}{\mathcal{D}}

\newcommand{\R}{\mathcal{R}}

\newcommand{\CS}{\mathcal{S}}

\newcommand{\OA}{\textrm{OA}}
\newcommand{\DG}{\textrm{DG}}
\newcommand{\DMI}{\textrm{DMI}}
\newcommand{\RBTS}{\textrm{RBTS}}
\newcommand{\twobin}{\{L, H\}}
\DeclareMathOperator{\PC}{P_{\text{cont}}}

\usepackage{xcolor}
\usepackage{color}
\definecolor{gray}{gray}{0.5}
\definecolor{lightred}{rgb}{1,0.6,0.6}
\definecolor{darkgreen}{rgb}{0,0.5,0}
\definecolor{myorange}{rgb}{0.8,0.7,0.5}
\definecolor{darkblue}{rgb}{0.0,0.0,0.5}
\usepackage[colorinlistoftodos,textsize=tiny]{todonotes}
\tikzset{notestyleraw/.append style={inner sep = 2pt}}

\usepackage{natbib}
\setcitestyle{authoryear,open={[},close={]}} %

\newtheorem{lemma}{Lemma}
\newtheorem*{lemma*}{Lemma}
\newtheorem{corollary}{Corollary}
\newtheorem{theorem}{Theorem}
\newtheorem*{theorem*}{Theorem}
\newtheorem{proposition}{Proposition}
\newtheorem{observation}{Observation}
\theoremstyle{definition}\newtheorem{definition}{Definition}
\theoremstyle{definition}
\theoremstyle{definition}
 
\begin{document}

 \title{Peer Prediction with More Signals than Reports}
 \author{%
Rafael Frongillo \\ University of Colorado Boulder \\ \texttt{raf@colorado.edu} \and
Ian Kash \\ University of Illinois at Chicago \\ \texttt{iankash@uic.edu} \and
Mary Monroe \\ University of Colorado Boulder \\ \texttt{mary.monroe@colorado.edu}}
 \date{}
 
 \maketitle

\begin{abstract}
  Peer prediction mechanisms are typically proposed and analyzed under the assumption that the report and signal spaces are identical.
  In practice, however, agents often observe richer information which they then map to a coarser report space.
  Motivated by this discrepancy between theory and practice, we initiate the study of peer prediction mechanisms with signal spaces that are richer than the report space.
  We begin by formalizing a model with real-valued signals and binary reports.
  In this setting, it is natural to study symmetric threshold strategies, where agents map their signals to binary reports according to a single real-valued threshold.
  For several well-known binary-report peer prediction mechanisms, we show that most equilibria under the original assumption of binary signals are no longer equilibria in our model.
  Furthermore, dynamic analysis proves that some of the remaining thresholds are unstable.
  These results extend beyond real-valued signals and binary reports to settings where the signal space is finer-grained than the report space.

  While the results above suggest important limitations for the deployment of existing peer prediction mechanisms in practice, we also use them to develop a new, more robust mechanism. 
  This mechanism generates a larger number of stable threshold equilibria under our model, thus allowing the designer more flexibility in choosing how agents map their signals to reports.
\end{abstract}

\section{Introduction}

Eliciting high-quality information from agents without verification is a well-studied problem, motivated by settings such as peer grading or data labeling where the ground truth is unknown, subjective, or hard to acquire.
\emph{Peer prediction} mechanisms aim to incentivize truthful behavior of agents by paying them based on how their reports correlate with other submissions, without the need for ground truth.
Agents are typically modeled as receiving a signal, for example $H$ (``high'') or $L$ (``low'') for the quality of an essay, from a fixed joint probability distribution $P$ unknown to the mechanism.
Upon seeing their signal, agents update their beliefs and, under the right conditions, are incentivized to truthfully report the signal they receive.

In practice, however, the signal space of information that agents receive is much richer than the small number of categories typically studied in the literature.
An essay contains far more information than a simple $H$ or $L$ impression.\footnote{One implication of a more nuanced signal space which we do not focus on is the well-known issue of spurious correlation, such as adopting a strategy of conditioning one's report on the first word of the essay; see \S~\ref{sec:discussion}.}
Existing results often rely on the assumption of a one-to-one mapping from signals to reports, making it unclear how one would expect these mechanisms to perform in practice.

To address this question, we introduce and study a model of peer prediction with richer signal spaces.
We primarily focus on the binary-report setting, where a designer wishes to elicit binary information (e.g. is an essay's equality $H$ or $L$?) from an agent about a task.
Here we assume that agents receive a real number as their signal and then based on it select a report of $H$ or $L$.
This real-valued model captures settings where agents form a range of posterior beliefs about a task.
We focus on a natural class of threshold strategies, where an agent reports $H$ if and only if their signal exceeds a particular threshold.
For example, an agent may report $H$ if they deem the quality of an essay to be above 0.7, and $L$ otherwise.

We generally find that mechanisms which are truthful under the original binary signal model fail to yield correct incentives under the more nuanced model with real-valued signals.
Intuitively, agents with signals near the threshold may have an incentive to misreport, as the separation that enforces truthfulness breaks down.
We give necessary and sufficient conditions for a threshold to be an equilibrium, i.e., not suffer from this issue, in a variety of peer prediction mechanisms.
Furthermore we study \emph{dynamics} arising if a small fraction of agents with the greatest incentive to do so change their report, causing the threshold to slowly shift. 
We then distinguish \emph{stable} equilibria under our dynamics as reasonable to arise naturally, in contrast to unstable equilibria.

\subsection{Motivating example: Output Agreement (OA)}

Consider Output Agreement (OA), which provides a reward of 1 if an agent's binary report of $L$ or $H$ agrees with a peer and 0 otherwise~\citep{von2004labeling,von2008designing}.
Traditional analysis assumes the signal space is the same as the report space, i.e. $\{L,H\}$.
In contrast, our real-valued signal model assumes each agent $i$ recieves a real-valued signal $X_i$ from some joint distribution. 
As a running example, let $X_i = Z + Z_i$ with $Z$ and $Z_i$ both unit normals, where $Z$ is shared by all agents and $Z_i$ is signal noise unique to agent $i$.
In a peer grading setting the signals $X_i$ might capture a fine-grained assessment of the quality of an assignment.

The mechanism designer has some intended mapping between the real-valued signals and binary reports.
Given the interpretation of $H$ and $L$ as ``high'' and ``low'',  it is natural for this mapping to have a real-valued threshold $\tau$, with values above the threshold mapping to $H$, and below mapping to $L$.
In peer grading this assumption naturally corresponds to an intution that a ``good-enough'' assignment might be graded as satisfactory while a poor one would be deemed unsatisfactory.

Let us suppose first that the mechanism designer announces a desired threshold of $\tau=0$.
In the binary signal model with this threshold and distribution (i.e. agents only receive a signal of $H$ or $L$ as determined by $\tau$), OA is truthful: after receiving $H$, an agent believes it is more likely that their peer received $H$ than $L$ and vice versa ($\Pr[H \mid H],\Pr[L \mid L] > 0.5$).

In the real-valued case, we can no longer speak of OA as truthful because it is no longer a direct revelation mechanism, but we can ask whether reporting as intended by $\tau$ is a Bayes-Nash equilibrium.
It turns out $\tau=0$ is such an equilibrium.
For example, if $x_i = 0.11$, we have $\Pr[X_j > 0 \mid x_i] \approx 0.52 > 0.5$, so reporting $H$ as intended is indeed the best response.

Now suppose the desired threshold were $\tau=0.1$.  In the binary signal model, $\Pr[H \mid H] = \Pr[X_j > 0.1 \mid X_i > 0.1] \approx 0.64$ and $\Pr[L \mid L] \approx 0.69$ are both greater than $0.5$, so OA remains truthful.
Now in the real-valued model consider an agent receiving a signal of $x_i = 0.11$.  As $x_i$ is just above the threshold, the intended report is $H$.
But  $\Pr[X_j > 0.1 \mid x_i] \approx 0.48 < 0.5$.  Thus the best response is $L$, so reporting according to $\tau$ is {\em not} an equilibrium.
The above analysis will give the same qualitative result for any finite nonzero $\tau$, intuitively because for $\tau > 0$ an agent whose signal lies just above the threshold still assigns higher probability to a peer's signal falling below $\tau$ than above it; analagous reasoning applies for negative thresholds.
Thus, the \emph{only} finite initial threshold $\tau$ for which reporting according to $\tau$ is an equilibrium is $\tau=0$.

Given that $\tau=0.1$ is not an equilibrium, what should we expect to happen if the mechanism designer announces it?
We saw that agents with signals $x_i$ slightly above $0.1$ have an incentive to report $L$ instead of $H$.  This incentive weakens as $x_i$ grows, so it is natural to expect only a small fraction of agents to deviate from the intended strategy, specifically those with signals closest to but above $0.1$.  But effectively, this raises our threshold from $\tau_0 = 0.1$ to some higher value, say $\tau_1 = 0.13$.  Since the same argument applies to our new threshold, in future tasks further deviations could lead to $\tau_2 = 0.16$.  Thus, we should expect that, over time, the system will converge toward the ``trival'' equilibrium at $\tau=\infty$ where agents report $L$ regardless of their signal.
The same logic shows that announcing $\tau < 0$ will lead toward the trivial ``always $H$'' equilibrium at $\tau = -\infty$.

From a dynamical systems perspective, the above shows that the $\tau = 0$ equilibrium is {\em unstable} and can only be realized by starting exactly at it.  As there are no other finite equilibria, any variation or small error in the model will send the threshold toward infinity.
Given the inherent noisiness of many peer prediction applications, we argue similar to \citet{shnayder2016measuring} that {\em the only reasonable equilibra in practice are stable ones}.
Stability is thus a crucial equilibrium refinement for peer prediction.

To summarize, the standard analysis of OA in the binary signal model concludes that the mechanism has desirable incentive properties as long as agents believe others are more likely to share their signal than have the opposite.  In contrast, by introducing our more nuanced real-valued signal model, we have seen that the only reasonable outcome to expect from OA is a ``trivial'' equilibrium where no information is gained.  Thus, at least in this example, {\em this shift in perspective takes OA from useful to entirely useless}.

\subsection{Results}

We explore the generality of our observation that moving to a real-valued signal model substantially changes the incentives and behavior of peer prediction mechanisms.
We begin in \S~\ref{sec:OA-Model} by formally analyzing the Output Agreement (OA) mechanism mentioned above.
As discussed, not only do uninformative equilibria exist where all agents submit the same report, but they also remain stable under dynamics.
That is, at least for some natural distributions, we would expect any initial threshold to drift toward $\pm\infty$ over time, eventually leading to blind agreement. 

In the binary signal setting, a multi-task mechanism proposed by~\citet{dasgupta2013crowdsourced} (DG) fixes the uninformative equilibria problem in OA.
The DG mechanism adds a penalty corresponding to the empirical frequency of agents' reports across other tasks, effectively discouraging blind agreement and leading to strictly worse payments in uninformative equilibria.
Under our richer signal model, we provide necessary and sufficient conditions for a finite threshold to be an equilibrium.
Under some conditions like monotonicity, these equilibria are stable, while the uniformative equilibria are unstable.
However, we find that in many cases there exists only one such equilibrium determined by the underlying information structure, meaning the designer has no control over how agents decide to map their signals to $H$ versus $L$. 
We conclude that DG is potentially more useful than OA, but only allows limited control over the equilibrium threshold. 

We perform similar analyses for two other more complex mechanisms, the Robust Bayesian Truth Serum (RBTS)~\citep{witkowski2012robust} and Determininant Mutual Information (DMI) \citep{kong2020dominantly,kong2024dominantly}.
In both cases, similar characterizations of equilibria and stability emerge as DG.
For all four mechanisms we augment our general analysis with a worked example of the Gaussian case and numerical illustrations that exhibit similar results under skewness and multimodality.
Finally, we demonstrate that these limitations persist in settings of finer-grained signals beyond the real-valued case, as well as in non-binary report settings using the Correlated Agreement (CA) mechanism~\citep{shnayder2016informed}.
Collectively, then, we find that current peer prediction mechanisms are inflexible under a richer signal space in many informational settings.

A natural next question, then, is whether we can come up with a new framework that grants the designer more choice in how agents map information to reports. 
In Section~\ref{sec:positive-results}, we propose the following report-mapping framework: ask agents for a raw, non-binary report using the Correlated Agreement mechanism, e.g., agents submit $\hat L$, $\hat M$ (medium), and $\hat H$. 
Then the designer can impose a mapping back to a final binary report of $L$ or $H$, with potentially more flexibility.
Using our results for non-binary reports, we show that CA parameters can be tuned within this mapping framework to robustly generate a larger number of stable thresholds than binary-report mechanisms can achieve under the Gaussian model.

Our paper serves as a launching point for uncovering how behavior changes under peer prediction in the real world.
We show that under a more realistic signal model, the designer has much less flexibility to establish a mapping between information and reports under traditional binary-report mechanisms.
We also show that some of this flexibility can be recovered by leveraging the Correlated Agreement mechanism and mapping larger report spaces back to a binary decision.
Our results clarify the implications of choosing various mechanisms in practice, and provide tools to better control the semantics of their reports. 

\subsection{Related Work}
For a general overview of peer prediction, we direct the reader to \citet{faltings2023game,lehmann2024mechanisms,frongillo2024recent}.
Most relevant to our work is the distinction between peer prediction mechanisms which are {\em minimal}, in that agents only report their signal (as in OA), and those where additional information about agent beliefs is reported (e.g.~the agent's posterior belief about the report of another agent as in the Bayesian Truth Serum~\citep{prelec2004bayesian} and its variants).
Another important dimension is whether the mechanism is for a single task or a larger \emph{multi-task} collection where reports from unrelated tasks can be compared to improve incentives.
We examine two multi-task mechanisms (DG and DMI) and a non-minimal one (RBTS).
While more robust than OA, our results show that neither of these structures avoids the issues our work raises.

To our knowledge, our approach of assuming agents have a fundamentally richer signal than they are asked to reveal is novel.
The idea of using dynamical stability for equilibrium refinement in peer prediction was introduced by \citet{shnayder2016measuring}.
That work looks at a population of agents playing strategies from a discrete set, such as $\{$ truthful, always $H$, always $L$ $\}$.
Using replicator dynamics, they measure the stability of equilibria by their basin of attraction, the volume of initial conditions leading to that equilbrium.
At a very high level, their conclusions bear similarity to ours: useful equilibria are less robust for OA, and multi-task mechanisms like DG are more robust.

Looking closer, two key differences between our work and \citet{shnayder2016measuring} are the signal model and strategic model.
For the signal model, they consider a discrete signal space matching the report space, as usual.
For the strategic model, they consider agents learning whether to switch to another strategy in a discrete set, whereas we allow agents to choose any threshold strategy.
As a result, our results differ substantially in some cases: for OA, the truthful equilibrium is unstable even in a technical sense (measure zero basin of attraction), and for DG, ``truthfulness'' as given by the mechanism designer's desired threshold need not be an equilibrium.

The model and tools in our paper also bear some resemblance to work in the epistemic democracy literature. 
Specifically,~\citet{duggan2001bayesian} and~\citet{meirowitz2002informative} study a common-value voting setting, where there is a ground truth ``correct'' alternative that all agents prefer, and they each receive signals about which alternative that is.
As in our paper, the alternative space is binary while the signal space is continuous.
The authors similarly identify threshold equilibria, where agents map their signals to a vote for either alternative according to a cutoff point, and prove that under continuity and monotonicity conditions over the prior densities, there is a unique symmetric Bayes-Nash equilibrium characterized by a threshold. 
We note that though the tools are similar, the setting and goals of the papers are quite different:
the authors aim to show that the ``right'' alternative is chosen with high probability as the number of agents grows large.

\section{Output Agreement}
We begin by analyzing the popular peer prediction mechanism Output Agreement (OA) when signals are real-valued and reports are binary.
Output Agreement is a minimal, single task mechanism: each agent $i$ submits a report for a task and receives a positive payment if their report matches that of another agent $j$ on the same task.
In the binary signal model, truthful reporting forms a Bayes-Nash equilibrium under some assumptions about the information structure.
Specifically, truthfulness holds under \emph{strong diagonalization}, where the conditional probability of another signal is maximized by the signal being conditioned on ($\Pr[H \mid H] > \Pr[L \mid H]$ and $\Pr[L \mid L] > \Pr[H \mid L]$).
However, \emph{uninformative} equilibria where agents misreport their information also exist.
Specifically, agents can coordinate to all submit the same report and each receive a maximum payment.
Nonetheless, Output Agreement remains popular because of its simplicity.
We characterize OA under our richer signal space model to help practitioners understand how behavior occurs in the real world.

\subsection{Model}
\label{sec:OA-Model}

We begin by introducing the model that, with minor variations to accommodate the form of different peer prediction mechanisms, will be used througout the paper.
There are $n$ agents; each agent $i$ receives a signal $X_i \in \reals$ representing information gained about a shared task.
The signals are drawn from a joint distribution $\D$ that is symmetric in the sense that each agent $i$ has (1) the same marginal with CDF $F(x) = \Pr[X_i \leq x] $, and (2) the same posterior distribution $\beta(x) = \Pr[X' = \cdot \mid X_i = x]$.
Identical marginals are realistic when the signals are exchangeable, e.g. when they have the same conditional distribution over some latent variable $\theta$.

Let the report space be $\R = \{L, H\}$.
Each agent $i$ calculates a report $r_i \in \mathcal{R}$ for the task according to a deterministic strategy $\sigma_i: \reals \to \mathcal{R}$ mapping from signal to report space.
The Output Agreement mechanism pays each agent $i$ an amount $M_{\OA}(r_i, r_j)$ as a function of reports $r_i, r_j$, where
\[ M_{\OA}(r_i, r_j) = \ones [r_i = r_j]. \]

Thus the (interim) expected utility for playing strategy $\sigma_i$ when $j$ plays according to $\sigma_j$ is
\begin{equation} \label{eq:expected-util-oa}
	U_i(\sigma_i, \sigma_j, x) = \E_{x' \sim \beta(x)} \ones[\sigma_i(x) = \sigma_j(x')].
\end{equation}

We consider a natural class of strategies, \emph{threshold strategies}, where agent $i$ reports $H$ if and only if their signal $x$ satisfies $x \geq \tau$ for some fixed threshold $\tau \in \reals \cup \{ \pm \infty \}$.
That is, for some $\tau$,
\[
	\sigma^{\tau}(x)=
	\begin{cases}
	L, & x \leq \tau \\
	H, & x > \tau.
	\end{cases}
\]
We will often denote a strategy directly by its threshold $\tau$ when it is clear. 
Thresholding is a natural way to assign semantic meaning to the discrete labels $H$ and $L$ according to the continuous signal space. 
Moreover, in many settings the mechanism designer themselves may announce a threshold to establish norms that they would like agents to follow.
For example, in peer grading a teacher may establish a higher threshold to encourage a high bar for good marks. 

\subsection{Equilibrium Characterization}
We are interested in characterizing threshold strategies which are \emph{symmetric  Bayes-Nash equilibria}.
That is, equilibria where each agent $i$ commits to the same threshold strategy $\sigma^{\tau}$.
For brevity, we refer to these as threshold equilibria.
Symmetric strategies are natural in our model since agents share the same ex-ante expected utilities.
\begin{definition}
	A threshold strategy $\sigma^{\tau}: \reals \to \R$ is a {\em threshold equilibrium} under OA if
\begin{equation} \label{cond:symmetric-strategy-equil-oa}
	\forall x \in \reals,\E_{x' \sim \beta(x)} \ones[\sigma^{\tau}(x) = \sigma^{\tau}(x')] \geq \E_{x' \sim \beta(x)} \ones[\overline{\sigma^{\tau}(x)} = \sigma^{\tau}(x')],
\end{equation}
where $\overline{H} = L, \overline{L} = H.$ 
\end{definition}

In the case where the signal $x$ the agent receives satisfies $x \leq \tau$, Condition~\eqref{cond:symmetric-strategy-equil-oa} simplifies to
\begin{equation*}
	\E_{x' \sim \beta(x)} \ones[x' \leq \tau] \geq \E_{x' \sim \beta(x)} \ones[x' > \tau], \mbox{ or equivalently }
	\Pr[X' \leq \tau \mid X = x] \geq 1/2.
\end{equation*}
Similarly if $x > \tau$, \eqref{cond:symmetric-strategy-equil-oa} simplifies to $\Pr[X' \leq \tau \mid X = x] \leq 1/2$.
Let $P(\tau; x) = \Pr[X' \leq \tau \mid X = x]$ be the probability another agent reports $L$ conditioned on seeing a signal $x$, for a fixed threshold $\tau$.
As a summary, then, Condition~\eqref{cond:symmetric-strategy-equil-oa} is equivalent to
\begin{align}
	\forall x \leq \tau, P(\tau; x) &\geq 1/2 \label{eq:equil-1-oa}, \\
	\forall x > \tau, P(\tau; x) &\leq 1/2 \label{eq:equil-2-oa}.
\end{align}

\paragraph{Results.}
Under these conditions, we first show that the thresholds $\tau^* = \pm \infty$ are always equilibria under Output Agreement.

\begin{proposition} \label{prop:oa-infinite-equil}
	$\tau^* = \pm \infty$ are both always threshold equilibria under OA. 
\end{proposition}

\begin{proof}
	Let $\tau^* = \infty$. 
	Then for all signals $x$ such that $x \leq \tau^*$, $P(\tau^*; x) = 1 \geq 1/2$ (since any $X' \in \reals$ satisfies $X' \leq \tau^*$ with probability one).
	Meanwhile, Statement~\ref{eq:equil-2-oa} is vacuous since no signal $x \in \reals$ satisfies $x > \infty$. 
	A similar argument follows for $\tau^* = -\infty$: for all signals $x > \tau^*$, $P(\tau^*; x) = 0 \leq 1/2$, while Statement~\ref{eq:equil-1-oa} is vacuous since no signal $x \in \reals$ satisfies $x < -\infty$.
\end{proof}

These thresholds exactly correspond to uninformative equilibria in the binary setting: if $\tau^* = \infty$, agents always report $L$, and if $\tau^* = -\infty$, agents always report $H$.
Under some smoothness conditions on the predictive posterior function, we can furthermore characterize all finite threshold equilibria.
To do so, we define the function $G(x) = P(x; x) = \Pr[X' \leq x \mid X = x]$ as the probability of another signal lying below $x$, conditioned on $x$. 
In Output Agreement, we find that threshold equilibria are characterized by thresholds $\tau$ such that $G(\tau)$ crosses 1/2.  More precisely, we give the following necessary and sufficient conditions.

\begin{theorem} \label{thm:oa-equil}
	Let finite threshold $\tau$ be given and $P(\tau; x)$ be continuous in $x$.
	If $\tau$ is a threshold equilibrium under the OA mechanism then $G(\tau) = 1/2$.
	Conversely, if $G(\tau) = 1/2$ and either (a) $P(\tau; x)$ is monotone decreasing in $x$
	or (b) $P(\tau; x) - 1/2$ has a single crossing of $0$ from positive to negative 
	then $\tau$ is a threshold equilibrium under the OA mechanism.
\end{theorem}

\begin{proof}
	For necessity, assume that $\tau$ is a threshold equilibrium, so
	Equations \eqref{eq:equil-1-oa} and \eqref{eq:equil-2-oa} hold.
	Since $P(\tau; x)$ is continuous over $x$, we must have $\lim_{x \to \tau^+} P(\tau; x) = \lim_{x \to \tau^-} P(\tau; x) = 1/2.$
	But $\lim_{x \to \tau^+} P(\tau; x) \leq 1/2$ and $\lim_{x \to \tau^+} P(\tau; x) \geq 1/2$, so we must have $P(\tau; \tau) = G(\tau) = 1/2$.

	For sufficiency, assume $G(\tau) = 1/2$.  Equivalently $P(\tau; \tau) - 1/2 = 0$, so (a) implies (b) and the single crossing is at $\tau$.  Thus assume (b) holds and
	take some $x > \tau$. 
  	By single crossing, $P(\tau; x) \leq P(\tau; \tau) = G(\tau) = 1/2$.
  	Similarly, if $x \leq \tau$, $P(\tau; x) \geq P(\tau; \tau) = 1/2$.
  	This establishes Equations \eqref{eq:equil-1-oa} and \eqref{eq:equil-2-oa}, so $\tau$ is a threshold equilibrum.
\end{proof}

The stronger monotonicity condition holds when a higher signal consistently increases the probability that a peer will receive a signal above the threshold.
Meanwhile, the single-crossing condition allows for local belief fluctuations: a slightly higher signal might decrease confidence that a peer will report `H', but there is no ambiguity regarding the optimal action. 
For example, in peer grading an agent may read an essay which is well-written but uses some amount of em-dashes: this may slightly \emph{decrease} their confidence that someone else would give it a high mark because there is some chance it was generated by AI. 

To better understand the conditions, we also observe that the derivative of expected utility with respect to agent $i$'s choice of threshold $\tau$ in Equation~\eqref{eq:expected-util-oa} is $f(x)(2G(x)-1)$, where $f$ is the PDF of $F$ and therefore non-negative.  
Thus the necessary condition corresponds to the first order condition of optimizing the best response threshold and the sufficient conditions ensure this optimization is concave and quasiconcave respectively.  
However, the proof provided establishes optimality in the space of all possible strategies (which is required to be a threshold equilibrium), not merely that the chosen threshold is the optimal threshold.

More broadly, this result shows that ``truthfulness'' in the sense of following the strategy prescribed by the mechanism designer is fragile in this setting.  Only a sharply limited set of thresholds can be implemented in equilibrium---in some cases a single one.  This means the mechanism designer's ability to choose the ``meaning'' of $H$ and $L$ can be quite limited in practice.

\subsection{Dynamics}
We have characterizations of both the uninformative infinite equilibria and more interesting finite equilibria.
Which are more likely to occur?
To answer this question, we turn to modeling the dynamics of the system.
We are interested in studying how agents will behave under our model over time, as a tool for equilibrium refinement.
That is, we want to know what threshold agents will land on as they continue to grade new batches of papers or submit new labels. 
Since it is likely that the mechanism shares empirical results before all agents equilibriate, we consider a dynamic setup where a small fraction of agents are able to best respond at each time step. 

Formally, consider the following discrete best response dynamic: the mechanism publishes an initial threshold $\tau(0)$, the ``norms" with which agents should interpret their signals as $H$ and $L$.
At each time step, a small, random fraction $\Delta t$ of agents are able to commit to a new best response threshold strategy $\hat{\tau}(t)$ against the current threshold $\tau(t)$.
Meanwhile, the rest of the agents stick to the threshold they used previously, whether that be a previous best response or the initial $\tau(0)$.
Note that since $\hat{\tau}(t)$ is a best response, its calculation depends on (1) the peer prediction mechanism's payment scheme and (2) the current threshold $\tau(t)$.
Finally, the empirical frequency of reports is updated to reflect this new mixture of thresholds, and the process repeats.
This dynamic corresponds to the following equation:
\begin{align}
	\tau(\Delta t) &= (1 - \Delta t)\tau(0) + \Delta t \hat{\tau}(0) \nonumber \\ 
	\tau(2 \Delta t) &= (1 - \Delta t) \tau(\Delta t) + \Delta t \hat{\tau}(\Delta t) \nonumber \\
	&\ldots \nonumber \\
	\tau(t + \Delta t) &= (\Delta t) \hat{\tau}(t) + (1 - \Delta t) \tau(t). \label{eq:discrete-dynamic}
\end{align}

We note that $\hat{\tau}(k \Delta t)$ is technically a best response to the \emph{mixture} of thresholds at the previous time step $k-1$, rather than the probability distribution over thresholds.  Ignoring this for the moment and taking the limit as $\Delta t \to 0$ of the discrete system in Equation \ref{eq:discrete-dynamic}, we end up with the following continuous best response dynamic:
\begin{equation} \label{eq:continuous-dynamic}
	\dot \tau = \hat{\tau} - \tau.
\end{equation}
While not, strictly speaking, the correct update from best response dynamics, this update rule has the appealing form of the threshold slowly drifting toward the best response.
Alternatively, since agents near the threshold have the strongest incentive to change their strategy, this update rule could be intepreted as a form of quantal response.
In any case, our focus in this work is on what sort of equilibria we should expect to end up at rather than the exact process the system takes to get there, so this choice of dynamics seems simple, natural, and amenable to analysis.

We begin our study of OA under the dynamics described by Equation \eqref{eq:continuous-dynamic} with the following observation.
A nice property of Output Agreement (which relates to our sufficient conditions from Theorem~\ref{thm:oa-equil}) is that for a fixed threshold $\tau$, under decreasing monotonicity of $P(\tau;x)$ the best response over \emph{all} strategies is always a threshold even out of equilibrium.

\begin{proposition}
	Assume for a fixed threshold $\tau$ that $P(\tau; x)$ is strictly decreasing and continuous over $x$.
	If all agents are playing according threshold strategy $\tau$, the unique best response of an agent across all strategies $\sigma: \reals \to \R$ is to play according to threshold strategy $\hat \tau$ satisfying  $P(\tau; \hat \tau) = 1/2$.
\end{proposition}

\begin{proof}
	If $\tau$ is the current threshold strategy that all other agents are following, an agent who receives signal $x$ will best respond with $H$ if $\Pr[X \leq \tau \mid X = x] \leq 1/2$, and $L$ otherwise. 
	Since $P(\tau; x)$ is strictly decreasing and continuous over $x$, there will be a unique point $\hat \tau$ such that $P(\tau; \hat \tau) = 1/2$, with $P(\tau; x) \geq 1/2$ for $x \leq \hat \tau$ and $P(\tau; x) \leq 1/2$ for $x > \hat \tau$.
	This threshold $\hat \tau$ thus corresponds to the best response strategy.
\end{proof}

We can now use concepts from dynamical systems and characterize equilibria as {\em stable} or {\em unstable} according to properties of the function $G$.  This distinction captures the behavior of the dynamics near the equilibrium: do they drive the system toward the equilibrium (stable) or away from it (unstable)?  More precisely, an equilibrium is (locally) stable if a sufficiently small perturbation yields dynamics that always return to it while it is unstable if all perturbations diverge from it.

We argue that only stable equilibria are reasonable to find in practice.  Otherwise any small error in the choice of $\tau^*$ will fail to yield the desired equilibrium as the dynamics push away from it.  Since the equilibria depend on the  joint distribution $\D$, which will typically not be perfectly known to the mechanism designer, such small errors are to be expected.  As a result, we view stability as an important equilibrium refinement in the context of peer prediction.

\begin{theorem} \label{thm:oa-dynamics}
	Assume for all $\tau \in \reals$ that $P(\tau; x)$ is strictly decreasing and continuous over $x$.
	Then if $G(\tau)$ is strictly increasing at equilibrium point $\tau^*$, $\tau^*$ is unstable,
	while if it is instead strictly decreasing $\tau^*$ is stable.
\end{theorem}

\begin{proof}
	Consider the dynamics at time step $t$, with the current threshold $\tau(t)$.
	Since we are considering a fixed time $t$, we refer to $\tau(t)$ as $\tau$ throughout the proof.

	We consider what happens when an agent receives a signal exactly at $\tau$.
	There are three cases to consider.
	In the first, $G(\tau) = 1/2$; then the expected utility of reporting $L$ is $\Pr[X' \leq \tau \mid X = \tau] = 1/2$, while the expected utility of reporting $H$ is $\Pr[X' > \tau \mid X = \tau] = 1/2$; thus the agent is indifferent between reporting $L$ or $H$.
	Moreover, the best response $\hat \tau$ is exactly $\tau$. 
	It follows the system is at an equilibrium of the dynamics since $\dot \tau = \hat\tau - \tau = 0$.

	In the second case, $G(\tau) > 1/2$. 
	Then reporting $L$ is strictly preferred to reporting $H$.
	Since $P(\tau; \hat \tau) = 1/2$, by decreasing monotonicity we have $\hat{\tau} > \tau$. 
	Thus $\dot \tau = \hat\tau - \tau > 0$.
	In the third case, $G(\tau) < 1/2$. 
	Then reporting $H$ is strictly preferred to reporting $L$.
	Since $P(\tau; \hat \tau) = 1/2$, by decreasing monotonicity we have $\tau > \hat{\tau}$. 
	Thus $\dot \tau = \hat\tau - \tau < 0$.

	If $G(\tau)$ is strictly increasing at $\tau^*$, then at any perturbed point $\tau$ to the left of $\tau^*$ we have $G(\tau) < 1/2$ and $\dot \tau < 0$; and at any perturbed point $\tau$ to the right of $\tau^*$ we have $G(\tau) > 1/2$ and $\dot \tau > 0$.
	It follows that $\tau^*$ is unstable.  The same logic implies stability in the strictly decreasing case.
\end{proof}

In most reasonable settings, one would expect $G(x)$ to limit toward 1 as $x$ approaches $\infty$, and toward 0 as $x$ approaches $-\infty$. 
Specifically, as a signal $x$ becomes increasingly small, the probability another agent receives a signal smaller than $x$ should approach 0.
Meanwhile, as a signal $x$ grows large, the probability another agent receives a signal smaller than $x$ should approach 1. 
As the quality of essay increases in a normally distributed class, for example, a grader would believe that the mass of essays below that quality should increase proportionally. 

We show formally in \S~\ref{appendix:oa} that under such limit behavior of $G$, uninformative equilibria at $\tau^* = \pm \infty$ are \emph{stable} in OA, while at least one finite equilibrium exists which is unstable. 
In fact, one would expect in smooth, unimodal settings, as in \S~\ref{subsec:gaussian-intro}, that $G$ is monotone increasing. 
In this case, there are three equilibria, and by Theorem~\ref{thm:oa-dynamics}, the internal equilibrium is unstable while the uninformative equilibria are stable. 

So what does this instability of internal equilibria mean for OA?
When signals correspond monotonically to quality of an essay or task, agents will naturally move toward uninformative equilibria. 
These dynamics only further the fragility of truthfulness in OA: unless the threshold $\tau$ exactly balances out the conditional probabilities that other agents will report $H$ or $L$, agents will inevitably and increasingly misreport until they reach an uninformative consensus. 
There are multimodal Bayesian examples where several internal equilibria exist and some are thus topologically required to be locally stable (see \S~\ref{sec:experiments}).
However, as long as $G$ remains monotone increasing at extreme signals, we expect uninformative equilibria to be stable.

\begin{figure}[t]
	\centering
	\begin{minipage}[t]{0.4\textwidth}
	  \centering
	  \includegraphics[width=1\columnwidth]{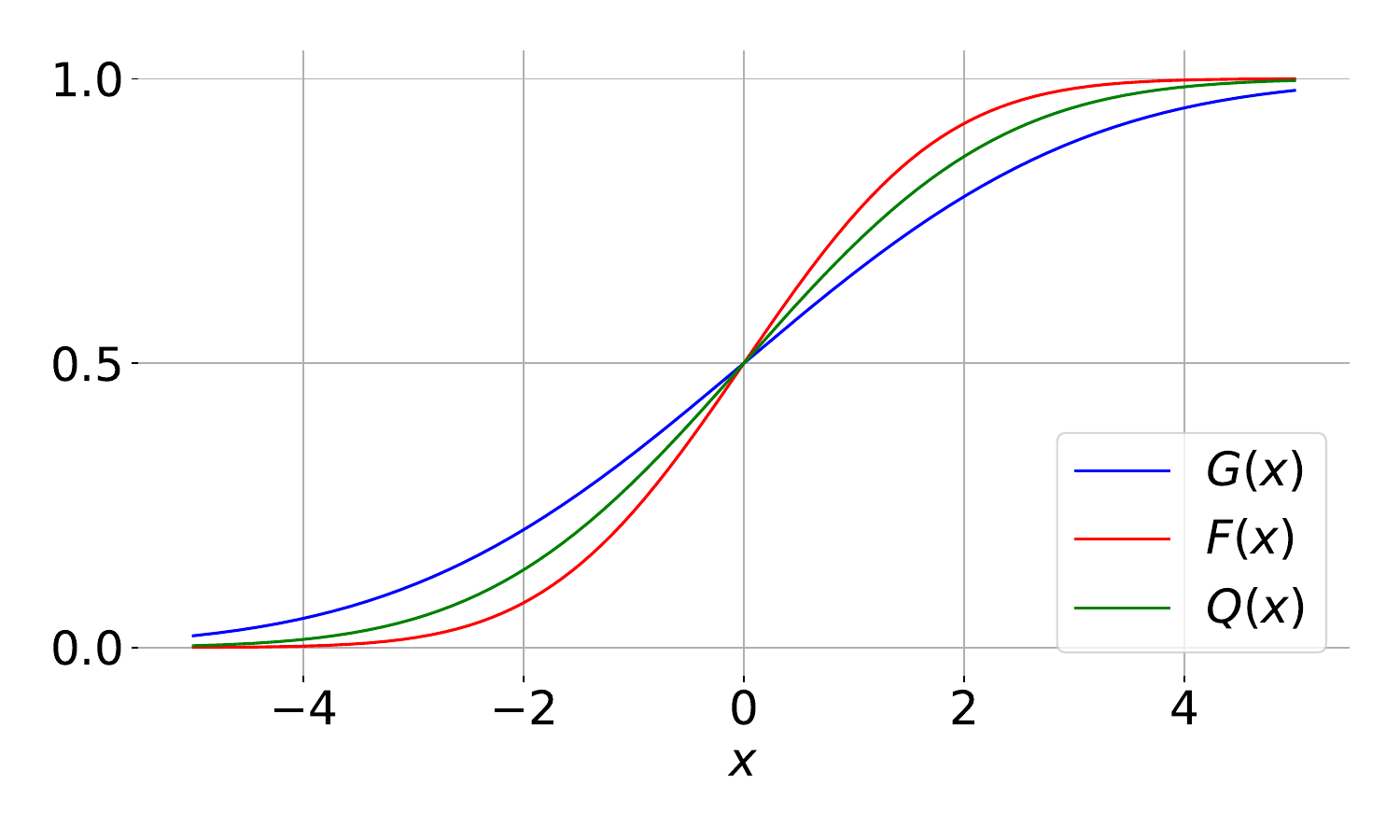}
		\caption{Plots of $F$, $G$, and $Q$ in the Gaussian model (\S~\ref{subsec:gaussian-intro}) for the parameters $a = b = 1$.}
		\label{fig:gaussian-F-G}
	\end{minipage}%
	\hspace{3em}
	\begin{minipage}[t]{0.4\textwidth}
		\centering
		\includegraphics[width=1\columnwidth]{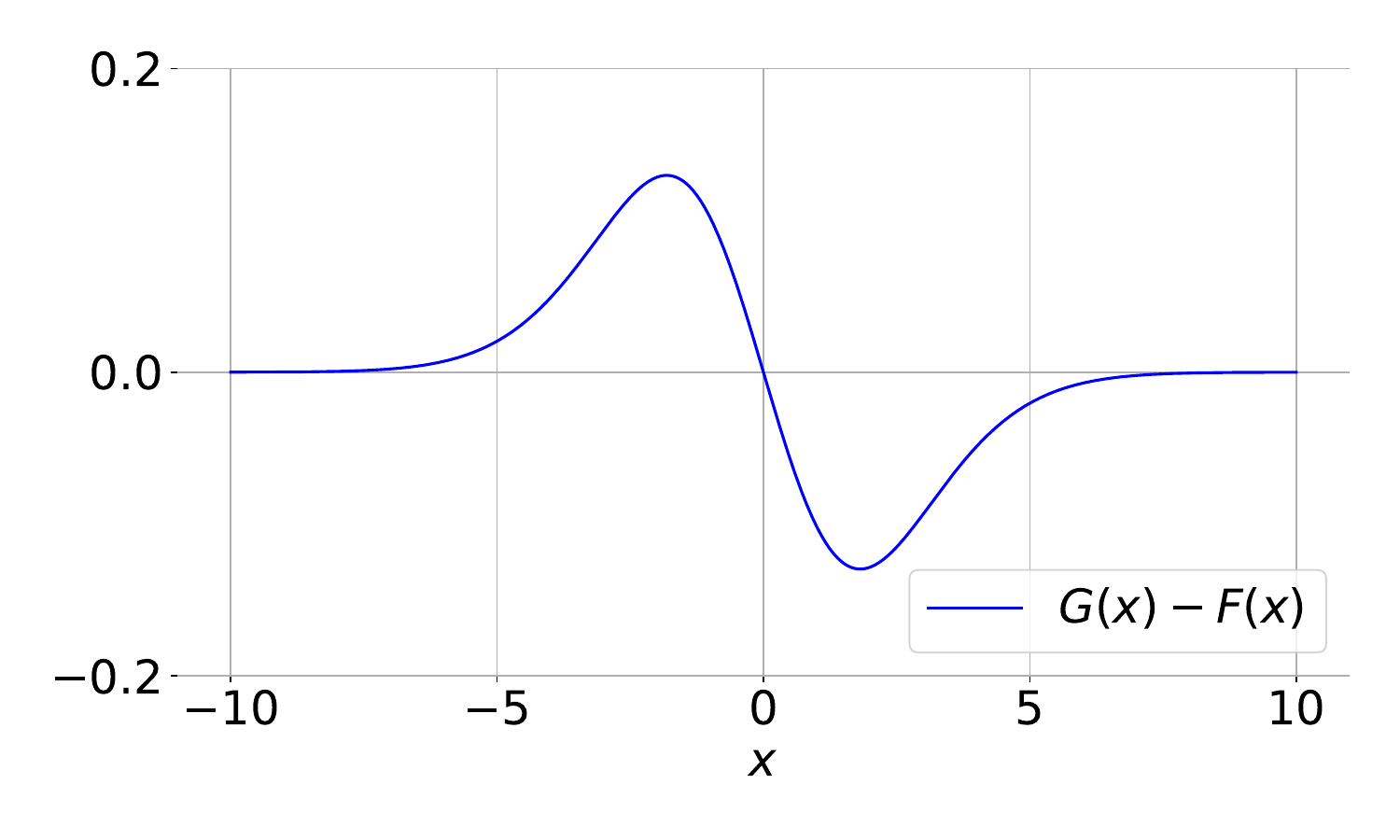}
		\caption{Plot of $G-F$ in the Gaussian model (\S~\ref{subsec:gaussian-intro}), with $a = b = 1$.
		  $G(x)-F(x)$ is decreasing at $0$, meaning it a stable equilibrium.}
		\label{fig:gaussian-F-G-difference}
	\end{minipage}
	\vspace{-\baselineskip}
\end{figure}

\subsection{A Gaussian Model} \label{subsec:gaussian-intro}
Suppose each agent $i$ receives a noisy signal $X_i$ from a Gaussian distribution, with the noise also normally distributed. 
That is, $X_i = a Z + b Z_i$, where $Z \sim N(0, 1)$, $Z_i \sim N(0, 1)$, and $a, b \in \reals_{> 0}$.
This model e.g. realistically captures peer grading settings where there is a more fine-grained evaluation of essays anchored in a standard bell-curve distribution. 
It follows that the marginal distribution for each agent is $F(x) = N(x \mid 0, a^2 + b^2)$.
This model allows for a continuum of agents with the same $F$, so the setting is consistent with our dynamics. 

The correlation coefficient for two signals $X, X'$ is $\rho = \frac{a^2}{a^2 + b^2}$.
Fixing any agent, the predictive posterior distribution upon seing $x$ is $\Pr[X' = x' \mid X = x] = N(x' \mid \hat \mu(x), \hat\sigma^2)$, where $\hat\mu(x) = \rho x$ and $\hat\sigma^2 = b^2 (1 + \rho).$
Denoting the CDF of a standard Normal distribution by $\Phi$, it follows that for a fixed $\tau$,
\begin{equation} \label{eq:p}
	P(\tau; x) = \Phi\left( \frac {\tau - \rho x} {b \sqrt{1+\rho}} \right).
\end{equation}
We can now also write
\begin{equation}
	\label{eq:gaussian-FG}
	F(x)
	= \Phi\left( \frac x {\sqrt{a^2 + b^2}} \right)
	= \Phi\left( \frac {\sqrt{\rho}} {a} \; x \right)
	\mbox{ and }
	G(x)
	 = \Phi\left( \frac {x - \rho x} {b \sqrt{1+\rho}} \right)
	= \Phi\left( \frac {(1-\rho)} {b \sqrt{1+\rho}} \; x \right)~.
\end{equation}

We then immediately observe existence of equilibria according to our general results:
\begin{corollary}
	In the Gaussian model under OA, we have three equilibria at $\tau = 0$ and $\pm \infty$.
\end{corollary}

\begin{proof}
	Existence of equilibria at $\tau = \pm \infty$ follows from Proposition~\ref{prop:oa-infinite-equil}.
	Note that $P(\tau; x)$ is strictly monotone decreasing and continuous in $x$ so that Theorem \ref{thm:oa-equil} applies. 
	As $G(0) = \Phi(0) = \tfrac 1 2$, $\tau=0$ is an equilibrium for OA.
	Note also that $\Phi$ is strictly monotone, so $G(x)$ only crosses 1/2 once and $\tau = 0$ is therefore the \emph{only} finite equilibrium (see a visualization in Figure~\ref{fig:gaussian-F-G}).
\end{proof}

Now, consider the stability of these equilibria. 

\begin{corollary}
	In the Gaussian model under OA, the threshold equilibrium $\tau = 0$ is unstable, while the uninformative equilibria $\pm \infty$ are stable.
\end{corollary}

\begin{proof}
	Let $c_G(\rho) = \frac {(1-\rho)} {b \sqrt{1+\rho}}$ be the coefficient of $x$ in $G(x)$. 
	Then $G'(x) = c_G(\rho) \phi(c_G x)$ for $\phi$ the PDF of the standard Normal, and so $G'(0) = c_G(\rho) \phi(0) > 0$. 
	By Theorem~\ref{thm:oa-dynamics}, then, $\tau = 0$ is unstable. 
	Since $0$ is the only finite equilibrium, the equilibria $\tau = \pm \infty$ are stable by topological necessity.
\end{proof}

For any finite starting point $\tau(0) \neq 0$, then, we have $\lim_{t \to \infty} \tau(t) = \pm \infty$.
Therefore, we find in settings where graders receive a noisy, normally distributed version of information about an essay, Output Agreement incentivizes agents to stabilize at an equilibrium where they all submit the same report.
This uninformative consensus is in contrast to the truthful equilibrium guarantees of the binary signal model.\footnote{In the Gaussian model, strong diagonalization holds within a reasonable range of thresholds around 0, depending on $\rho$.}

\section{Other Binary-Report Mechanisms}

In this section, we characterize equilibria and their dynamics under the same real-valued signal model for three other popular binary-report peer prediction mechanisms.

\subsection{Dasgupta-Ghosh}
We consider the mechanism proposed by~\citet{dasgupta2013crowdsourced}, which we will refer to as the Dasgupta-Ghosh (DG) mechanism.
DG is a multi-task mechanism: agents participate in a set of peer prediction tasks, and their payment is a function of their reports across these tasks.
For a given task, agent $i$ receives a payment of one if their report $r_i$ matches that of another agent $j$ on that task, as in Output Agreement.
However, they are also \emph{penalized} if their report matches that of the other agent on some unrelated task agent $i$ did not complete, with the intution that agents should be paid only for their performance over and above what would be expected by chance.\footnote{There have been a variety of ways considered to calculate this penalty.  This version corresponds to their $d=1$ option.}
In the binary signal setting, this penalty ensures that the uninformative equilibria in OA where all agents agree to report H or L now have zero value, so that strategic agents are not incentivized toward them.

\paragraph{Model.}
Formally, if $r_i$ agent $i$'s report on the task, $r_j$ is the peer report on the task, and $r_k$ is the peer report on a distinct task, 
\[ M_{\DG}(r_i, r_j, r_k) = \ones [r_i = r_j] - \ones [r_i = r_k]. \]
Assume the task $r_k$ comes from is i.i.d.~to the task being scored and, as with our analysis of OA, all agents are using symmetric strategies.  Then agent $i$'s expected penalty is $\E \ones [r_i = \sigma(X)]$.  That is, the expected penalty is exactly the prior probability of a signal that leads to a report of $r_i$.  We denote this as 
\[\pi_L = \Pr[\sigma(X) = L] = \Pr[X \leq \tau] = F(\tau),\]
with $\pi_H = 1 - \pi_L$.  
Therefore, we instead work with the following formulation of DG which has the same expected payment:
\begin{equation} \label{eq:dg-payment}
	M_{\DG}(r_i, r_j, \pi_L) = \ones [r_i = r_j] - \pi_{r_i}.
\end{equation}
Despite the nature of DG as a multi-task mechanism, this means that (under our assumption of symmetry and if other agents use the same strategy across tasks) we can essentially analyze each task in isolation.  More formally, this means we can adopt the same model from \S~\ref{sec:OA-Model} we used to analyze OA, changing only the payment rule and updating the quantities that derive from it accordingly.

\paragraph{Equilibrium characterization and dynamics.}
We leave formal derivations of equilibrium characterization and dynamics to \S~\ref{appendix:dg}, because they follow in a similar way to our OA analysis.
Importantly, while the uninformative equilibria still exist, the penalty term induces a different characterization of finite equilibria: under the same continuous and single-crossing assumptions over the predictive posterior, $\tau^*$ is an equilibrium if and only if $G(\tau^*) = F(\tau^*)$ (where as a reminder, $F(\tau^*) = \Pr[X \leq \tau^*])$.
Moreover, $\tau^*$ is stable if and only if $G(\tau) - F(\tau)$ is strictly increasing at $\tau^*$.

The difference between $G$ in the OA case and $G-F$ here is important.
We saw that $G$ was strictly increasing in natural settings like the Gaussian model.  In contrast, in similar settings $G - F$ is strictly decreasing at a unique finite equilibrium point.
We find under the Gaussian model, in fact, that the same equilibria exist at $(0, \pm \infty)$; but $0$ is stable, while the uninformative equilibria are unstable (see Figures~\ref{fig:gaussian-F-G} and~\ref{fig:gaussian-F-G-difference}).
Since $\tau=0$ is the only stable equilibrium, it is in fact globally attracting: for any finite starting threshold $\tau(0)$, we have $\lim_{t \to \infty} \tau(t) = 0$.  

In contrast to OA, this means that the behavior of DG in the Gaussian case is quite robust: unless the initial threshold starts exactly at an uninformative equilibrium it inevitably converges toward an informative one where $H$ and $L$ each get reported half the time.  
In this sense, DG behaves in a way one might hope based on its traditional analysis: it discourages uninformative equilibria.  
While not all settings will behave as nicely as our Gaussian example, under reasonable behavior in the tails we expect existence of a stable nontrivial equilibria, while the uninformative equilibria remain unstable. 
We discuss such generalizations in \S~\ref{appendix:dg}.

While robust, the stable equilibrium in DG under the Gaussian model is also inflexible: the mechanism designer cannot specify a desired threshold, and must settle for $\tau=0$.
Thus the semantics of $H$ and $L$ relative to the underlying signal are predetermined, as half of the tasks must be classified each way.
In peer grading, for example, half the assignments would be labeled unsatisfactory, which might not be the desired outcome.\footnote{See \S~\ref{sec:discussion} for futher discussion.}
We expect this general phenomenon, that there is a unique equilibrium threshold for DG which is determined by the underlying signal distribution, extends beyond the Gaussian case; see \S~\ref{sec:experiments}.

\subsection{Robust Bayesian Truth Serum}
The Robust Bayesian Truth Serum (RBTS) mechanism~\citep{witkowski2012robust}, inspired by the Bayesian Truth Serum (BTS)~\citep{prelec2004bayesian}, is a non-minimal mechanism: each agent $i$ submits both an information report $r_i \in \{L, H\}$ and also a \emph{prediction} report $p_i \in [0, 1]$.
The prediction report represents agent $i$'s belief about the frequency of $H$ reports from all agents.
The mechanism incentivizes prediction reports using the \emph{Brier score} $S(p, r) = 2 p \ones[r = H] - p^2$ (also known as the quadratic score), which is an example of a strictly proper scoring rule~\citep{brier1950verification,gneiting2007strictly}.
RBTS randomly picks a reference agent $j$ and peer agent $k$, and pays agent $i$
\begin{equation} \label{eq:rbts-payment}
	M_{\RBTS}((r_i, p_i), (r_j, p_j), (r_k, p_k)) = S(p_j + \delta t(r_i), r_k) + S(p_i, r_k),
\end{equation}
where $t(L) = -1$ and $t(H) = 1$, and $\delta>0$ is chosen by the mechanism.

We leave our formal analysis of RBTS to \S~\ref{app:rbts}, where we (perhaps surprisingly) are able to characterize equilibria under similar monotonicity and continuity conditions as OA and DG.
We find equilibria coincide with all thresholds that satisfy $Q(\tau) = G(\tau)$, where $Q(x) = \E_{x' \sim \beta(x)} P(x; x')$. That is, $Q$ represents the expected posterior of the agent's peer at signal $x$, conditional on $x$.
We then show that the location and dynamics of equilibria in RBTS under real-valued signals match DG for the Gaussian model (though convergence to stable equilibria is faster under RBTS), meaning the same inflexibility is inherited. 

\subsection{Determinant-based Mutual Information (DMI) Mechanism}

The Determinant-based Mutual Information (DMI) mechanism was introduced by \citet{kong2024dominantly,kong2020dominantly}.
It is a minimal multi-task mechanism with several strong guarantees, such as being dominantly truthful under consistent strategies.
In the case of binary reports, the simplest version of the DMI mechanism collects 4 reports from a pair of agents, and computes a score based on a measure of their mutual information.
To construct an unbiased estimator for this mutual information measure, the pairs of reports are split into two groups, and the determinants of the count matrices associated to each group are multiplied together.

Formally, the mechanism is given as follows.
Define $\ones_H = (1,0)$ and $\ones_L = (0,1)$ to be indicator vectors in $\reals^2$.
Then $M_\DMI : \{L,H\}^4 \times \{L,H\}^4 \to \reals$ is the function
\begin{align*}
  M_\DMI(r_1,\ldots,r_4,r_1',\ldots,r_4') = \det M_{12} \det M_{34}~,
\end{align*}
where $M_{ij} = \ones_{r_i} \ones_{r_i'}^T + \ones_{r_j} \ones_{r_j'}^T$ is the count matrix for tasks $i$ and $j$.

In \S~\ref{app:dmi}, under some constraints on the form of strategies over each of the 4 task signals, we show that DMI has the same necessary condition for equilibrium as DG.
So, modulo a slightly different sufficient condition, its equilibria can be found at the same thresholds: those where $G(\tau) = F(\tau)$. 
In particular, our result about the dynamics (Theorem~\ref{thm:dg-dynamics}) immediately applies with the appropriate sufficient condition, and the Gaussian case has the same equilibra with the same stability.

\section{Distributions beyond the Gaussian} \label{sec:experiments}
We now aim to numerically explore existence and dynamics of equilibria with real-valued signals across mechanisms in settings other than the noisy Gaussian model.
We consider two natural departures from the ``nice" properties of Gaussians, skewness and multimodality.

\subsection{Skewness}
We extend our observations to asymmetric signal distributions by adapting the same noisy signal model as our running Gaussian example, but with a base distribution which is skewed.
Formally, each agent $i$ receives a signal $X_i$ such that $X_i = Z + Z_i$, where $Z$ is the skewed Normal distribution with mean 0 and variance 1, and $Z_i \sim N(0, 1)$. 
$Z$ has another parameter, $\alpha$, which controls the skewness of the distribution.
The magnitude of $\alpha$ increases the magnitude of skewness; moreover, $\alpha < 0$ leads to a left-skewed distribution, and $\alpha > 0$ a right-skewed distribution.
We study how threshold equilibria move as a function of the parameter $\alpha$. 

As depicted in Figure~\ref{fig:bifurcation-skewed}, we find a unique finite equilibrium under OA, DG (and DMI, by our theoretical results), and RBTS which shifts toward the longer tail, i.e. the thresholds move from negative to positive as the skewness moves from left to right.
We recover the equilibrium at $\tau = 0$ in the original Gaussian model when $\alpha = 0$.
Thus we expect when there is asymmetry in the distribution of signals, agents in DG, DMI, and RBTS will settle at a threshold that balances out the asymmetry in signal mass.
While we observe the same high-level dynamics for OA, these equilibria are unstable and so we still expect agents will move toward an uninformative consensus. 

\subsection{Multimodality}
Next we explore the effects that a distribution with multiple peaks has on the equilibria that occur under our signal model.
Consider a classic mixed Gaussian setting where each agent's signal $X_i$ is received i.i.d. from one of $K$ Gaussian components, each with mean $\mu_k$ and variance $\sigma_k^2$.
There is an underlying parameter $z \sim \text{Categorical}(\pi)$ for $\pi \in \Delta^K$ that determines which Gaussian component all the signals come from. 
For the purposes of our experiments, we fix $\pi$ to be the uniform distribution over the $K$ components. 
We numerically calculate the functions $F$, $G$, and $Q$, and check for the sufficient and necessary conditions of our theorems, to find equilibria and their stability.
We set the precision of each component to be the same, and then vary this precision parameter to generate bifurcation graphs for $K = 2$ and $3$ (see \S~\ref{appendix:experiments} for $K = 4$).
\begin{figure}[t]
    \centering
    \begin{minipage}[t]{0.35\textwidth}
        \centering
		\includegraphics[width=\textwidth]{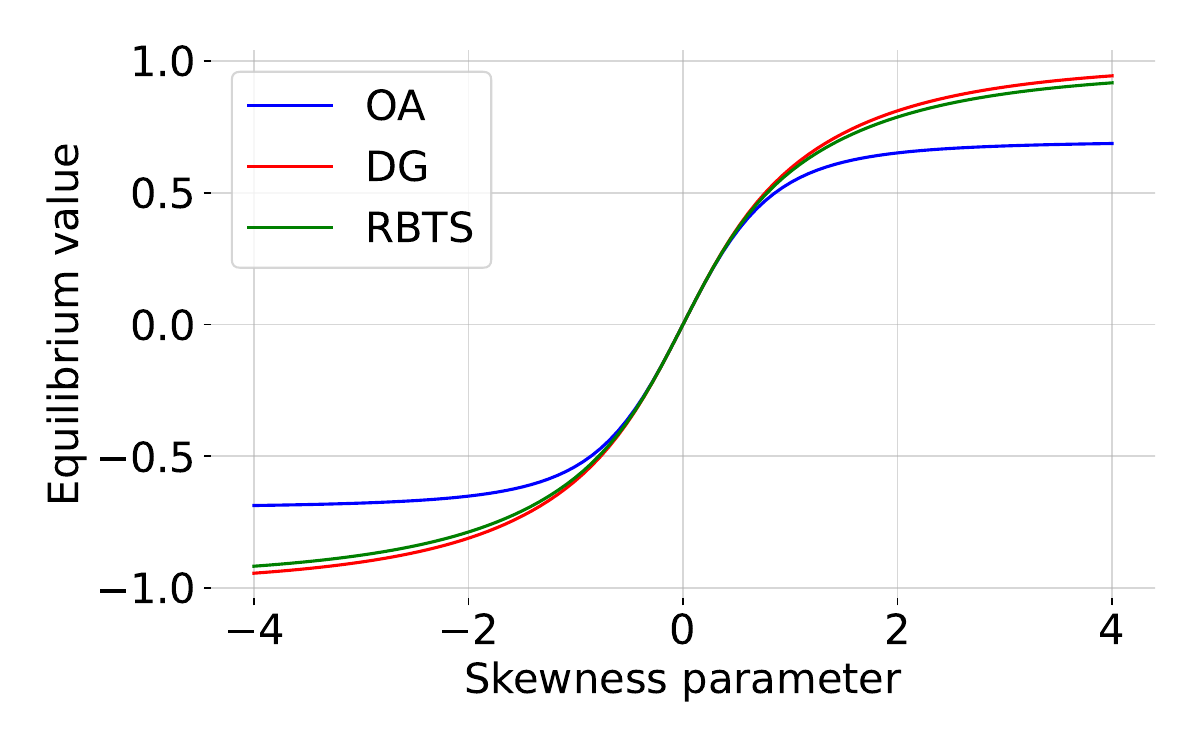}
		\caption{Equilibria in OA and DG when signals are noisy versions of a skewed Normal, across different values of the skewness parameter $\alpha$. }
		\label{fig:bifurcation-skewed}
    \end{minipage}
    \hspace{3em}
    \begin{minipage}[t]{0.35\textwidth}
        \centering
		\includegraphics[width=0.9\textwidth]{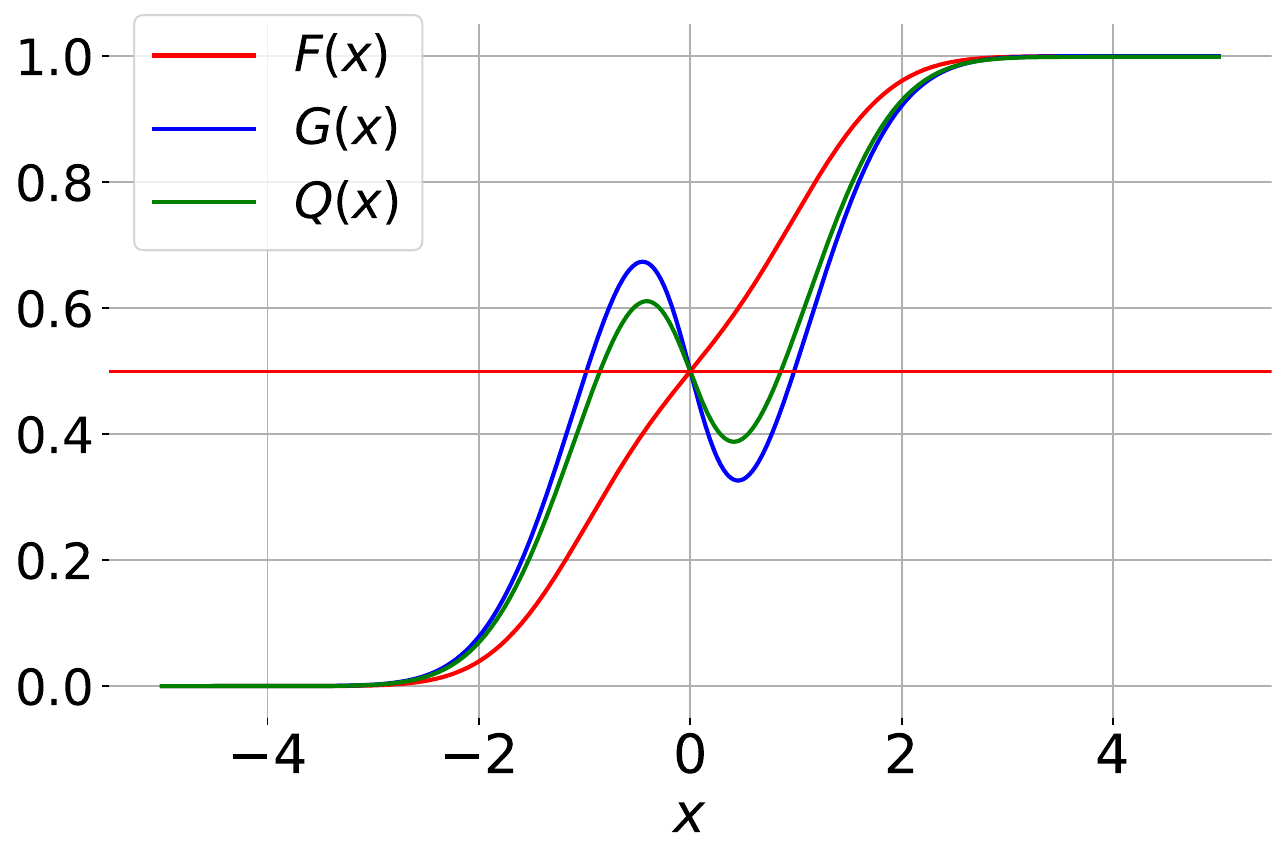}
		\caption{$F$, $G$, and $Q$ over signals $x$ in the two-component Gaussian mixture setting with means $-1, 1$ and precision 2.
		}
		\label{fig:snapshots}
    \end{minipage}
    \begin{subfigure}[t]{0.35\textwidth}
        \includegraphics[width=\textwidth]{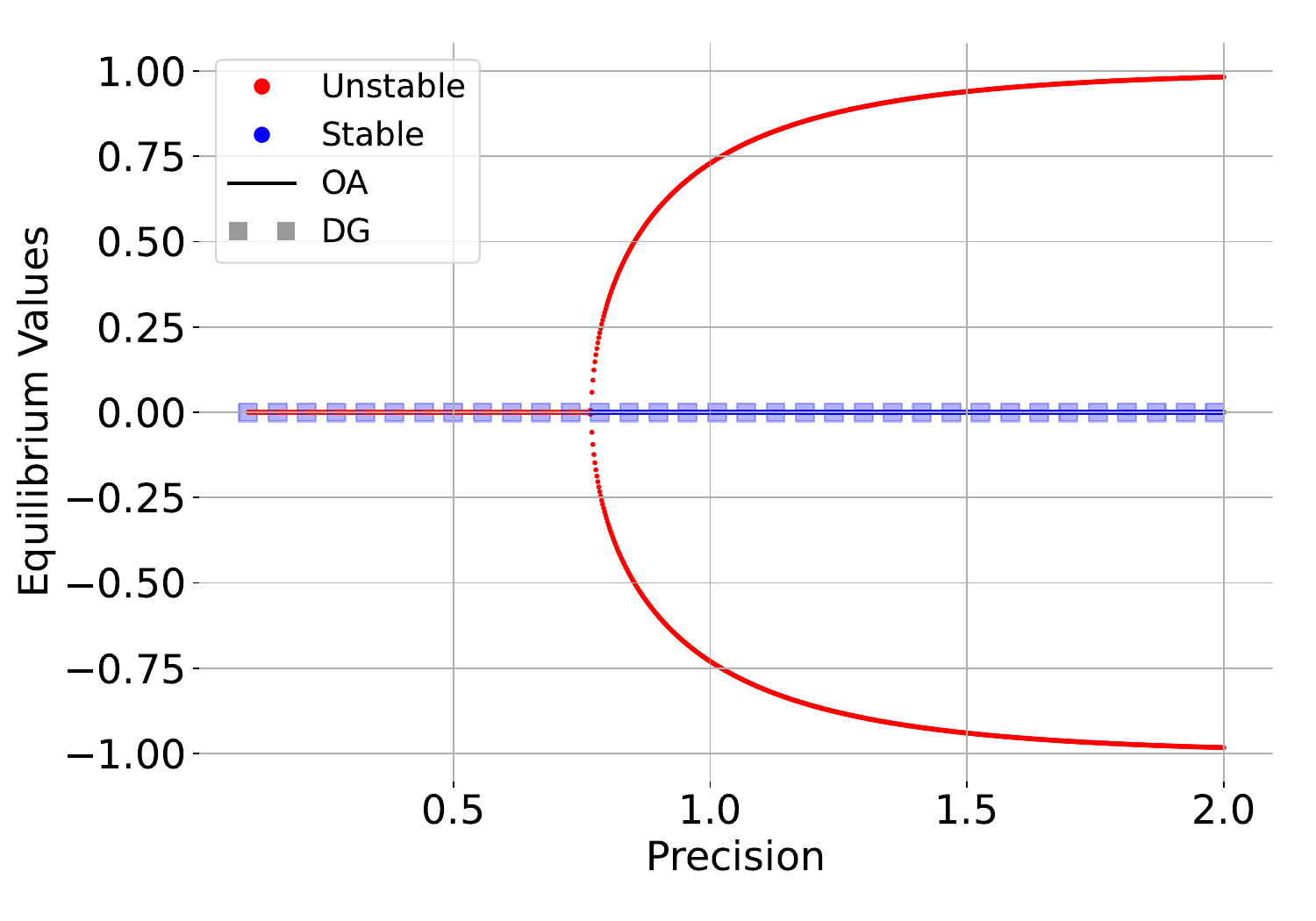}
    \end{subfigure}
	\hspace{3em}
    \begin{subfigure}[t]{0.35\textwidth}
        \includegraphics[width=\textwidth]{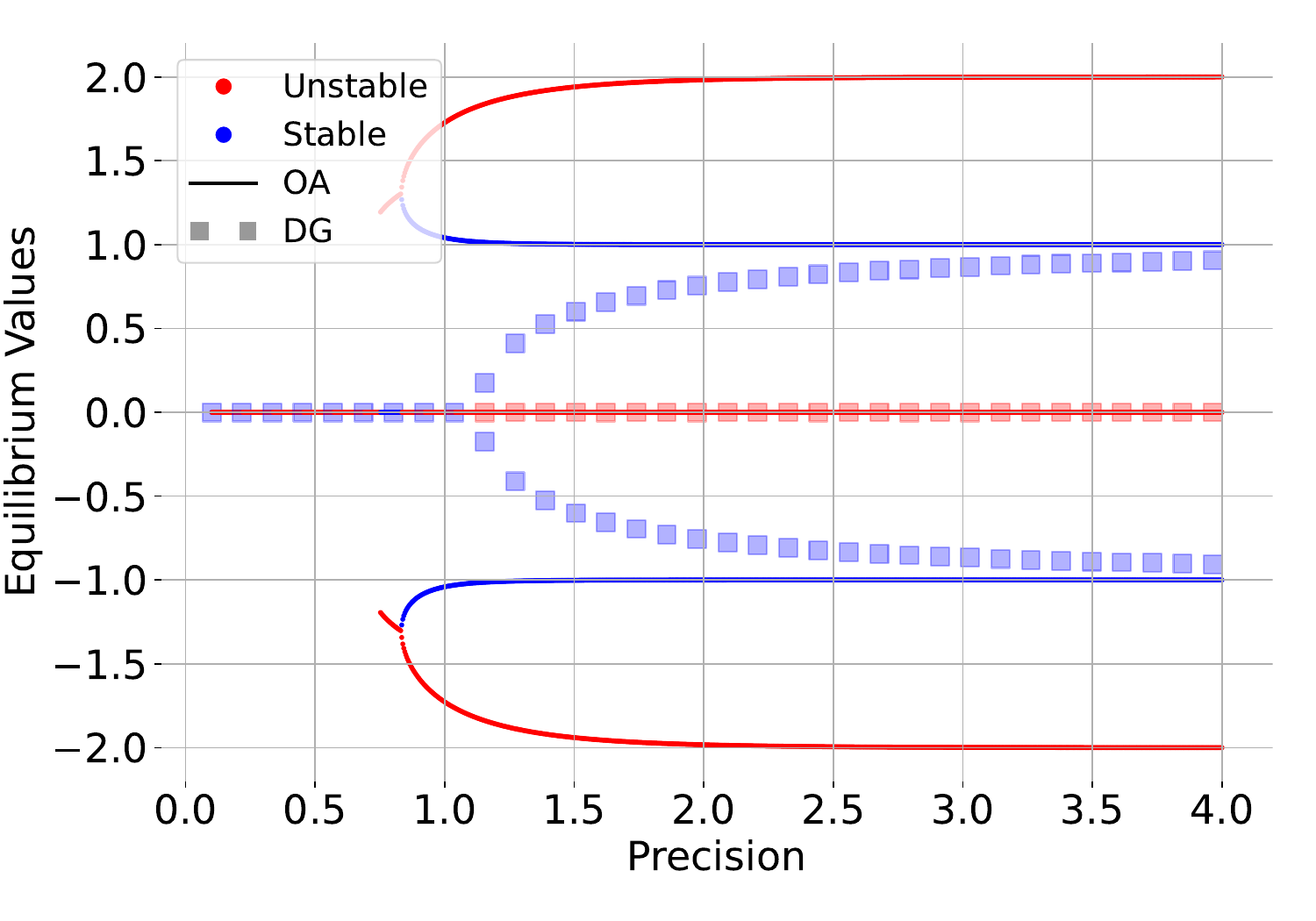}
    \end{subfigure}
    \caption{Bifurcation diagrams for two- and three-component Gaussian mixtures. 
    (Left) Two-component case, with means $-1,1$: OA (solid) exhibits three equilibria at high precision, while DG (dashed) maintains a single stable equilibrium at $\tau=0$. 
    (Right) Three-component case with means $-2,0,2$: at high precision, OA (solid) has five equilibria, with two stable at $\tau=\pm1$; DG (dashed) has three equilibria, with stability reversed from the two-component OA case.}
    \label{fig:bifurcation-summary}
\end{figure}

We first observe that the number of equilibria increases by two each time we add a Gaussian component (Figure~\ref{fig:bifurcation-summary}) for both OA and DG.
OA begins with three equilibria for large enough precision under the two-component model, with $\tau = 0$ emerging as a nontrivial stable point.
In general, for high enough separation between the two modes, stable equilibria exist beyond trivial agreement, and the area of starting conditions that lead to these equilibria increases.
One can view this phenomenon as a reflection of the original binary signal model: with enough separation between the two modes of the richer signal space, we essentially return to a setting where Output Agreement is useful because the ``signal'' is effectively which mode has been chosen.
The pattern continues more generally for $K$ components, with stable threshold equilibria emerging between pairs of modes for high enough precision values. 

For DG, $\tau = 0$ remains the unique stable equilibrium for two components, while two stable equilibria emerge as the precision increases at $\pm 1$, between the distribution's modes, for three components.
It follows that if signal noise is sufficiently small and the underlying distribution is multimodal (with at least three modes), the designer has more choice in how agents map information to reports. 
E.g., in peer grading, if the quality distribution has three modes, the designer can choose between thresholds separating low and high marks at the first or second mode.
On the other hand, the designer does not have control or knowledge of the underlying signal distribution. 

\paragraph{RBTS and DMI.}
We view equilibria in the snapshot of a specific precision value in Figure~\ref{fig:snapshots}.
The intersection of $G$ with 1/2 corresponds to equilibria in OA, with $F$ corresponds to equilibria in DG (and DMI), and with $Q$ corresponds to equilibria in RBTS.  
The function $Q$ closely tracks $G$ for high enough precision values. 
We therefore find that the same stability and number of equilibria occur in RBTS as in DG.
However, we note that bifurcation (the increase in equilibria) occurs at lower precision values for RBTS (see \S~\ref{appendix:experiments} for visualization).
From a design perspective, a lower bifurcation point in RBTS is nice because there is a wider set of signal distributions allowing for some choice over stable thresholds. 

\section{Flexibility of a Richer Report Space} \label{sec:positive-results}
While a single, inflexible threshold equilibrium at the median exists under the Gaussian model, we found in Section~\ref{sec:experiments} that multiple stable threshold equilibria emerge when the underlying signal distribution is multimodal with low individual signal noise.
However, in peer prediction settings the designer cannot control, and often cannot diagnose, modality of the signal distribution.
Moreover, when the distribution is unimodal, our results remain negative.
A natural next question is whether these negative results persist when the report space is larger: i.e., what if the designer aims to employ a peer prediction mechanism designed for three reports instead of two?

In this section, we study an extension of DG to non-binary reports, \emph{Correlated Agreement} (CA) \citep{shnayder2016informed}.
Under the same model of equilibria and dynamics, we find issues of inflexibility persist for any parameters of the mechanism: if there are $m$ possible reports under the Gaussian model, there only exists a single nontrivial, stable threshold equilibrium $\tau \in \reals^{m-1}$.
While the outlook is still negative in this case, we uncover a new peer prediction framework with more flexibility by leveraging our results for the non-binary report setting.
Specifically, if the designer intends for agents to map their signals to a binary report, what if we first ask for more?
E.g., the designer uses CA with a report space $\{1,2,3,4\}$ of size $m=4$, and is then able to generate a \emph{mapping} back to $\{L, H\}$.
If there exists a three-dimensional stable threshold equilibrium at, say, $(-0.5, 0, 0.5)$, the designer can choose between a binary threshold of $-0.5$, $0$, or $0.5$ through this mapping.
We formalize this \emph{report-mapping framework} in \S~\ref{subsec:report-mapping} and discuss how the designer can set CA parameters in order to guarantee these choices exist over thresholds.

\subsection{Correlated Agreement} \label{subsec:correlated-agreement}

We begin by studying the non-binary report setting under a natural extension of DG called \emph{Correlated Agreement} \citep{shnayder2016informed}.
At a high level, CA gives the designer flexibility in allowing agents to receive payments when their peers submit reports that differ from their own.

Let $[m] = \{1,2,\dots,m\}$, and let the report space be $\R = [m]$.
Each agent $i$ still calculates a report $r_i \in \R$ according to deterministic strategy $\sigma_i: \reals \to \R$.
Now we consider a threshold vector $\tau \in \reals^{m-1}$ with $\tau_1 < \tau_2 \dots < \tau_m$, where a threshold strategy is defined as follows:
\[
	\sigma^{\tau}(x)=
	\begin{cases}
	1, & x \leq \tau_1 \\
    2, & \tau_1 < x \leq \tau_2 \\
    \ldots \\
    k, & \tau_{k-1} < x \leq \tau_k \\ 
    \ldots \\
	m, & x > \tau_{m-1}.
	\end{cases}
\]
For convenience we define $\tau_0 = -\infty$ and $\tau_m = \infty$.  

In CA, the designer chooses a symmetric, binary $m \times m$ \emph{score matrix} $\Delta$. 
While~\citet{shnayder2016informed} suggest initializing $\Delta$ as the sign matrix of correlations between signal bins, this underlying structure is often unknown, and furthermore depends on the current thresholds used by agents under our dynamics model.
Thus we allow the designer to fix the score matrix beforehand.
Then if agents $i$ and $j$ submit reports $r_i$ and $r_j$ for some shared task, and a peer reports $r_k$ on a distinct, i.i.d. task, agent $i$ receives payment $\Delta(r_i, r_j) - \Delta(r_i, r_k)$.

As in DG, if all agents are using a symmetric threshold strategy $\sigma^{\tau}$, agent $i$'s expected penalty for reporting $\ell$ is $\pi_{\ell}$, where $\pi_{\ell} = \sum_{k \in [m]} \Pr[\tau_{k-1} < X \leq \tau_k] \Delta(\ell, k)$.
That is, $\pi_{\ell}$ is the expected score over the prior of report $\ell$ against another agent using the symmetric strategy on a randomly selected task. 
Thus we work with the following formulation of CA with the same expected payment:
\begin{equation} \label{eq:ca-payment}
M_{CA}(r_i, r_j, \pi) = \Delta(r_i, r_j) - \pi_{r_i}. 
\end{equation} 
Note that when using the identity matrix as $\Delta$, CA exactly corresponds to extending the DG mechanism to $m$ reports (and in the case of binary reports, the identity matrix is the \emph{only} valid score matrix, meaning CA is equivalent to DG).
However, for other matrices $\Delta$, an agent's expected utility now depends on the probability of \emph{any} report $j$ such that $\Delta[i,j] = 1$.

Under more general monotonicity conditions over the utility functions of each report $k$ (see Condition~\ref{cond:seq-envelope-ca} in \S~\ref{app:ca}), we can derive a similar characterization of equilibria by equating utilities of adjacent reports at each threshold. 
As their forms are intuitive extensions of our analysis for binary-report settings, we leave the formal equilibrium characterization to \S~\ref{app:ca}. 
We also note that it is possible for thresholds to collapse, e.g. $\tau_1 = \tau_2 = 0$.
In this case the score matrix also collapses to a smaller report space, and we can use our characterization to check for equilibria with our new effective number of thresholds. 

\paragraph{Dynamics.}
If our monotonicity condition (Condition~\ref{cond:seq-envelope-ca}) hold for a fixed threshold $\tau$, it immediately follows that the best response to a threshold strategy $\tau \in \reals^{m-1}$ is also a threshold $\hat \tau \in \reals^{m-1}$. 
We can therefore study the same dynamics model, with $\dot \tau = \hat \tau - \tau$. 
As $m$ increases, analytically studying the dynamics of threshold strategies becomes more unwieldly, so in this section we numerically simulate the system for different initial conditions and visualize the vector field to understand which equilibria are stable or unstable.
We focus on the Gaussian model, normalizing the base signal so its variance is $a^2 = 1$.
After numerically guaranteeing that Condition~\ref{cond:seq-envelope-ca} holds in each case, we vary the individual signal precision ($1/b^2$) and study how equilibria change.

	\begin{figure}[t]
		\begin{minipage}[t]{0.49\textwidth}
			\centering
			\includegraphics[width=0.49\textwidth]{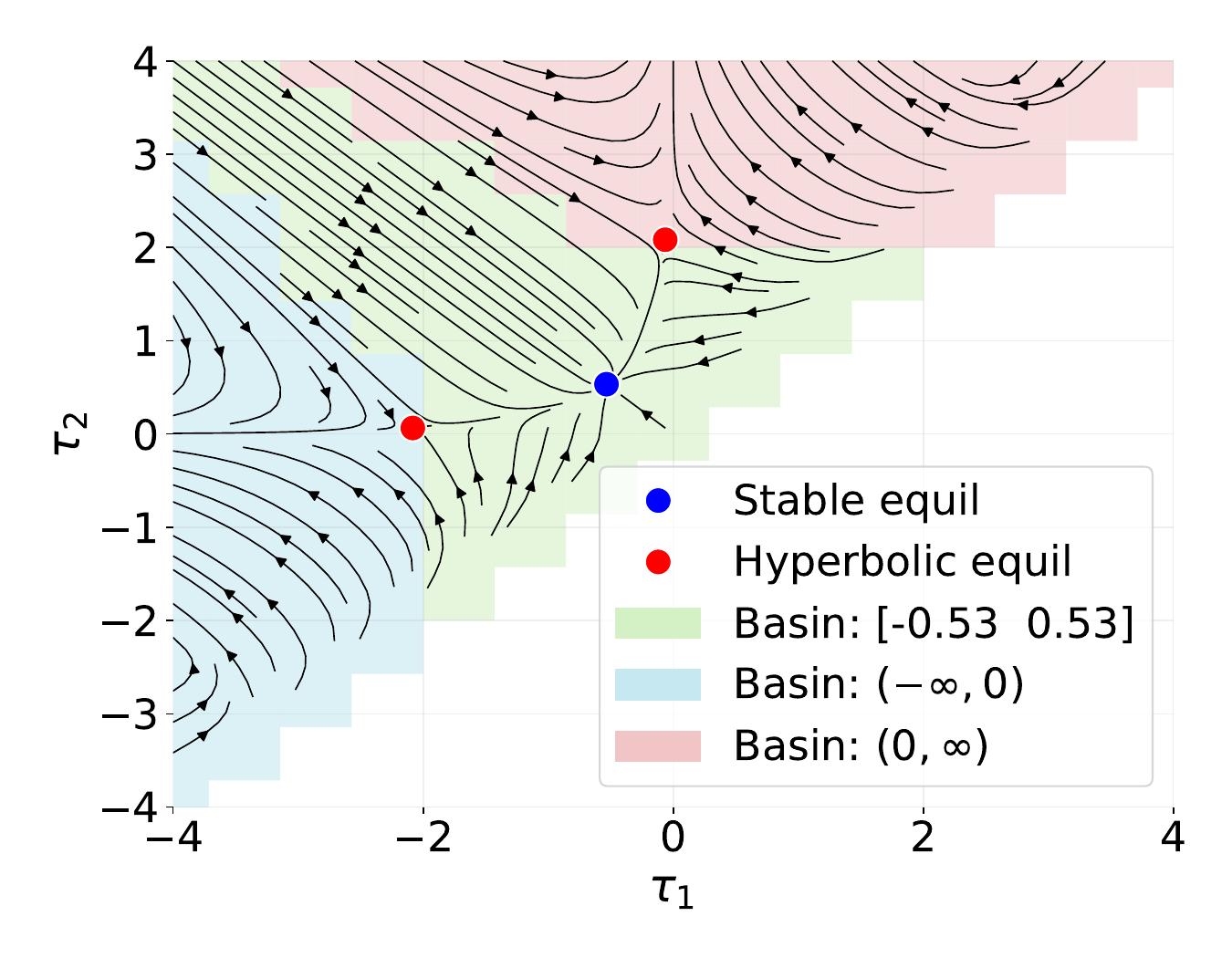}
			\includegraphics[width=0.49\textwidth]{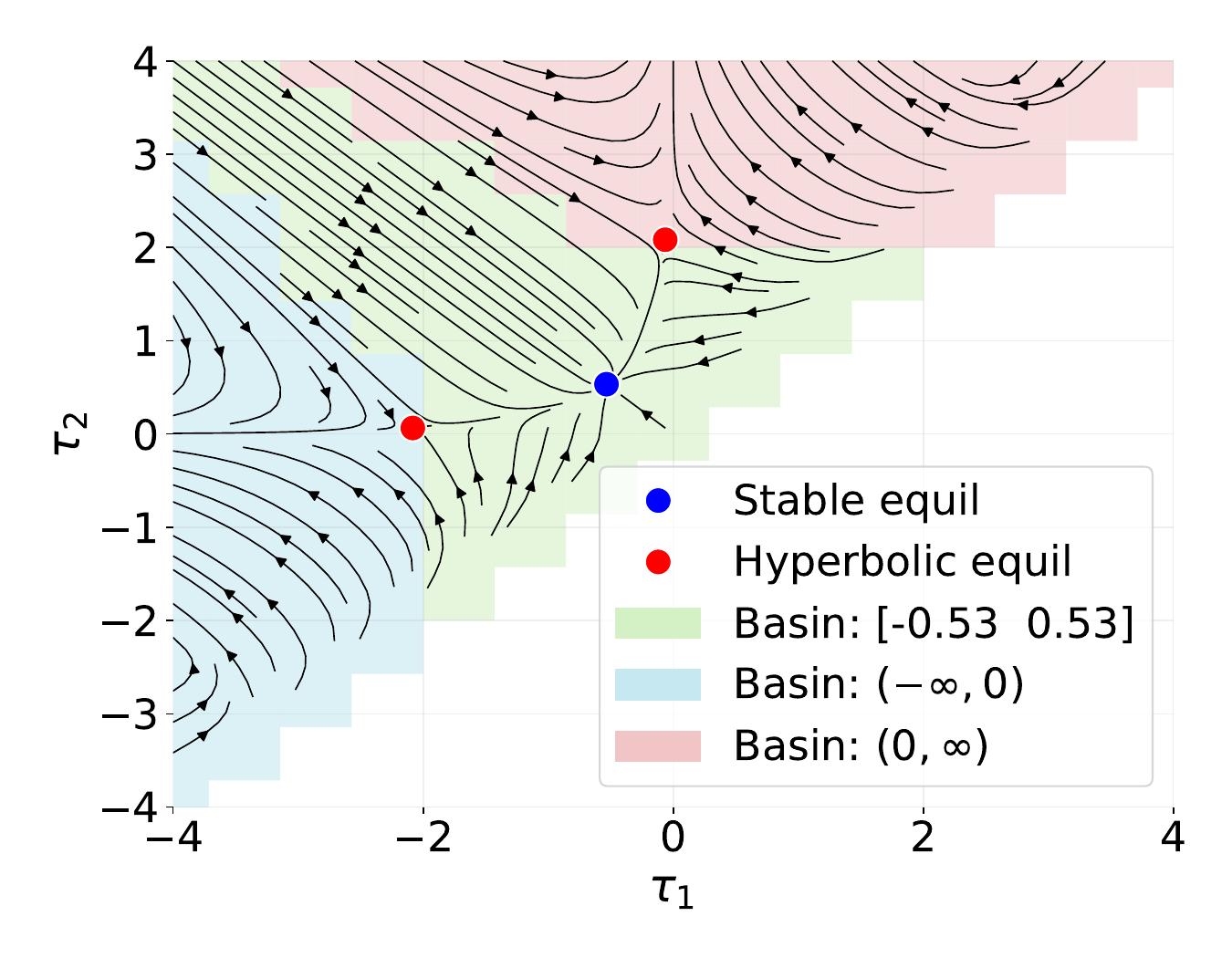}
			\caption{Vector fields of thresholds $\tau \in \reals^2$ for CA with score matrix $\Delta_I$, $m=3$, and precision $1$ (left) and $4$ (right). The basins of attraction for each stable equilibrium are indicated by color.}
			\label{fig:3-report-ca-vector}
		\end{minipage}
		\hfill
		\begin{minipage}[t]{0.49\textwidth}
			\includegraphics[width=0.49\textwidth]{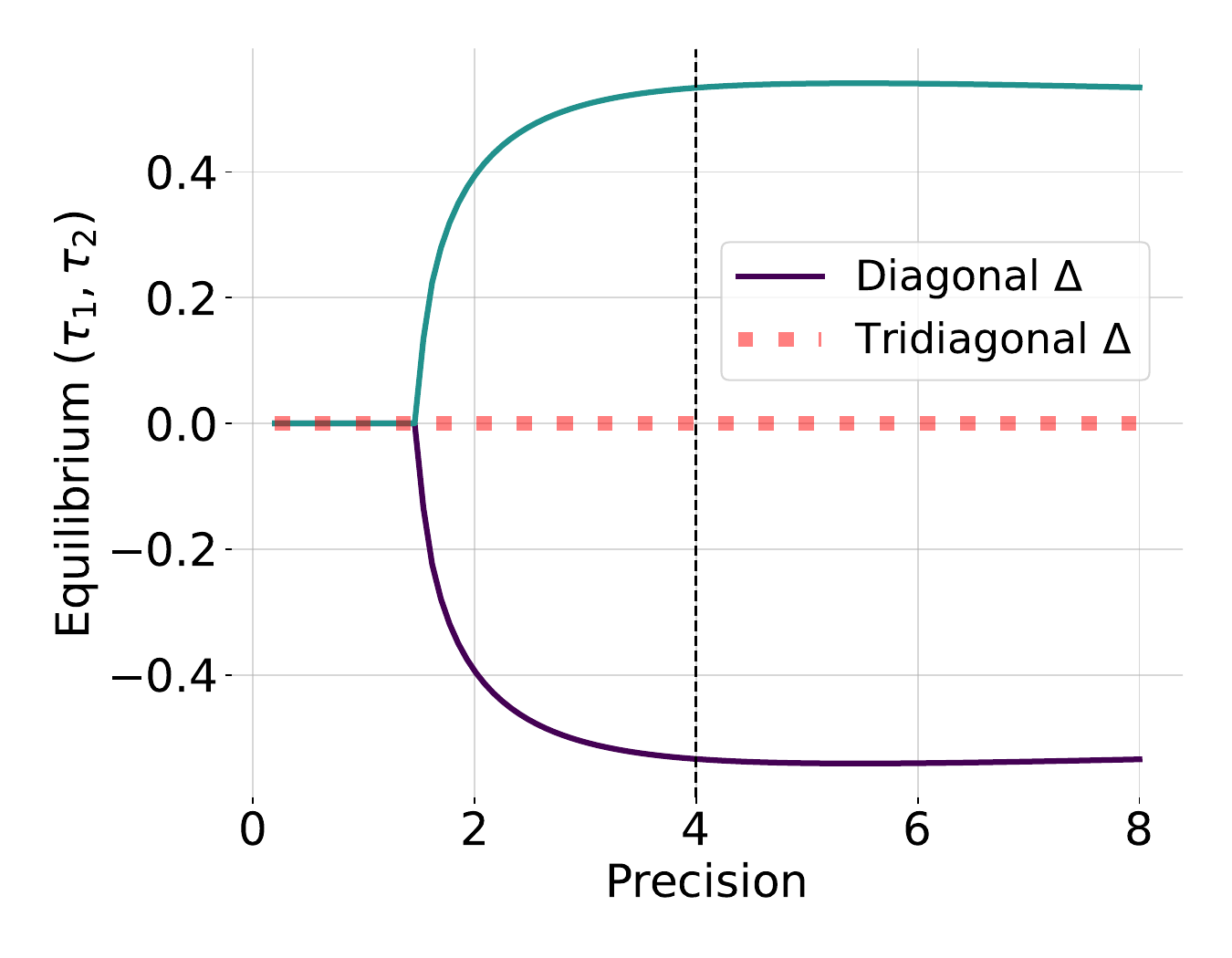}
			\includegraphics[width=0.49\textwidth]{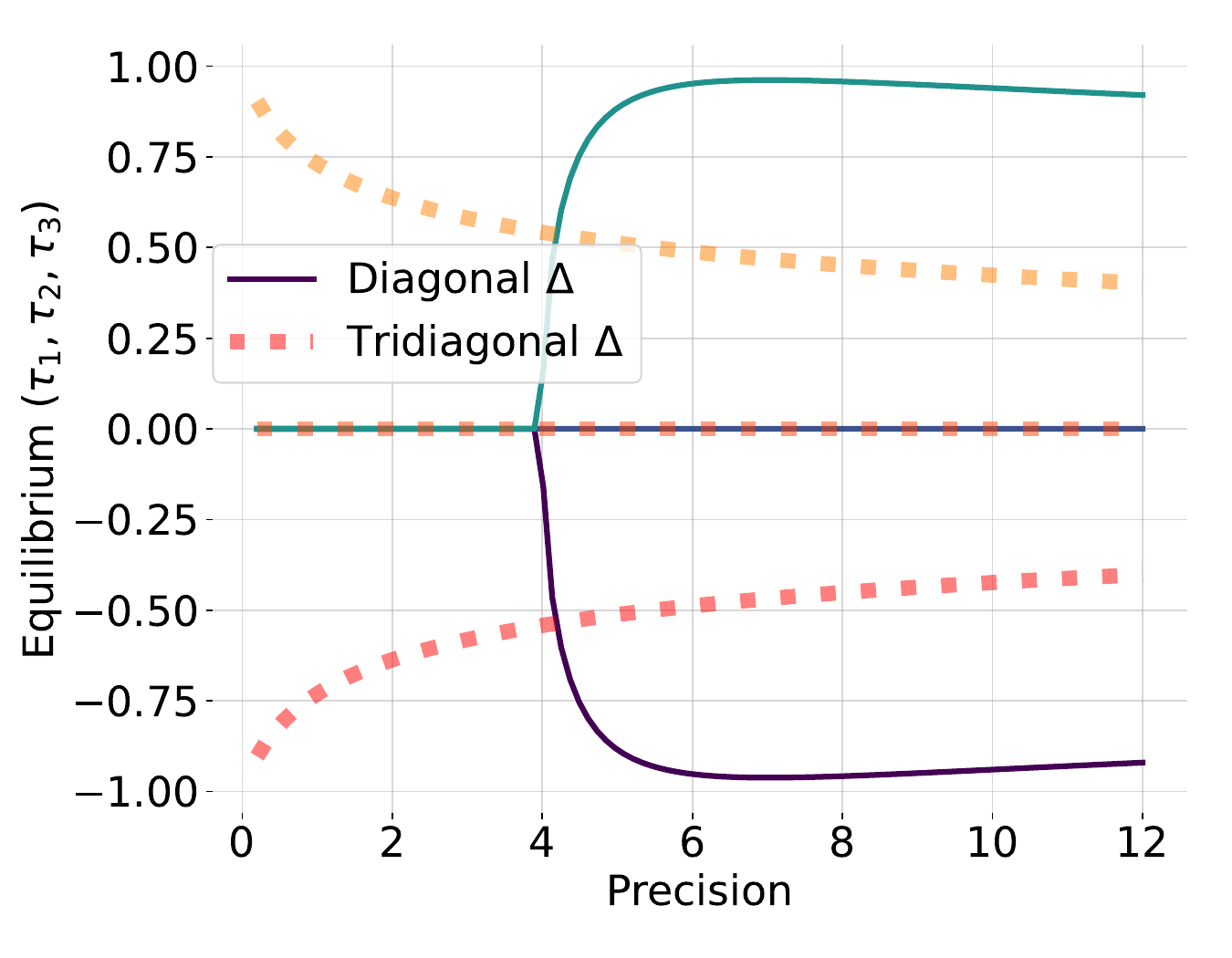}
			\caption{Nontrivial stable threshold equilibria for CA with score matrix $\Delta_I$ (solid) or $\Delta_T$ (dotted) for $m=3$ (left) and $m=4$ (right), varying precision of signals $Z_i$ under the Gaussian model.
			}
			\label{fig:ca-over-precision}
		\end{minipage}
	\end{figure}

\subsection{Diagonal score matrix (DG)}

Let $\Delta = \bm{I}$ be the identity matrix, so that CA serves as an extension of DG to non-binary reports.
For any precision value, when $m=3$ we find there exists a stable equilibrium at a point $(-c, c)$ for some $c \geq 0$.
Figure~\ref{fig:ca-over-precision} shows how the two thresholds at $-c$ and $c$ change with signal precision.
At small precision values, $c$ collapses to $0$, so that we return to the same equilibrium as the binary report case (e.g., report `1' for all signals below $0$, and `3' otherwise).
Because the posterior is flatter with longer tails, there is enough mass above any threshold $\tau_2 > 0$ relative to the prior that agents with signals $\tau_2 - \epsilon$ will prefer reporting `3', and vice versa for $\tau_1 < 0$ with signals $\tau_1 + \epsilon$.
However, for larger precision values, shorter tails allow the value of $c$ to move away from zero.

Figure~\ref{fig:3-report-ca-vector} illustrates that alongside this nontrivial stable equilibrium (blue), there also exist stable equilibria at $(-\infty, 0)$ and $(0, \infty)$.
Again, these equilibria effectively collapse thresholds to $0$, e.g. agents report `2' when receiving signals below $0$ and `3' otherwise in the former case. 
These collapsed, stable equilibria are separated from the center point $(-c,c)$ by two hyperbolic equilibria (red).
Intuitively, the basin of attraction for such collapsed equilibria correspond to settings with enough positive correlation between two adjacent signals far enough in the tails that agents are incentivized to collapse their meanings together. 
In general, then, we observe that there still only exists (at most) one nontrivial stable equilibrium beyond a collapsed threshold of $0$, and the designer has no choice in where this point is.

The pattern continues when $m = 4$ in Figure~\ref{fig:ca-over-precision} (indicated by the solid lines), with higher correlation coefficients leading to nontrivial stable threshold equilibria at some symmetric $\tau = (-c, 0, c)$.
Again, there exist several stable equilibria collapsing the threshold to $0$, now at $(-\infty, 0, 0)$, $(0, 0, \infty)$, $(-\infty, -\infty, 0)$, and $(0, \infty, \infty)$, separated by saddle points from the center equilibrium at $\tau$.
Our observations generalize for larger $m$, with at most one nontrivial stable equilibrium of dimension $m-1$ emerging for high enough precision (see \S~\ref{subsec:larger-m} for more details).

\subsection{Tridiagonal matrix}
We also study equilibria and their stability under different precison values when the score matrix corresponds to the \emph{tridiagonal matrix} $\Delta_{D}$, where $\Delta_T[i,j] = 1$ when $i$ and $j$ are \emph{adjacent}.
E.g., for $m=3$ and $m=4$ respectively, $\Delta_T$ looks like
\[
	\begin{pmatrix} 1 & 1 & 0 \\ 1 & 1 & 1 \\ 0 & 1 & 1 \end{pmatrix} \qquad \text{and} \qquad \begin{pmatrix} 1 & 1 & 0 & 0 \\ 1 & 1 & 1 & 0 \\ 0 & 1 & 1 & 1 \\ 0 & 0 & 1 & 1 \end{pmatrix}. \\
\]

When signal precision is lower, it is reasonable to expect this mechanism is robust against thresholds collapsing because agents receive payments for a neighborhood of signals which may be more likely under their posteriors. 
We begin by noting that for a report space of size $m = 3$, the expected reward under the tridiagonal matrix for report $2$ is zero, and agents are always incentivized to collapse this middle signal and return to the binary report case (see \S~\ref{subsec:tridiagonal-m-3} for a formal analysis).\footnote{This perhaps reflects the fact that in the original CA setup, the score matrix cannot be tridiagonal when $m=3$ because a signal cannot be positively correlated with all signals (as the second row indicates).}

We visualize the setting of $m = 4$ with the score matrix $\Delta_T$ in Figure~\ref{fig:ca-over-precision}.
Unlike the diagonal score matrix, here there exists a \emph{unique}, non-zero, stable equilibrium $\tau \in \reals^{3}$ \emph{even when the precision is low} (indicated by the dotted lines).
For example, when the precision is $1$, a stable equilibrium emerges under $\Delta_T$ with $\tau \approx (-0.73, 0, 0.73)$.
While a collapse to $0$ does not occur, agents will still always converge to this unique set of thresholds, and thus the designer still has no flexibility in choice over equilibria. 
We numerically confirm this behavior for higher $m$ in \S~\ref{subsec:larger-m}.

\subsection{Report mapping framework} \label{subsec:report-mapping}

The previous subsections show that with non-binary report spaces, employing CA with either score matrix leads to the same inflexibility of equilibria. 
So is there any hope in developing a mechanism that provides the designer more choices in how agents map signals to reports? 
To make progress toward this goal, we propose a framework which leverages our results for non-binary reports. 

\begin{definition}{(Report mapping framework)}
	We define the \emph{report mapping framework} as follows.
	\begin{enumerate}
		\item Run the Correlated Agreement mechanism over $m$ \emph{raw reports} with fixed score matrix $\Delta$. E.g., if $m = 3$, agents submit one of three reports in $\{1,2,3\}$.
		\item Map each agent's raw report to a \emph{final binary report} of $L$ or $H$.
\end{enumerate}
\end{definition}

Intuitively, given our results in the previous section, running step 1 with $m$ reports will lead to a threshold equilibrium of dimension $m-1$ in some regimes.
For example, when $m=3$ and the precision is 2, there is a unique stable threshold equilibrium $\tau \in \reals^2$ when $\Delta = \bm{I}$ with $\tau \approx (-0.5, 0.5)$, so that agents report `1' if their signal is below $-0.5$, `2' if their signal is between $-0.5$ and $0.5$, and `3' otherwise.
In step 2, the designer can then choose from either $-0.5$ or $0.5$ as a final threshold: e.g., a threshold of $-0.5$ is simulated by mapping reports of `2' and `3' to $H$ and reports of `1' to $L$. 
This is in contrast to our results using binary-report peer prediction mechanisms directly, where agents inevitably converge to a threshold of $0$. 

\paragraph{Using the identity matrix $\Delta_I$.}

As previously mentioned and shown in Figure~\ref{fig:ca-over-precision}, when the precision is low enough equilibria still collapse to a single threshold at $0$ under score matrix $\Delta_I$, meaning there is no inherent benefit in running the report mapping framework for $m=3$ or $m=4$.
However, when precision is high, the designer can choose from three different thresholds.
For example, if $m=3$ and the precision is $\geq 3/2$, a nontrivial threshold equilibrium emerges of the form $(-c, c)$, limiting to around $(-0.5, 0.5)$; and for $m=4$, if the precision is $\gtrapprox 4$, a nontrivial threshold equilibrium emerges of the form $(-c, 0, c)$ limiting to around $(-0.75, 0, 0.75)$.
As evident in Figure~\ref{fig:3-report-ca-vector}, the basin of attraction for these nontrivial equilibria also increases in size as precision grows, so that there is a larger margin of error for the designer in estimating the initial threshold.

We note that the bifurcation point at which equilibria no longer collapse to $0$ is higher for $m=4$ than for $m=3$, intuitively because the amount of probability per report bin under the signal distribution decreases and exacerbates agents near thresholds reporting against them.
We numerically verify that this pattern continues, with the minimum precision required for a non-zero, stable equilibrium $\tau \in \reals^{m-1}$ to emerge increasing over $m$; see \S~\ref{subsec:larger-m} for visualization.
Thus there is a tradeoff when using $\Delta_I$: while larger values of $m$ increase the number of potential thresholds in the report-mapping framework that the designer can choose from, the signal precision required for a non-zero, stable equilibrium $\tau \in \reals^{m-1}$ to emerge in the first place is also increasing over $m$.

\paragraph{Using the tridiagonal matrix $\Delta_T$.}

As a reminder, when $m=3$ our only nontrivial scoring matrix option is the identity. 
But for $m=4$, employing the report-mapping framework when using the score matrix $\Delta_T$ is useful, even for low precision values. 
Specifically, there exists a nontrivial, unique threshold equilibrium even when agent information is highly noisy. 
Intuitively, the correlation between an agent's signal and adjacent bins is high enough under $\tau$ that, even under higher signal noise, they are incentivized to report truthfully according to the bin their signal is in.
In these regimes, then, the mechanism designer is able to choose between three different threshold mappings of signals to final reports ($-c$, $0$, or $c$ for some $c > 0$). 

This pattern continues for larger $m$, with a single stable threshold equilibrium $\tau \in \reals^{m-1}$ existing for any precision level. 
We note that there is a tradeoff for larger report spaces: as $m$ grows and precision limits to $0$, the outer thresholds move toward $\pm \infty$. 
However, this transition to a smaller-dimension equilibrium occurs at a significantly lower precision level than the equilibrium collapse under the diagonal matrix $\Delta_I$.
We direct the reader to \S~\ref{subsec:larger-m} for a more detailed discussion. 

\paragraph{Takeaways.}
In practice, then, using the report-mapping framework with the tridiagonal matrix $\Delta_T$ is useful when the designer is either unsure of underlying agent signal precision, or believes the precision is low.
For example, in peer grading the designer may believe there is high variability in how individual graders read an essay due to skill or lack of effort, or there is high subjectivity or disagreement about the quality of a submission.
When the designer has confidence that signal precision is higher, e.g. graders evaluate essay qualities more consistently, the diagonal matrix $\Delta_I$ can also be used to ensure there exist several threshold options for the designer.

We additionally note that while this section focuses on settings where the principal is seeking out binary reports, based on our results the report-mapping framework could be adapted to any final report space, as long as the raw report space is set to be larger.
\section{Discussion}
\label{sec:discussion}
We initiate the study of peer prediction when signals are finer-grained than reports. 
To start, we characterize behavior in several binary-report peer prediction mechanisms when signals are real-valued. 
While all these mechanisms have various truthfulness guarantees in the binary signal setting, in our continuous signal model we find that the notion of truthfulness breaks. 
Specifically, when agents map signals to reports according to a real-valued threshold, under dynamics we find equilibria rarely stay at the initial threshold a mechanism designer sets. 
Instead, agents will often stabilize at a different (and potentially uninformative) threshold according to the specific payment scheme.
We demonstrate that this inflexibility of equilibria extends to other general settings where signals are finer-grained than reports, e.g. when the report space is non-binary or when the signal space is finite (see below).

Motivated by these negative results, we present a new report-mapping framework which allows the principal more flexibility in setting threshold equilibria.
By using the Correlated Agreement mechanism to elicit raw reports from a larger space, the designer has the flexibility to choose several different mappings between these reports and $\{L,H\}$ depending on the norms with which they prefer agents to interpret their signals. 

Our results imply that the standard modeling assumptions in peer prediction literature miss important nuances, and these holes can propagate to real changes in how we expect agents to behave in practice. 
In particular, models of peer prediction have often taken as given an arbitrary joint distribution $P$ over pairs of signals drawn from a finite set.
This suggests that we should be able to choose an arbitrary meaning for each signal and then classic peer prediction should ``just work'' for the resulting distribution.
In contrast, our results show that classic peer prediction is substantially more inflexible as currently used in practice.  
Fortunately, our report-mapping framework avoids this collapse of information to a single equilibrium, by asking agents to map their information to a finer list of reports first. 
We discuss several extensions of our results and avenues for future work below. 

\paragraph{Beyond real-valued signals.}
We additionally study OA and DG when the signal space is finite and larger than the binary report space in \S~\ref{app:finite-signal}.
When signals are ordered by stochastic dominance, we find inflexibility persists: uninformative equilibria remain stable under OA, while only a few stable threshold equilibria exist around the middle signals under the Gaussian model for DG. 
A natural next question for future work, then, is how peer prediction mechanisms can be adapted to settings where there is no underlying monotone ordering of signal meanings. 
Moreover, for other settings like labeling an image from $\{$cat,dog,fish$\}$, one would not expect a single-dimensional real-valued signal model to be appropriate, with perhaps $\reals^3$ being a better choice.

\paragraph{Further broadening the model.}
Another missing piece in our model is the fact that in many settings some deterministic \emph{context}, like the content of the essay being graded, is shared by all looking at the same task.
In principle, participants may choose strategies that depend on these inessential details, like the first word of the essay, to correlate their reports, all while getting optimal payoffs in multi-task mechanisms.
While there is work dealing with behavior heterogeneity~\citep{agarwal2020peer}, such behavior seems impossible to fully rule out.
Moreover, even if participants are well-meaning, there is often still ample metadata about tasks such as categories or difficulties that naturally influence participants' process of translating their actual experience into a ``signal.''

Studying effort would also be natural in our model.
For example, one could model effort in the Gaussian model as giving the agent $aZ + bZ_i$ where $b$ is decreasing in their exerted effort.
One interesting interplay with our dynamics is that low-effort agents seem most likely to move the threshold, as high-effort agents will have more reason to respect the current one; e.g. when all agents are high-effort, they all coordinate perfectly on the correct report.
\paragraph{CA with other score matrices.}
While the tridiagonal matrix $\Delta_T$ is a natural score matrix for Correlated Agreement when signals follow the symmetric Gaussian model, using other score matrices could be useful in settings of multimodality or skewness (e.g. rewarding larger clusters of report bins when the distribution is skewed toward them).
Future work could therefore include studying how to ideally set the score matrix when the principal has some knowledge about the underlying information structure, as well as better understanding the relationship between behavior under a fixed score matrix and the true correlation structure of signals.

\subsection*{Acknowledgments}

This paper is based on initial results when the first two authors were working with the Oinc cryptocurrency team.
We thank Tim Roughgarden, and members of Oinc and Kleros, for useful discussions.

\bibliography{references}

\break

\appendix

\section{Omitted Results for Output Agreement} \label{appendix:oa}
\subsection{Equilibrium characterization generalization}
In general, we can identify \emph{where} finite equilibria exist in OA based on the extreme behavior of the function $G$.

\begin{condition} \label{cond:oa-interval}
	Let $I = [a, b] \subset \reals$ be an interval. 
	Upon seeing signal $x \leq a$, then an agent believes with probability less than 1/2 that another signal will be less than $x$.
	Meanwhile, upon seeing signal $x \geq b$, then an agent believes with probability greater than 1/2 that another signal will be less than $x$. 
  	Formally,
	$x \leq a \implies P(x'\leq x \mid x) < 1/2$ and
	$x \geq b \implies P(x'\leq x \mid x) > 1/2$.
\end{condition}

\begin{theorem} \label{thm:oa-interval}
	Let the agent signal structure satisfy Condition \ref{cond:oa-interval}, with $\Pr[X' \leq \tau \mid X = x]$ monotone decreasing and continuous in $x$, and $G$ continuous.
	Then there exists an equilibrium $\tau^* \in I$.
\end{theorem}

\begin{proof}
	Condition \ref{cond:oa-interval} implies $G(a) < 1/2 $ and $G(b) > 1/2$.
	Since $G$ is continuous, we can apply the Intermediate Value Theorem to conclude there exists a point $\tau^* \in [a, b]$ such that $G(\tau^*) = 1/2$. 
	Then $\tau^*$ is an equilibrium by Theorem \ref{thm:oa-equil}.
\end{proof}
Note in particular that if $G(x)$ is monotone increasing and crosses 1/2, Condition~\ref{cond:oa-interval} is satisfied for the single point interval $\tau$ where $G(\tau) = 1/2$ and we have a unique equilibrium.

\begin{proposition} \label{prop:oa-skewed}
	Let the agent signal structure satisfy Condition \ref{cond:oa-interval}, with $\Pr[X' \leq \tau \mid X = x]$ monotone decreasing and continuous in $x$, and $G$ continuous.
	Let $m$ be the median of $F$, and consider running the OA mechanisim.
	Then if $\Pr[X' \leq m \mid m] > 1/2$, there exists an equilibrium $\tau^*$ such that $\tau^* < m$.
	Similarly, if $\Pr[X' \leq m \mid m] < 1/2$, there exists an equilibrium $\tau^*$ such that $\tau^* > m$.
\end{proposition}

\begin{proof}
	First, by Theorem \ref{thm:oa-interval}, we know at least one equilibrium $\tau^*$ exists in the interval $I$. 
	Now assume that $\Pr[X' \leq m \mid m] = G(m) > 1/2$. 
	It follows under Condition~\ref{cond:oa-interval} that either $m > b$ or $m \in I$.
	If $m > b$, then by Theorem~\ref{thm:oa-interval} there exists an equilibrium $\tau^*$ in the interval $I$; $\tau^* \leq b < m$, so the result follows.
	If $m \in I$, then $m \geq a$.
	Since $G$ is continuous, $G(a) < 1/2$, and $G(m) > 1/2$, by the Intermediate Value Theorem there exists a point $\tau^* \in [a, m]$ such that $G(\tau^*) = 1/2$.
	By our characterization in Theorem~\ref{thm:dg-equil}, $\tau^*$ is an equilibrium.

	A symmetric argument holds for when $\Pr[X' \leq m \mid m] < 1/2$. 
\end{proof}

\subsection{Dynamics generalization}

We can observe instability of equilibria under Condition~\ref{cond:oa-interval}.
\begin{proposition} \label{prop:oa-unstable}
	Let the agent signal structure satisfy Condition~\ref{cond:oa-interval} for interval $I$, and such that $G$ is continuous and differentiable. 
	Then there exists an equilibrium $\tau^* \in I$ which is unstable, while the equilibria $\pm \infty$ are stable.
\end{proposition}

\begin{proof}
	We know already from Theorem \ref{thm:oa-interval} that at least one equilibrium exists in the interval $I$. 
	Note first that since $G(a) < 1/2$ and $G(b) > 1/2$, by continuity and differentiability $G$ must cross 1/2 at some point $\tau_0$ with $G'(\tau_0) > 0$. 
	Assume not: then $G(x) \leq 1/2$ for all $x$, which contradicts the fact that $G(x) > 1/2$ for $x \geq b$.
	It follows by Theorem \ref{thm:oa-dynamics} that $\tau_0$ is an unstable equilibrium.

	\begin{enumerate}
		\item Now, WLOG we pick the equilibrium $\tau_1$ which is closest to $a$. 
		By Condition \ref{cond:oa-interval} and since $G(x)$ is continuous, $G(x) < 1/2$ for $x < \tau_1$. 
		Then it must follow that $G'(\tau_1) \geq 0$. 
		Thus $\tau_1$ is unstable, at least on the left (if $G'(\tau_1) = 0$, then $\tau_1$ is unstable on the left and stable on the right).
		\item Pick the equilibrium $\tau_2$ which is closest to $b$ (this could be the same as the $\tau_1$ we chose in the previous step).
		By Condition \ref{cond:dg-interval} and since $G(x)$ is continuous, $G(x) > 1/2$ for $x > \tau_2$. 
		Then it must follow that $G'(\tau_2) \geq 0$. 
    Thus $\tau_2$ is unstable, at least on the right (if $G'(\tau_2) = 0$, then $\tau_2$ is unstable on the right and stable on the left).
	\end{enumerate}
	We know by Condition \ref{cond:oa-interval} that $G(x) \neq 1/2$ for $x \notin [a, b]$, so that no equilibria occur outside the interval $I$ other than $\pm \infty$.
	Since both $\tau_1$ and $\tau_2$ (or just $\tau_1$, if they are the same equilibrium) are unstable (at least to the left for $\tau_1$, and right for $\tau_2$), it is topologically necessary that the uninformative equilibria $\pm \infty$ are stable.
\end{proof}

In most reasonable settings, one would expect $G(\tau)$ to behave according to Condition \ref{cond:oa-interval}.
Specifically, as a signal $x$ becomes increasingly small, the probability another agent receives a signal smaller than $x$ should approach 0.
Meanwhile, as a signal $x$ grows large, the probability another agent receives a signal smaller than $x$ should approach 1. 
Proposition \ref{prop:oa-unstable} thus suggests that when information corresponds predictably and monotonically to quality of an essay or task, over time agents will naturally move toward uninformative equilibria. 
Intuitively, these dynamics make sense: an agent who thinks others are submitting more ``high'' reports will shift their threshold down to match, reinforcing larger and larger thresholds over time. 

\section{Omitted Results for Dasgupta-Ghosh} \label{appendix:dg}

\subsection{Equilibrium characterization}
We begin by formally characterizing equilibria for DG, as in OA.
By Equation~\ref{eq:dg-util}, the condition for a symmetric Bayes-Nash equilibrium can be written as
\[
	\forall x \in \reals,\; \E_{x' \sim \beta(x)} \ones[\sigma(x) = \sigma(x')] - \pi_{\sigma(x)} \geq \E_{x' \sim \beta(x)} \ones[\overline{\sigma(x)} = \sigma(x')] - \pi_{\overline{\sigma(x)}}.
\]
Since $\E_{x' \sim \beta(x)} \ones[\sigma(x) = \sigma(x')] = 1 -  \E_{x' \sim \beta(x)} \ones[\overline{\sigma(x)} = \sigma(x')]$, this condition simplifies to
\begin{equation} \label{cond:general-equil}
	\forall x \in \reals,\; \E_{x' \sim \beta(x)} \ones[\sigma(x) = \sigma(x')] \geq \pi_{\sigma(x)}.
\end{equation}
For a threshold equilibrium, Condition~\ref{cond:general-equil} is equivalent to:
\begin{align}
	\forall x \leq \tau, P(\tau; x) \geq \pi_L \label{eq:equil-1-dg} \\
	\forall x > \tau, P(\tau; x) \leq \pi_L. \label{eq:equil-2-dg}
\end{align}

Now we have the tools to analyze threshold equilibria in DG.
We first observe that uninformative strategies $\tau^* = \pm \infty$ remain equilibria in the DG mechanism.
\begin{proposition} \label{prop:dg-infinite-equil}
	$\tau^* = \pm \infty$ are both always threshold equilibria under DG.
\end{proposition}

\begin{proof}
	Let $\tau^* = \infty$. 
	Then for all signals $x$ such that $x \leq \tau^*$, $P(\tau^*; x) = 1 \geq \Pr[X' \leq \infty]$ (since any $X' \in \reals$ satisfies $X' < \tau^*$ with probability one.)
	Meanwhile, Statement~\ref{eq:equil-2-dg} is vacuous since no signal $x \in \reals$ satisfies $x > \infty$. 
	A similar argument follows for $\tau^* = -\infty$: for all signals $x > \tau^*$, $P(\tau^*; x) = 0 \leq \Pr[X' \leq -\infty]$, while Statement~\ref{eq:equil-1-dg} is vacuous since no signal $x \in \reals$ satisfies $x < -\infty$.
\end{proof}

We can also provide necessary and sufficient conditions for a finite threshold equilibrium for DG.  Both their form and proof follow the same logic we saw for OA, adapted to account for the additional $\pi_L$ term also depending on the choice of threshold.

\begin{theorem} \label{thm:dg-equil}
	Let finite threshold $\tau$ be given and $P(\tau; x)$ be continuous in $x$.
	If $\tau$ is a threshold equilibrium under the DG mechanism then $G(\tau) = F(\tau)$.
	Conversely, if $G(\tau) = F(\tau)$ and either (a) $P(\tau; x)$ is monotone decreasing in x 
	or (b) $P(\tau; x) -\pi_L$ has a single crossing of $0$ from positive to negative 
	then $\tau$ is a threshold equilibrium under the DG mechanism.
\end{theorem} 

\begin{proof}
	For necessity, assume that $\tau$ is a threshold equilibrium, so
	Equations \eqref{eq:equil-1-dg} and \eqref{eq:equil-2-dg} hold.
	Since $P(\tau; x)$ is continuous, we must have $\lim_{x \to \tau^+} P(\tau; x) = \lim_{x \to \tau^-} P(\tau; x) = P(\tau; \tau).$
	Further, $\lim_{x \to \tau^+} P(\tau; x) \leq \pi_L$ and $\lim_{x \to \tau^-} P(\tau; x) \geq \pi_L$, so  $G(\tau) = P(\tau; \tau) = \pi_L  = F(\tau)$.

	For sufficiency, assume $G(\tau) = F(\tau)$.  Equivalently $P(\tau; \tau) - \pi_L = 0$, so (a) implies (b) and the single crossing is at $\tau$.  Thus assume (b) holds and
	take some $x > \tau$. 
  	By single crossing, $P(\tau; x) \leq P(\tau; \tau) = G(\tau) = F(\tau) = \pi_L$.
  	Similarly, if $x \leq \tau$, $P(\tau; x) \geq P(\tau; \tau) = \pi_L$.
  	This establishes Equations \eqref{eq:equil-1-dg} and \eqref{eq:equil-2-dg}, so $\tau$ is a threshold equilibrum.
\end{proof}

\subsection{Dynamics}
We consider stability of equilibria for DG under the dynamics described by Equation \eqref{eq:continuous-dynamic}.
Note that as in OA, when everyone else is playing according to $\tau$, there exists a unique best response which is also a threshold strategy.

\begin{proposition} \label{prop:dg-best-response}
	Assume that $P(\tau; x)$ is strictly decreasing and continuous over $x$ for a fixed threshold $\tau$.
	If all agents are playing according threshold strategy $\tau$, the unique best response of an agent across all strategies $\sigma: \reals \to \R$ is to play according to threshold strategy $\hat \tau$ satisfying  $P(\tau; \hat \tau) = F(\tau)$.
\end{proposition}

\begin{proof}
	If $\tau$ is the current threshold strategy that all other agents are following, an agent who receives signal $x$ will best respond with $H$ if $P(\tau; x) \leq F(\tau)$, and $L$ otherwise. 
	Since $P(\tau; x)$ is strictly decreasing and continuous over $x$, there will be a unique point $\hat \tau$ such that $P(\tau; \hat \tau) = F(\tau)$, with $P(\tau; x) \geq F(\tau)$ for $x \leq \hat \tau$ and $P(\tau; x) \leq F(\tau)$ for $x > \hat \tau$.
	This threshold $\hat \tau$ thus corresponds to the best response.
\end{proof}

Now, as in OA, we can characterize the stability of equilibria.

\begin{theorem} \label{thm:dg-dynamics}
	Assume for all $\tau \in \reals$ that $P(\tau; x)$ is strictly decreasing and continuous over $x$.
	Then if $G(\tau) - F(\tau)$ is strictly decreasing at equilibrium point $\tau^*$, $\tau^*$ is stable.  Similarly, if it is strictly increasing $\tau^*$ is stable.
\end{theorem}

\begin{proof}
	Consider the dynamics at time step $t$, with the current threshold $\tau(t)$.
	Since we are considering a fixed time $t$, we refer to $\tau(t)$ as $\tau$ throughout the proof.

	We consider what happens when an agent receives a signal exactly at $\tau$.
	There are three cases to consider.
	In the first, $G(\tau) = F(\tau)$; then the expected utility of reporting $L$ is $\Pr[X' \leq \tau \mid X = \tau] - \Pr[X' \leq \tau] = 0$, while the expected utility of reporting $H$ is $\Pr[X' > \tau \mid X = \tau] - \Pr[X' > \tau] = 0$; thus the agent is indifferent between reporting $H$ or $L$.
	Moreover, the best response $\hat \tau$ is exactly $\tau$.
	It follows the system is at an equilibrium since $\dot \tau = \hat\tau - \tau = 0$.

	In the second case, $G(\tau) > F(\tau)$. 
	Then reporting $L$ is strictly preferred to reporting $H$.
	Since $P(\tau; \hat \tau) = F(\tau) < G(\tau) = P(\tau; \tau)$, by decreasing monotonicity we have $\hat{\tau} > \tau$. 
	Thus $\dot \tau = \hat\tau - \tau > 0$.
	In the third case, $G(\tau) < F(\tau)$. 
	Then reporting $H$ is strictly preferred to reporting $L$.
	Since $P(\tau; \hat \tau) = F(\tau) > G(\tau) = P(\tau; \tau)$, by decreasing monotonicity we have $\tau > \hat{\tau}$. 
	Thus $\dot \tau = \hat\tau - \tau < 0$.

	If $G(\tau) - F(\tau)$ is strictly decreasing at $\tau^*$, then at a sufficiently close perturbed point $\tau$ to the left of $\tau^*$ we have $G(\tau) > F(\tau)$ and $\dot \tau > 0$; and at a perturbed point $\tau$ to the right of $\tau^*$ we have $G(\tau) < F(\tau)$ and $\dot \tau < 0$.
	Stability follows.  The same logic implies instability in the strictly increasing case.
\end{proof}


\subsection{Gaussian model} \label{subsec:gaussian-dg}

We revisit the Gaussian setting introduced in \S~\ref{subsec:gaussian-intro} under the DG mechanism.
First note that the same equilibria occur under DG as in OA.

\begin{corollary}
	In the Gaussian model under DG, we have three equilibria at $\tau = 0$ and $\pm \infty$.
\end{corollary}

\begin{proof}
	Existence of equilibria at $\tau = \pm \infty$ immediately follows from Proposition~\ref{prop:dg-infinite-equil}.

	Now, note that $P(\tau; x)$ (Equation~\eqref{eq:p}) is strictly decreasing and continuous so that Theorem \ref{thm:dg-equil} applies. 
	By Equation~\eqref{eq:gaussian-FG}, $G(0) = F(0) =\Phi(0)$, so $\tau=0$ is an equilibrium under DG.
	Moreover, there could not be another finite equilibrium unless the coefficients of $x$ in Equation~\eqref{eq:gaussian-FG} were equal.
	Letting $c_F = \frac {\sqrt{\rho}} {a}$ and $c_G = \frac {(1-\rho)} {b \sqrt{1+\rho}}$ be these coefficients, and substituting $a^2 = \frac {\rho}{1-\rho} b^2$, we have
	\begin{align}
	\label{eq:gaussian-coefficients}
	\left(\frac{c_F} {c_G}\right)^2
	= \frac{\rho b^2 (1+\rho)} {a^2 (1-\rho)^2}
	= \frac{\rho b^2 (1+\rho)} {\frac {\rho}{1-\rho} b^2 (1-\rho)^2}
	= \frac{1+\rho} {1-\rho}
	> 1~.
	\end{align}
	Uniqueness of $\tau=0$ as a finite threshold equilibrium follows.
\end{proof}

Now, consider the stability of these equilibria. 

\begin{corollary}
	In the Gaussian model under DG, the threshold equilibrium $\tau = 0$ is stable, while the uninformative equilibria $\pm \infty$ are unstable.
\end{corollary}

\begin{proof}
	As depicted in Figure~\ref{fig:gaussian-F-G-difference}, $G(x)-F(x)$ is strictly decreasing at $0$.
	Formally, $G'(x) - F'(x) = c_G \phi(c_G x) - c_F \phi(c_F x)$, so $G'(0) - F'(0) = (c_G - c_F) \phi(0) < 0$ (since $c_G < c_F$ by~\eqref{eq:gaussian-coefficients}).
	Thus, by Theorem~\ref{thm:dg-dynamics}, the equilibrium at $0$ is stable.
	Since $\tau = 0$ is stable, the equilibria $\tau = \pm \infty$ are unstable by topological necessity.
\end{proof}



\subsection{Equilibrium characterization generalization}
As in OA, while we do not expect all settings to behave as symmetrically as the Gaussian model in Section~\ref{subsec:gaussian-intro}, we can still characterize existence and location of finite threshold equilibria based on the shapes of the functions $F$ and $G$.

\begin{condition} \label{cond:dg-interval}
	Let $I = [a, b] \subset \reals$ be an interval. 
	Upon seeing signal $x \leq a$, then relative to the prior, an agent believes it more likely that another signal will be \emph{less} than $x$.
	Meanwhile, upon seeing signal $x \geq b$, then relative to the prior, an agent believes it more likely that another signal will be \emph{greater} than $x$.
  	Formally,
	$x \leq a \implies P(x'\leq x \mid x) > P(x' \leq x)$ and
	$x \geq b \implies P(x'\geq x \mid x) > P(x' \geq x)$.
\end{condition}

\begin{theorem} \label{thm:dg-interval}
	Let the agent signal structure satisfy Condition \ref{cond:dg-interval}, with $\Pr[X' \leq \tau \mid X = x]$ monotone decreasing and continuous in $x$, and $F$ and $G$ continuous.
	Then there exists an equilibrium $\tau^* \in I$.
\end{theorem}

\begin{proof}
	Take the function $H(x) = G(x) - F(x)$. 
	Then  Condition \ref{cond:dg-interval} implies $H(a) = G(a) - F(a) > 0$ and $H(b) = G(b) - F(b) < 0$.
	Since $G$ and $F$ are continuous, $H$ is continuous; thus we can apply the Intermediate Value Theorem to conclude there exists a point $\tau^* \in [a, b]$ such that $H(\tau^*) = 0$. 
	Then $\tau^*$ is an equilibrium by Theorem \ref{thm:dg-equil}.
\end{proof}

If Condition~\ref{cond:dg-interval} holds and $F$ is symmetric about a point (its median) in the interval $I$, as in the Gaussian model in Section~\ref{subsec:gaussian-intro}, it is reasonable to expect that $G$ is symmetric about the same point and an equilibrium occurs at the median.
However, we can also characterize how equilibria change relative to the median in skewed settings.

\begin{proposition} \label{prop:dg-skewed}
	Let the agent signal structure satisfy Condition~\ref{cond:dg-interval}, with $\Pr[X' \leq \tau \mid X = x]$ monotone decreasing and continuous in $x$, and $G$ and $F$ continuous.
	Let $m$ be the median of $F$, and consider running the DG mechanism.
	Then if $\Pr[X' \leq m \mid m] > 1/2$, there exists an equilibrium $\tau^*$ such that $\tau^* > m$.
	Similarly, if $\Pr[X' \leq m \mid m] < 1/2$, there exists an equilibrium $\tau^*$ such that $\tau^* < m$.
\end{proposition}

\begin{proof}
	First, by Theorem \ref{thm:dg-interval}, we know at least one equilibrium $\tau^*$ exists in the interval $I$. 
	Now assume that $\Pr[X' \leq m \mid m] > 1/2$. 
	In other words, since $F(m) = 1/2$, $G(m) > F(m)$. 
	It follows under Condition~\ref{cond:dg-interval} that either $m < a$ or $m \in I$.
	If $m < a$, then by Theorem  \ref{thm:dg-interval} there exists an equilibrium $\tau^*$ in the interval $I$; $\tau^* \geq a > m$, so the result follows.
	Otherwise, consider the case where $m \in I$.
	Then define $H(x) = G(x) - F(x)$. 
	$H(m) > 0$, while under Condition~\ref{cond:dg-interval}, $H(b) < 0$. 
	By the Intermediate Value Theorem, since $H$ is continuous, there exists a point $\tau^* \in [m, b]$ such that $H(\tau^*) = 0$.
	By our characterization in Theorem~\ref{thm:dg-equil}, $\tau^*$ is an equilibrium.

	A symmetric argument holds for when $\Pr[X' \leq m \mid m] < 1/2$. 
\end{proof}

In unimodal, continuous settings, we can often diagnose skewness of a distribution based on the difference between its mean and median. 
In a right-skewed distribution, it is often the case that the bulk of high probability outcomes lie below the median, so that conditioning on the median essentially amplifies the impact of these outcomes.
We would therefore expect $\Pr[X' \leq m \mid m]$ to \emph{increase} relative to 1/2. 
By Proposition~\ref{prop:dg-skewed}, equilibria drift to the right of $F$'s median.
A similar argument holds for left-skewed distributions: there is more high probability mass concentrated to the right of the median, so conditioning on the median should \emph{decrease} the probability of a signal being below the median.
Thus we would expect by Proposition~\ref{prop:dg-skewed} that when $F$ is left-skewed, equilibria drift to the left of $F$'s median.

\subsection{Dynamics generalization}
Under reasonable behavior in the tails, we expect existence of a stable, nontrivial equilibrium, while the uninformative equilibria remain unstable. 
\begin{proposition} \label{prop:dg-stable}
	Let the agent signal structure satisfy Condition~\ref{cond:dg-interval} for interval $I$, and such that $F$ and $G$ are continuous, and differentiable. 
	Then there exists an equilibrium $\tau^* \in I$ which is locally stable, while the equilibria $\pm \infty$ are unstable.
\end{proposition}

\begin{proof}
	We know already from Theorem \ref{thm:dg-interval} that at least one equilibrium exists in the interval $I$. 
	Again, we let $H(x) = G(x) - F(x)$.
	Note first that since $H(a) > 0$ and $H(b) < 0$, by continuity and differentiability $H$ must cross 0 at some point $\tau_0$ with $H'(\tau_0) < 0$. 
	Assume not: then $H(x) \geq 0$ for all $x$, which contradicts the fact that $H(x) < 0$ for $x \geq b$.
	It follows by Theorem \ref{thm:dg-dynamics} that $\tau_0$ is a stable equilibrium.

	\begin{enumerate}
		\item Now, WLOG we pick the equilibrium $\tau_1$ which is closest to $a$. 
		By Condition \ref{cond:dg-interval} and since $H(x)$ is continuous, $H(x) > 0$ for $x < \tau_1$. 
		It follows that $H'(\tau_1) \leq 0$. 
		Then $\tau_1$ is stable, at least on the left (if $H'(\tau_1) = 0$, then $\tau_1$ is stable on the left and unstable on the right).
		\item Pick the equilibrium $\tau_2$ which is closest to $b$ (this could be the same as the $\tau_1$ we chose in the previous step).
		By Condition \ref{cond:dg-interval} and since $H(x)$ is continuous, $H(x) < 0$ for $x > \tau_2$. 
		It follows that $H'(\tau_2) \leq 0$. 
    Then $\tau_2$ is stable, at least on the right (if $H'(\tau_2) = 0$, then $\tau_2$ is stable on the right and unstable on the left).
	\end{enumerate}
	We know by Condition \ref{cond:dg-interval} that $H(x) \neq 0$ for $x \notin [a, b]$, so that no equilibria occur outside the interval $I$ other than $\pm \infty$.
	Since both $\tau_1$ and $\tau_2$ (or just $\tau_1$, if they are the same equilibrium) are stable (at least to the left for $\tau_1$, and right for $\tau_2$), it is topologically necessary that the uninformative equilibria $\pm \infty$ are unstable.
\end{proof}

In natural settings, we would expect Condition \ref{cond:dg-interval} to hold.
Specifically, conditioning on a small enough signal, one would expect the probability of other small signals to increase relative to the prior.
Conditioning on a large signal, one would expect decreased probability of much smaller signals and thus a smaller value of $G$ relative to the prior.
Ultimately, then, when conditioning on a signal leads to local sensitivity, we would expect the dynamics in Proposition \ref{prop:dg-stable} to hold. 
In particular, note that Proposition~\ref{prop:dg-stable} applies to the model in \S~\ref{sec:experiments}, when the distributions are skewed. 
Even under complicated, multi-modal information structures, we expect tail behavior to often still satisfy Condition \ref{cond:dg-interval}, so that there exists a stable equilibrium in the interval.

\section{Omitted Results for RBTS} \label{app:rbts}
In this section we formally analyze the RBTS mechanism under real-valued signals; fully characterize equilibria under monotonicity and continuity conditions; and derive the dynamics of equilibria under the Gaussian model.
We focus on RBTS instead of its predecessor, the Bayesian Truth Serum (BTS), because incentive alignment under BTS requires an arbitrarily large number of agents; meanwhile, RBTS guarantees truthfulness for $n \geq 3$ agents. 

\subsection{Equilibrium characterization}
As a reminder, in RBTS each agent $i$ submits both an information report $r_i \in \{L, H\}$ and also a \emph{prediction} report $p_i \in [0, 1]$.
Let $S(p, r) = 2 p \ones[r = H] - p^2$.
RBTS randomly picks a reference agent $j$ and peer agent $k$, and pays agent $i$
\begin{equation*} 
	M_{\RBTS}((r_i, p_i), (r_j, p_j), (r_k, p_k)) = S(p_j + \delta t(r_i), r_k) + S(p_i, r_k),
\end{equation*}
where $t(L) = -1$, $t(H) = 1$, and $\delta>0$ is chosen by the mechanism.
We refer to the first term in Equation~\ref{eq:rbts-payment} as the \emph{information score}, and the second as the \emph{prediction score}.
Under binary signals and reports, truthfulness over both the information and prediction report forms a strict Bayes-Nash equilibrium.

If we consider a threshold strategy $\tau$ in our real-valued signal model, an agent with signal $x_i$ has belief $p_i(x_i) = \Pr[X' > \tau \mid X = x_i] = 1 -P(\tau; x_i)$.
Since the Brier score $S$ is proper, it immediately follows that agents will report $p_i$ truthfully in a symmetric threshold Bayes-Nash equilibrium.
Thus it is WLOG that we can study incentives in threshold equilibria under only the information score, $M_{\RBTS}(r_i, p_j, r_k) = S(p_j + \delta t(r_i), r_k)$.
Moreover, we can operate under the assumption that the prediction report $p_j$ is truthful, i.e. $p_j = p_j(x_j) = 1 - P(\tau; x_j)$ for agent $j$'s true signal $x_j$.

Assume an agent's signal $x$ satisfies $x \leq \tau$.
Then a threshold equilibrium $\tau$ must satisfy
\begin{align*}
	\E_{x_j, x_k \sim \beta(x)} \left[S \left(p_j(x_j) - \delta, \sigma(x_k) \right) \right] &\geq \E_{x_j, x_k \sim \beta(x)} \left[S \left(p_j(x_j) + \delta, \sigma(x_k) \right) \right] \\
	\E \left[ 2 (p_j(x_j) - \delta) \ones [ \sigma(x_k) = H] - (p_j(x_j) - \delta)^2 \right] &\geq \E \left[  2 (p_j(x_j) + \delta) \ones [ \sigma(x_k) = H] - (p_j(x_j) + \delta)^2 \right] \\
	4 \delta \E \left[ p_j(x_j) - \ones [ \sigma(x_k) = H]  \right] &\geq 0 \\
	\Pr [X' \leq \tau \mid X = x] &\geq \E_{x_j \sim \beta(x)} \left[ \Pr [ X' \leq \tau \mid X_j = x_j] \right]
\end{align*}
where the last line follows from plugging in $p_j(x_j) = 1 - \Pr[X' \leq \tau \mid X_j = x_j]$, and noting that $\E_{x_j, x_k \sim \beta(x)} \ones [ \sigma(x_k) = H] = 1 - \Pr[X' \leq \tau \mid X = x]$.
One can show a symmetric condition if $i$'s signal satisfies $x > \tau$, so that all symmetric Bayes-Nash threshold equilibria $\tau$ are characterized by the following conditions:
\begin{align}
	\forall x \leq \tau, P(\tau; x) &\geq \E_{x' \sim \beta(x)} P(\tau; x'), \label{eq:rbts-equil-1} \\ 
	\forall x > \tau,P(\tau; x) &\leq \E_{x' \sim \beta(x)} P(\tau; x'). \label{eq:rbts-equil-2}
\end{align}
Therefore, unlike the DG mechanism, in RBTS an agent's actions in equilibrium depend on the \emph{expected} conditional probability distribution.
We can immediately see that the uninformative thresholds $\tau = \pm \infty$ still always appear as equilibria: 

\begin{proposition} \label{prop:rbts-infinite-equil}
  $\tau^* = \pm \infty$ are threshold equilibria under RBTS.
\end{proposition}

\begin{proof}
	Let $\tau^* = \infty$.
	Then for all signals $x, x' \in \reals$, $P(\tau^*; x) = 1$ and $P(\tau^*; x') = 1$, so that each side of Equation~\ref{eq:rbts-equil-1} always evaluates to 1 and the condition for $x < \tau^*$ is satisfied.
	The condition of Equation~\ref{eq:rbts-equil-2} is vacuous.
	An analagous argument holds for $\tau^* = -\infty$.
\end{proof}

We also observe necessary and sufficient conditions for finite equilibria $\tau$, relative to the function $G$.
Let $Q(x) = \E_{x' \sim \beta(x)} P(x; x')$, and let $\bar{P}(\tau; x) = \E_{x' \sim \beta(x)} P(\tau; x')$ be the expected conditional probability of another agent reporting $L$ over signal $x$.

\begin{theorem} \label{thm:rbts-equil}
  Let a finite threshold $\tau$ be given and $P(\tau; x)$ be continuous over $x$; also assume $\beta(x)$ is continuous over $x$.
  If $\tau \in \reals$ is threshold equilibrium under the RBTS mechanism then $Q(\tau) = G(\tau)$.
  Conversely, if $Q(\tau) = G(\tau)$ and $P(\tau; x) - \bar P(\tau; x)$ has a single crossing of $0$ from positive to negative then $\tau$ is a threshold equilibrium under the RBTS mechanism.
\end{theorem}

\begin{proof}
	For necessity, first note that since $\beta(x)$ and $P(\tau; x)$ are continuous over $x$, $\bar{P}(\tau; x) = \E_{x' \sim \beta(x)} P(\tau; x')$ is also continuous over $x$ by an application of the Dominated Convergence Theorem.
	Let $f(\tau; x) = P(\tau; x) - \bar{P}(\tau; x)$.
	Since $f(\tau; x)$ is continuous, we must have $\lim_{x \to \tau^+} f(\tau; x) = \lim_{x \to \tau^-} f(\tau; x) = f(\tau; \tau).$
	Further, $\lim_{x \to \tau^-} f(\tau; x) \geq 0$ and $\lim_{x \to \tau^+} f(\tau; x) \leq 0$ by Conditions~\ref{eq:rbts-equil-1} and~\ref{eq:rbts-equil-2}, so $f(\tau; \tau) = 0 = G(\tau) - Q(\tau)$.

	For sufficiency, take some $x > \tau$. By single crossing, $P(\tau; x) \leq \bar P(\tau; x)$.
	Similarly, if $x \leq \tau$, $P(\tau; x) \geq \bar P(\tau; x)$.
	This establishes Equations~\eqref{eq:rbts-equil-1} and~\eqref{eq:rbts-equil-2}, so $\tau$ is a threshold equilibrium.
\end{proof}

While the exact crossing condition has changed from $G(\tau) = F(\tau)$ to $G(\tau) = Q(\tau)$, Theorem~\ref{thm:rbts-equil} paints a similar picture to what we saw for DG.  
Best responses are now more complex to compute, so we do not treat the general dynamics formally, but we can provide a comparison of their behavior in the Gaussian case.

\subsection{Gaussian model}
We return to our Normal example, where each agent receives a noisy version of a Gaussian signal.
We first prove the form of the functions $\bar P(\tau; x)$ and $Q(x)$ in this setting.

\begin{proposition}
  In the Gaussian model under the RBTS mechanism,
  \begin{equation}
		\bar{P}(\tau; x) = \Phi \left( \frac{\tau - \rho^2 x} {\sqrt{b^2(1 + \rho^2)(1 + \rho)} } \right),
	\end{equation}
	and thus,
	\begin{equation}
		Q(x) = \Phi \left( \frac{(1 - \rho^2) \tau} {b\sqrt{(1 + \rho)(1 + \rho^2)} } \right). 
	\end{equation}
\end{proposition}

\begin{proof}
  Note that for a fixed threshold $\tau$,
  $$\bar P(\tau; x) = \E_{x' \sim \Pr[\cdot \mid x]} \Pr [X'' \leq \tau \mid X' = x'] = \E_{X' \sim N(\rho x, b^2(1 + \rho))} \Phi \left(  \frac{\tau - \rho x'}{b \sqrt{1 + \rho}} \right),$$
  and since $X' = \rho x + b \sqrt{1 + \rho} X$ for $X \sim N(0, 1),$
  \begin{align*}
    \bar P(\tau; x) &=  \E_X \Phi \left(  \frac{\tau - \rho(\rho x + b \sqrt{1 + \rho}X)}{b \sqrt{1 + \rho}} \right) \\
    &= \Pr \left[ Y \leq \frac{\tau - \rho(\rho x + b \sqrt{1 + \rho}X)}{b \sqrt{1 + \rho}}  \right] \\
    &= \Pr \left[ Y + \rho X \leq \frac{\tau - \rho^2 x}{b \sqrt{1 + \rho}} \right],
  \end{align*}
  
  where $Y \sim N(0, 1)$.
  It follows that $Y + \rho X \sim N(0, 1 + \rho^2)$, so 
  $$\bar P(\tau; x) = \Phi \left( \frac{\tau - \rho^2 x} {b\sqrt{(1 + \rho^2)(1 + \rho)} } \right).$$
  The definition of $Q(x) = \bar P(x; x)$ follows.
\end{proof}

We can then show that the same equilibria occur in RBTS as the previous mechanisms.

\begin{proposition} \label{prop:gaussian-rbts}
  In the Gaussian model under RBTS, there are three equilibria at $\tau = 0$ and $\pm \infty$.
\end{proposition}

\begin{proof}
  Existence of equilibria at $\pm \infty$ follow as a corollary from~\ref{prop:dg-infinite-equil}.
  Next, note all non-infinite equilibria must satisfy $G(x) = Q(x)$ by Theorem~\ref{thm:rbts-equil}.
  But $G(x)$ and $Q(x)$ both correspond to the standard Normal CDF with different coefficients $c_G = \frac {(1-\rho)} {b \sqrt{1+\rho}}$ and $c_Q = \frac {(1-\rho^2)} {b \sqrt{(1 + \rho)(1+\rho^2)}}$.
  Since $\rho \in (0, 1)$,
  \[
		\left( \frac{c_Q}{c_G} \right)^2 = \frac{(1 - \rho^2)^2}{(1 - \rho)^2(1 + \rho^2)} = \frac{(1 + \rho)^2}{1 + \rho^2} > 1.
  \]
  Thus $\tau = 0$ is the unique point at which $G(\tau) = Q(\tau)$ (see Figure~\ref{fig:gaussian-F-G} for an example with $a = b = 1$), so there are no other potential non-infinite equilibria.

  Now we prove $\tau = 0$ actually is an equilibrium using the sufficiency condition in Theorem~\ref{thm:rbts-equil}.
  Note that when $\tau = 0$, $P(\tau; x)$ and $\bar P(\tau; x)$ are both Normal CDFs over $x$ with negative (and distinct) coefficients.
  One can thus easily check that if $x \leq \tau$, $P(\tau; x) \geq \bar{P}(\tau; x)$.
  Moreover, if $x > \tau$, $P(\tau; x) \leq \bar{P}(\tau; x)$.
  It follows by the sufficiency condition in Theorem~\ref{thm:rbts-equil} that $\tau = 0$ is an equilibrium. 
\end{proof}

We can also derive the exact form of the best response $\hat \tau = c(\rho) \tau$.

Next, we include a proof of the form of the best response in RBTS under the Gaussian model.
\begin{proposition} \label{prop:rbts-normal-br}
  When other agents are playing according to some threshold $\tau$, an agent will best respond with threshold strategy 
  \[
    \hat \tau = \left(\frac{\sqrt{1 + \rho^2} - 1}{\rho (\sqrt{1 + \rho^2} - \rho)}\right) \tau.
  \]
\end{proposition}

\begin{proof}
  Let $\tau$ be the current threshold strategy that all other agents are following.
  By equations~\eqref{eq:rbts-equil-1} and~\eqref{eq:rbts-equil-2}, an agent who receives signals $x$ will best respond with $H$ if $ P(\tau; x) \leq \bar P(\tau; x)$, and $L$ otherwise. 
  In the Gaussian model, note by Proposition~\ref{prop:gaussian-rbts} that
  \begin{align*}
    P(\tau; x) &= \Phi \left( \frac{\tau - \rho x}{b \sqrt{1 + \rho}} \right), \\
    \bar P(\tau; x) &= \Phi \left( \frac{\tau - \rho^2 x}{b\sqrt{(1 + \rho^2)(1 + \rho)}} \right).
  \end{align*}
  Thus $P(\tau; x) = \bar P(\tau; x)$ at a unique point which can be verified with simple algebra to be $x = \hat \tau$.
  One can check for the specified value of $\hat \tau$, $P(\tau; x) \geq \bar P(\tau; x)$ for $x \leq \hat \tau$ and $ P(\tau; x) \leq \bar P(\tau; x)$ for $x \geq \hat \tau$.
  The result follows. 
\end{proof}

Now, note that for any $\tau$, $\hat \tau = c(\rho) \tau$ for $c(\rho) \in (0, 1)$.
Under our dynamics, $\dot \tau = \hat \tau - \tau.$
If $\tau > 0$, then $\hat \tau < \tau$, while if $\tau < 0$, $\hat \tau > \tau$. 
It follows that $\tau = 0$ is stable, while by topological necessity the uninformative equilibria are unstable. 

We observe that the best responses under DG and RBTS are both linear functions of $\tau$ with coefficients $m_{\DG}$, $m_{\RBTS}$ respectively.
Thus we can formally compare convergence of the DG and RBTS mechanisms under the Gaussian model using the best response form derived in Proposition~\ref{prop:rbts-normal-br}.
We find that if both mechanisms begin with the same starting threshold $\tau(0)$, the stable equilibrium at $0$ is reached faster in DG than RBTS.
E.g. in a peer grading setting, graders will more quickly converge to a half-good half-bad quality consensus than RBTS graders, so that the designer may have more time to react and reset expectations in the latter case. 

\begin{proposition}
  In the Gaussian model, for any initial threshold $\tau(0) \neq 0$, convergence to the stable equilibrium $\tau = 0$ is strictly slower in the RBTS mechanism than in the DG mechanism.
\end{proposition}

\begin{proof}
  By Proposition~\ref{prop:rbts-normal-br}, in the RBTS mechanism a best response $\hat \tau$ to all other agents playing according to threshold $\tau$ is $\hat \tau = m_{\RBTS}(\rho)\tau$ for 
  \[m_{\RBTS}(\rho) = \frac{\sqrt{1 + \rho^2} - 1}{\rho (\sqrt{1 + \rho^2} - \rho)}. \]
  Meanwhile, the best response in DG to a fixed threshold $\tau$ satisfies $P(\tau; \hat \tau) = F(\tau)$, or 
  \[ 
    \Phi \left( \frac{\sqrt{\rho} \tau}{a} \right) = \Phi \left( \frac{\tau - \rho \hat \tau}{b \sqrt{1 + \rho}}\right).
  \]
The solution to this equality is easily reached by setting the arguments to $\Phi$ equal to each other.
We end up with $\hat \tau = m_{\DG}(\rho) \tau$ for 
  \[m_{\DG}(\rho) = \frac{a - b\sqrt{\rho(1 + \rho)}}{a \rho}. \]
  Now, one can check that $\frac{m_{\DG}(\rho)}{m_{\RBTS}(\rho)} < 1$ for all $a, b > 0$. 
  It immediately follows that the best response $\hat \tau$ converges to 0 strictly faster under the DG mechanism than under the RBTS mechanism.
\end{proof}

\section{Omitted results for DMI}
\label{app:dmi}

In this section, we characterize equilibria and their dynamics for the DMI mechanism.
Unlike the previous mechanisms, where we could compute the payment for a single task in isolation (albeit requiring multiple tasks from a different agent), DMI is inherently a multi-task mechanism.
As in the original analysis of DMI, we therefore restrict to \emph{consistent} strategies, where the same $\sigma$ is applied to each $X_i$. 
We include a discussion of behavior outside consistent strategies in \S~\ref{subsec:dmi-consistent}.
Even under this restriction, threshold equilibria are much more complex in DMI given the interactions between four different signals, making it unclear how to reason directly from the definition of a threshold equilibrium.
Thus we also restrict all agents to best responding with threshold strategies.

\subsection{Equilibrium Characterization}

Let 
\begin{equation}\label{eq:ex-ante-dmi}
	U_i(\sigma, \sigma') = \E\left[ M_\DMI(\sigma(X_1),\ldots,\sigma(X_4),\sigma'(X_1'),\ldots,\sigma'(X_4')) \right]
\end{equation}
be the (ex-ante) expected utility for playing threshold strategy $\sigma$ given the other agent's strategy $\sigma'$.
\begin{definition}
	A threshold strategy $\sigma: \reals \to \R$ is a {\em equilibrium restricted to threshold strategies} under DMI if for all threshold strategies  $\hat\sigma: \reals \to \R$, $U_i(\sigma,\sigma) \geq U_i(\hat\sigma,\sigma)$.
\end{definition}
This is a weaker condition than being a threshold equilibrium because it rules out deviations to non-threshold strategies, an issue we will return to.  
However, by restricting to this space we can provide an equilibrium characterization in the same spirit as we did for OA and DG.

As described above, the DMI mechanism is designed as an unbiased estimator of the (squared) determinant-based mutual information $\DMI(Y;Y')^2 = (\det U_{Y;Y'})^2$, where
$U_{Y,Y'}$ is the joint distribution matrix of the $\{H,L\}$-valued random variables $Y,Y'$, given by
\begin{align*}
U_{Y,Y'} =
  \begin{bmatrix}
    \Pr[Y=H,Y'=H] &
    \Pr[Y=H,Y'=L]
    \\
    \Pr[Y=L,Y'=H] &
    \Pr[Y=L,Y'=L]
  \end{bmatrix}~.
\end{align*}
This unbiasedness is implied by \citet[Claim 4.4]{kong2024dominantly}, which we paraphrase for the binary-report context:
for any $Y,Y'$, there exists a constant $\alpha > 0$ such that $\E \det \sum_{i\in S} \ones_{Y} \ones_{Y'}^T = \alpha \det U_{Y,Y'}$.

Back to our real-valued setting, recall that $(X_i,X_i')\sim\D$ independently.
Given strategies $\sigma,\sigma'$ with thresholds $\tau,\tau'$, for some $\alpha'>0$ we can rewrite Equation~\eqref{eq:ex-ante-dmi} as $U_i(\sigma,\sigma') = \alpha' h(\tau,\tau')^2$ for
\begin{equation} \label{eq:dmi-expected-score}
  h(\tau,\tau') =
\Pr[X>\tau,X'>\tau']\Pr[X\leq\tau,X'\leq\tau']\strut - \Pr[X>\tau,X'\leq\tau']\Pr[X\leq\tau,X'>\tau']~,
\end{equation}
where $(X,X')\sim\D$.
We are now in a position to analyze equilibria and dynamics.

\subsection{Equilibrium results}
Given the definition of $h$ from Equation~\eqref{eq:dmi-expected-score}, we can show that uninformative strategies $\tau^* = \pm \infty$ remain equilibria in the DMI mechanism.
The proof is immediate, as in both cases one of the two probabilities in each term of Equation~\eqref{eq:dmi-expected-score} is zero independent of $\tau$.

\begin{proposition} \label{prop:dmi-infinite-equil}
	$\tau^* = \pm \infty$ are both always equilibria restricted to threshold strategies under DMI.
\end{proposition}

We can also provide necessary and sufficient conditions for a finite threshold equilibrium for DMI.  The theorem relies on several regularity conditions that rule out corner cases where the first order condition is trivially satisfied for spurious reasons. 

\begin{theorem} \label{thm:dmi-equil}
	Let finite threshold $\tau$ be given, $h(x,\tau)$ be differentiable in $x$ at $x = \tau$, $h(\tau,\tau)\neq 0$, and $f(\tau) > 0$, where $f$ is the continuous density of $F$.
	If $\tau$ is an equilibrium restricted to threshold strategies under the DMI mechanism then $G(\tau) = F(\tau)$.
	Conversely, if $G(\tau) = F(\tau)$ and $h(x,\tau)$ is non-negative and strictly single-peaked 
	then $\tau$ is an equilibrium restricted to threshold strategies under the DMI mechanism.
\end{theorem} 

\begin{proof}
	For necessity, assume that $\tau^*$ is an equilibrium restricted to threshold strategies, so $f(\tau,\tau^*)$ is maximized at $\tau = \tau^*$.
	The first order condition for optimality is
\begin{align}
  \label{eq:1}
  0 &= \frac d {d\tau} \alpha' h(\tau,\tau^*)^2
    = 2 h(\tau,\tau^*) \alpha' f(\tau) \left( \Pr[X' > \tau^*] - \Pr[ X' > \tau^* \mid X=\tau] \right)~
\end{align}
	By assumption $h(\tau^*,\tau^*)\neq 0$ and $f(\tau^*) > 0$. Thus $\Pr[X' > \tau^*] - \Pr[ X' > \tau^* \mid X=\tau^*] = 0$, or equivalently $G(\tau^*) - F(\tau^*) = 0$

	For sufficiency, assume $G(\tau) = F(\tau)$.  By the above, $\tau$ satisfies the first order condition for equilibrium.  If $h(x,\tau)$ is non-negative and strictly single-peaked as a function of $x$ then $U_i$ is also strictly single-peaked as a function of $x$ (i.e. it is strictly increasing and then strictly decreasing), so the unique point satisfying the FOC is the global maximum.
\end{proof}

As discussed in the main text, Theorem~\ref{thm:dmi-equil} shows that DMI has the same necessary condition for equilibrium as DG. 
Thus, beyond the slightly different sufficient condition for equilibrium, its equilibria can be found at the same thresholds: those where $G(\tau) = F(\tau)$.
Moreover, with the appropriate sufficient condition, our results about the dynamics (Theorem~\ref{thm:dg-dynamics}) for DG immediately apply to DMI as well.
In particular, the Gaussian case has the same equilibra with the same stability.

\subsection{Beyond consistent threshold strategies} \label{subsec:dmi-consistent}
The above analysis is restricted to our weaker notion of equilibrium in the space of consistent threshold strategies.
It is unclear whether these remain equilibria if agents can use any consistent strategy, or if agents can use any four threshold strategies.
Even more broadly, one could consider \emph{joint-task} strategies $\sigma:\reals^4\to\{H,L\}^4$, which map all four signals simultaneously to all four reports.
While not shown in \citet{kong2020dominantly,kong2024dominantly}, in the binary report setting the DMI mechanism turns out to satisfy the even stronger property that truthfulness is an equilibrium among all joint-task strategies.

Specifically, we can show in the binary signal, binary report model that truthfulness remains a best response even when the other agent deviates from truthful, but crucially, still plays a \emph{consistent} strategy, which here means strategies that (a) map each signal to a distribution over reports, and (b) use the same such map for every task.
\begin{definition}
  A strategy $\sigma:\{H,L\}^T\to\Delta(\{H,L\}^T)$ is \emph{consistent} if $\sigma(s_{1..T})(r_{1..T}) = \prod_{t=1}^T \hat\sigma(s_t)(r_t)$ for some $\hat\sigma:\{H,L\}\to\Delta(\{H,L\})$.
\end{definition}

For example, the truthful strategy $\sigma_{\text{true}} : s \mapsto \delta_s$, where $\delta_s$ is the point distribution on $s$, is a consistent strategy, with $\hat\sigma:s\mapsto\delta_s$ for $s\in\{H,L\}$.

We are able to prove that truthfulness maximizes the expected payoff of $M_\DMI$ when all other agents play consistent strategies.
As a corollary, truthfulness is a Bayes Nash equilibrium even over the larger space of strategies, where one can choose all $T$ reports simultaneously after looking at all $T$ signals.
We leave the proof to \S~\ref{subsec:dmi-proof}.

\begin{theorem}
  \label{thm:dmi-joint-strategy}
  In the binary signal model, suppose all agents $j\neq i$ play consistent strategies.
  Then $\sigma_{\text{true}}$ maximizes agent $i$'s expected payment under $M_\DMI$.
\end{theorem}

\begin{corollary}
  The truthful strategy $\sigma_{\text{true}}$ is a Bayes Nash equilibrium of the DMI mechanism.
\end{corollary}

We can then show that truthfulness is \emph{not} an equilibrium for real-valued reports: in many cases the best response flips one or two of the truthful reports to increase the chance of a nonzero payoff in $M_\DMI$.
To get some intuition, first consider the binary signal model and an agent who receives signals $H,L,L,L$.
From the form of the mechanism, if they report truthfully, their score will be zero deterministically, since $M_{34}$ will have rank 1.
Thus, they might consider deviating, say flipping the last $L$ to $H$.
It turns out doing so keeps their expected score at zero, since the other agent is equally likely to align or misalign with their reports.
The consistency of the other agent's strategies is crucial for this last fact; given any slight difference in their strategy for tasks 3 and 4, the first agent would deviate from truthfulness.

It is precisely this sort of imbalance that easily arises in the real-valued signal model.
Consider the Gaussian model with $a=b=1$ and $\rho=1/2$, at the equilibrium threshold $\tau=0$.
Suppose the agent receives $x_1 = 10$, $x_2 = -10$, $x_3 = -10$, and $x_4=-1$.
Reporting truthfully, $(H,L,L,L)$, would yield payoff 0 deterministically.
But the agent may note that flipping the last $L$ to $H$ is ``safe'' in the sense that it is extremely unlikely that the other agent will differ in their reports on the first 3 tasks.
So either the other agent reports $r_4 = L$, in which case the payment is zero for both agents, or $r_4 = H$, in which case both agents receive 1.
As there is a reasonable chance of the latter case, about $0.2819$, that is roughly the expected score.

Thus, given real-valued signals $(10,-10,-10,-1)$, the optimal report is $(H,L,L,H)$, which is not ``truthful'' for the prescribed threshold $\tau=0$.
Hence DMI is no longer truthful in this joint-task sense, even at the equilibrium threshold $\tau=0$.

\subsection{Proof of Theorem~\ref{thm:dmi-joint-strategy}} \label{subsec:dmi-proof}

The key observation needed for Theorem~\ref{thm:dmi-joint-strategy} is the following characterization of the possible expected values of each determinant term in the definition of $M_\DMI$.
\begin{lemma}
  \label{lem:dmi-det-formula}
  Suppose agent $j$ plays a consistent strategy.
  Let $R'_t$ be the random variable on $\{H,L\}$ resulting from agent $j$'s strategy, i.e., $R_t' \sim \hat\sigma(S_t')$ for each task $t$.
  Let $p(s_t) = \Pr[R_t' = H \mid S_t = s_t]$, which by assumption is independent of $t$.
  Then for agent $i$,
  \begin{align*}
    \E[\det M_{12} \mid S_1,S_2]
    \;=\;
    \begin{cases}
      0, 
      & \text{if } r_1 = r_2 \text{ or } s_1 = s_2,\\
      p(H) - p(L), 
      & \text{if } (r_1, r_2) = (s_1,s_2),\\
      p(L) - p(H), 
      & \text{if } (r_1, r_2) = (s_2,s_1).
    \end{cases}
  \end{align*}
  where as before
  \[
    M_{12} 
    \;=\; 
    \ones{r_1}\,\ones{R_1'}^\top 
    \;+\;
    \ones{r_2}\,\ones{R_2'}^\top~.
  \]
\end{lemma}
\begin{proof}
  If $r_1=r_2$, then $M_{12}$ has rank one and $\det M_{12} = 0$ deterministically.
  Suppose next $(r_1,r_2) = (H,L)$.
  We have
  \[
    \det M_{12}
    \;=\;
    \begin{cases}
      +1, &\text{if } (R_1',R_2') = (H,L),\\
      -1, &\text{if } (R_1',R_2') = (L,H),\\
      0,  &\text{otherwise}.
    \end{cases}
  \]
  As $R_1',R_2'$ are independent, we have
  \begin{align*}
    \Pr\bigl[(R_1',R_2') = (H,L)  \mid S_1,S_2\bigr] 
    &= p(s_1)\,\bigl(1 - p(s_2)\bigr)
    \\
    \Pr\bigl[(R_1',R_2') = (L,H)  \mid S_1,S_2\bigr] 
    &= \bigl(1 - p(s_1)\bigr)\,p(s_2)~.
  \end{align*}
  We thus have
  \begin{align*}
    \E[\det M_{12} \mid S_1,S_2]
    &=
      p(s_1) \,\bigl(1-p(s_2)\bigr)
      -
      \bigl(1-p(s_1)\bigr)\,p(s_2)
    \\
    &=
      p(s_1) - p(s_2)~.
  \end{align*}
  The same argument gives $\E[\det M_{12} \mid S_1,S_2] =
  p(s_2) - p(s_1)$ when $(r_1,r_2)=(L,H)$.

  We have now established
  \begin{align*}
    \E[\det M_{12} \mid S_1,S_2]
    \;=\;
    \begin{cases}
      0, 
      & \text{if } r_1 = r_2,\\
      p(s_1) - p(s_2), 
      & \text{if } (r_1, r_2) = (H,L),\\
      p(s_2) - p(s_1), 
      & \text{if } (r_1, r_2) = (L,H).
    \end{cases}
  \end{align*}
  To finish the proof, consider the cases $s_1=s_2$, $(s_1,s_2)=(H,L)$, and $(s_1,s_2)=(L,H)$.
  From the above, we obtain 0 in the first case.
  In the second, we have $p(H)-p(L)$ if $(r_1,r_2) = (s_1,s_2)$ and its negation when $(r_1,r_2) = (s_2,s_1)$.
  In the third, we again have $p(H)-p(L)$ if $(r_1,r_2) = (s_1,s_2)$ and its negation when $(r_1,r_2) = (s_2,s_1)$.

\end{proof}

Clearly Lemma~\ref{lem:dmi-det-formula} also applies to $M_{34}$.
And since $\det M_{12}$ and $\det M_{34}$ are independent when conditioned on $S_{1..T}$, the proof of Theorem~\ref{thm:dmi-joint-strategy} follows.

\begin{proof}[Proof of Theorem~\ref{thm:dmi-joint-strategy}]
  Let $s_{1..T}$ be the signal realization for agent $i$, and consider any report $r_{1..T}$.
  We have
  \begin{align*}
    \E[ M_\DMI(r_{1..T},R'_{1..T}) \mid S_{1..T}]
    &
      = \E[\det M_{12} \mid S_1,S_2] \E[\det M_{34} \mid S_3,S_4]~.
  \end{align*}
  Applying Lemma~\ref{lem:dmi-det-formula} to both $M_{12}$ and $M_{34}$, observe that if $s_1=s_2$ or $s_3=s_4$ we then have $\E[ M_\DMI(r_{1..T},R'_{1..T}) \mid S_{1..T}] = 0$ regardless of $r$.
  In particular, $\sigma_{\text{true}}$ maximizes the expected payoff.
  
  Now suppose $s_1\neq s_2$, $s_3\neq s_4$.
  Again we have a zero expectation if $r_1=r_2$ or $r_3=r_4$.
  Otherwise, for each matrix, we have two cases: $r$ matches $s$ or not.
  In particular,
  \begin{align*}
    \E[ M_\DMI(r_{1..T},R'_{1..T}) \mid S_{1..T}] =
    \begin{cases}
      (p(H) - p(L))^2 & (r_1,r_2)=(s_1,s_2), (r_3,r_4)=(s_3,s_4)
      \\
      -(p(H) - p(L))^2 & (r_1,r_2)\neq(s_1,s_2), (r_3,r_4)=(s_3,s_4)
      \\
      -(p(H) - p(L))^2 & (r_1,r_2)=(s_1,s_2), (r_3,r_4)\neq(s_3,s_4)
      \\
      (p(H) - p(L))^2 & (r_1,r_2)=(s_1,s_2), (r_3,r_4)=(s_3,s_4)
    \end{cases}
  \end{align*}
  In summary, the highest expected payoff in this case is $(p(H)-p(L))^2$, which is achieved by $\sigma_{\text{true}}$ and the permutation strategy which swaps $H$ and $L$.
\end{proof}

\section{Omitted Experiments} \label{appendix:experiments}

We include bifurcation figures for the four-component Gaussian mixture case within the DG and OA mechanisms (Figure~\ref{fig:multimodal-extra-dg-oa}), as well as bifurcation figures for RBTS (Figure~\ref{fig:extra-rbts}).
Again, we vary the precision of each component in the Gaussian mixture and study how the stability and number of equilibria change across mechanisms. 
\begin{figure}[ht]
\centering
\begin{subfigure}[t]{\textwidth}
  \centering
  \begin{subfigure}[t]{0.45\textwidth}
      \includegraphics[width=\textwidth]{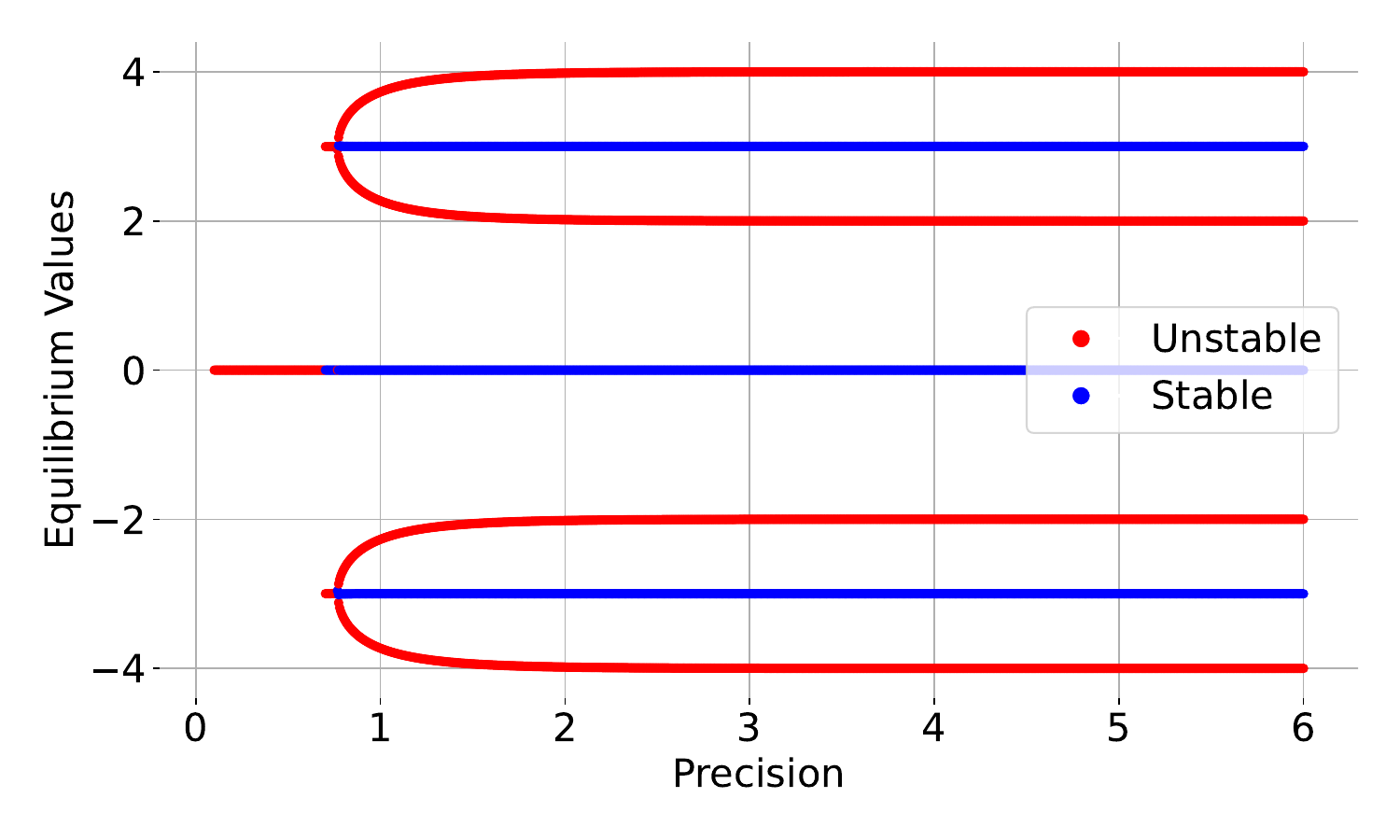}
  \end{subfigure}
  \hspace{0.5cm}
  \begin{subfigure}[t]{0.45\textwidth}
      \includegraphics[width=\textwidth]{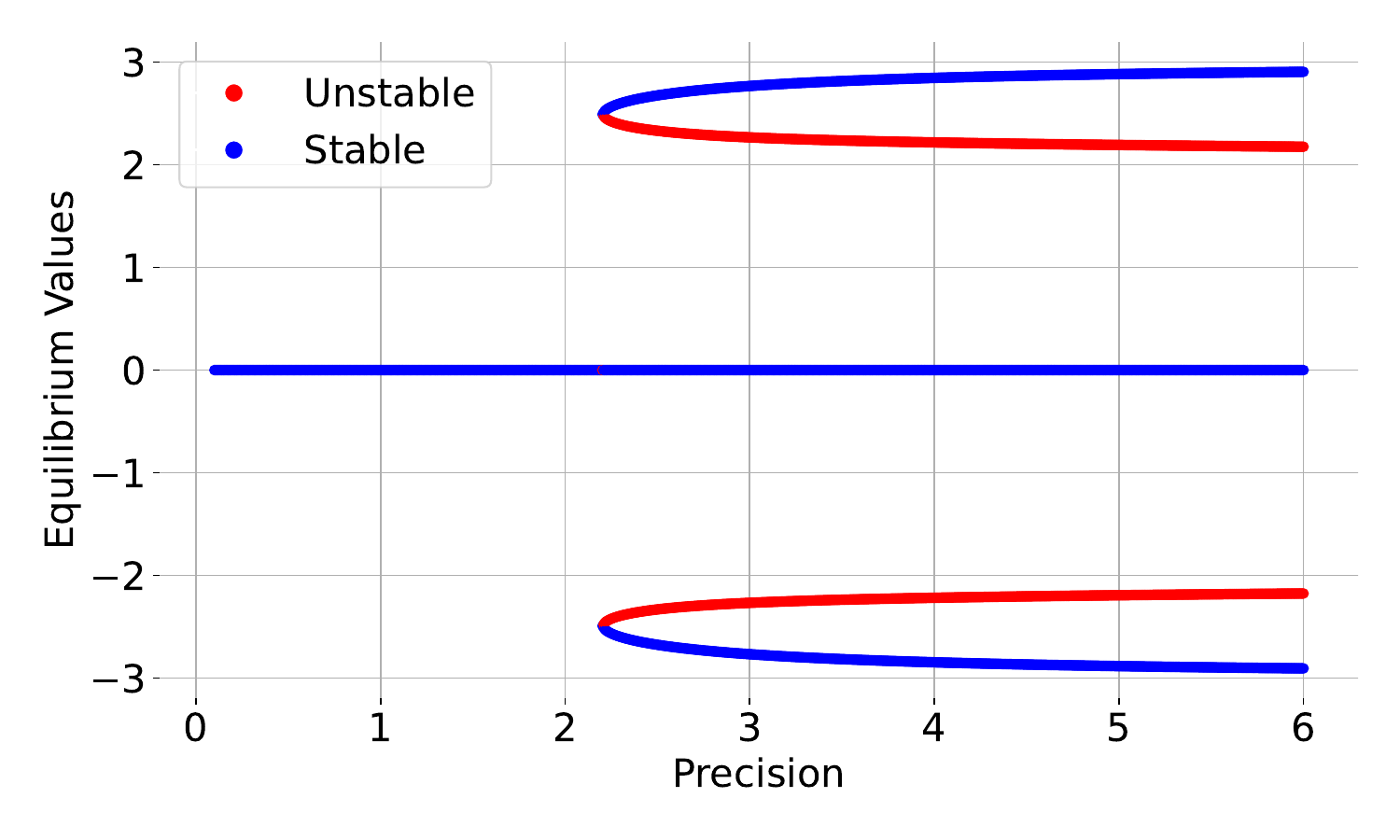}
  \end{subfigure}
  \caption{Bifurcation graphs for OA (left) and DG (right) in the setting of a four-component Gaussian mixture with means ($-4,-2,2,4$).
Following a pattern, we have seven equilibria in OA and five in DG for large enough precision. 
\label{fig:multimodal-extra-dg-oa}
}
\end{subfigure}
\begin{subfigure}{\textwidth}
  \centering
  \begin{subfigure}[t]{0.45\textwidth}
      \includegraphics[width=\textwidth]{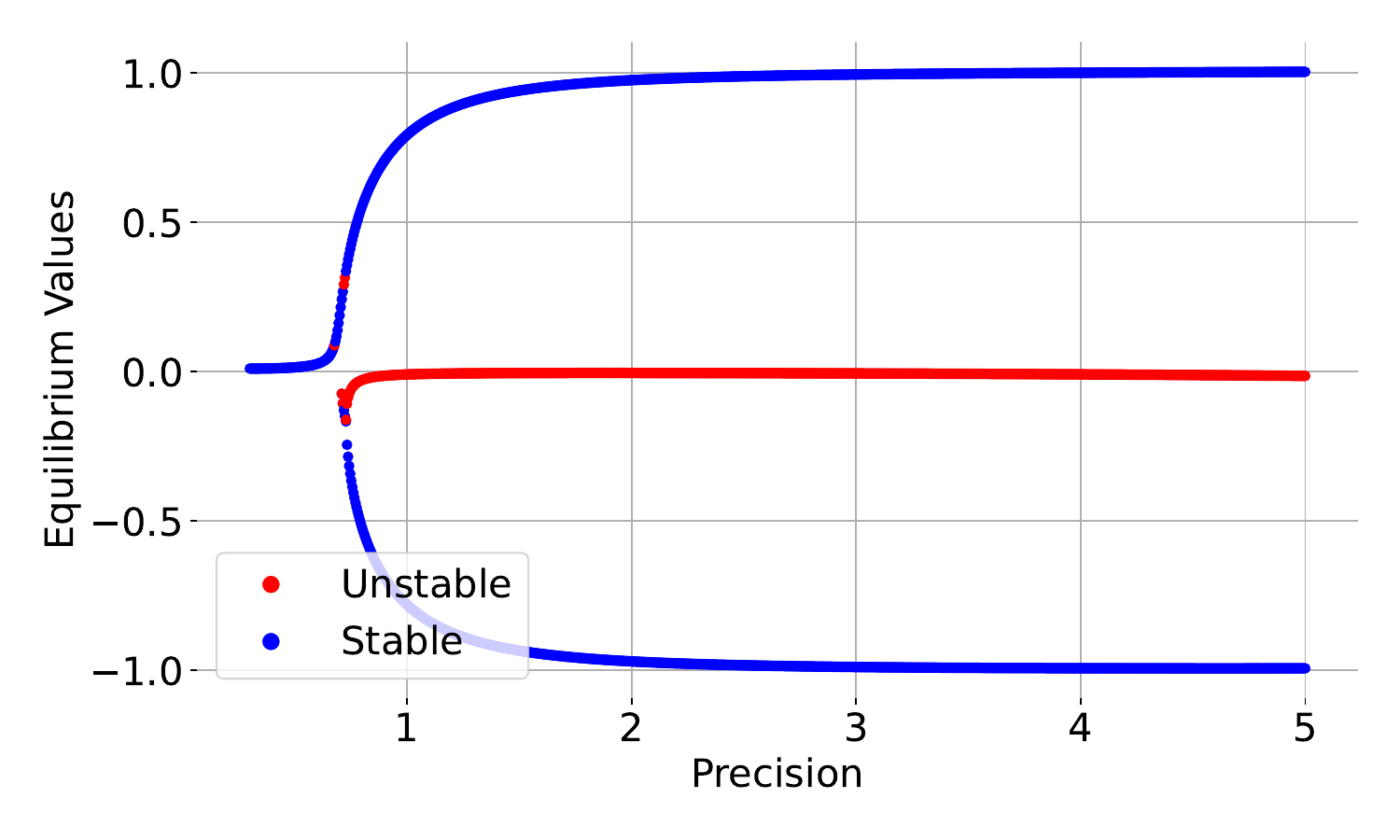}
  \end{subfigure}
  \hspace{0.5cm}
  \begin{subfigure}[t]{0.45\textwidth}
      \includegraphics[width=\textwidth]{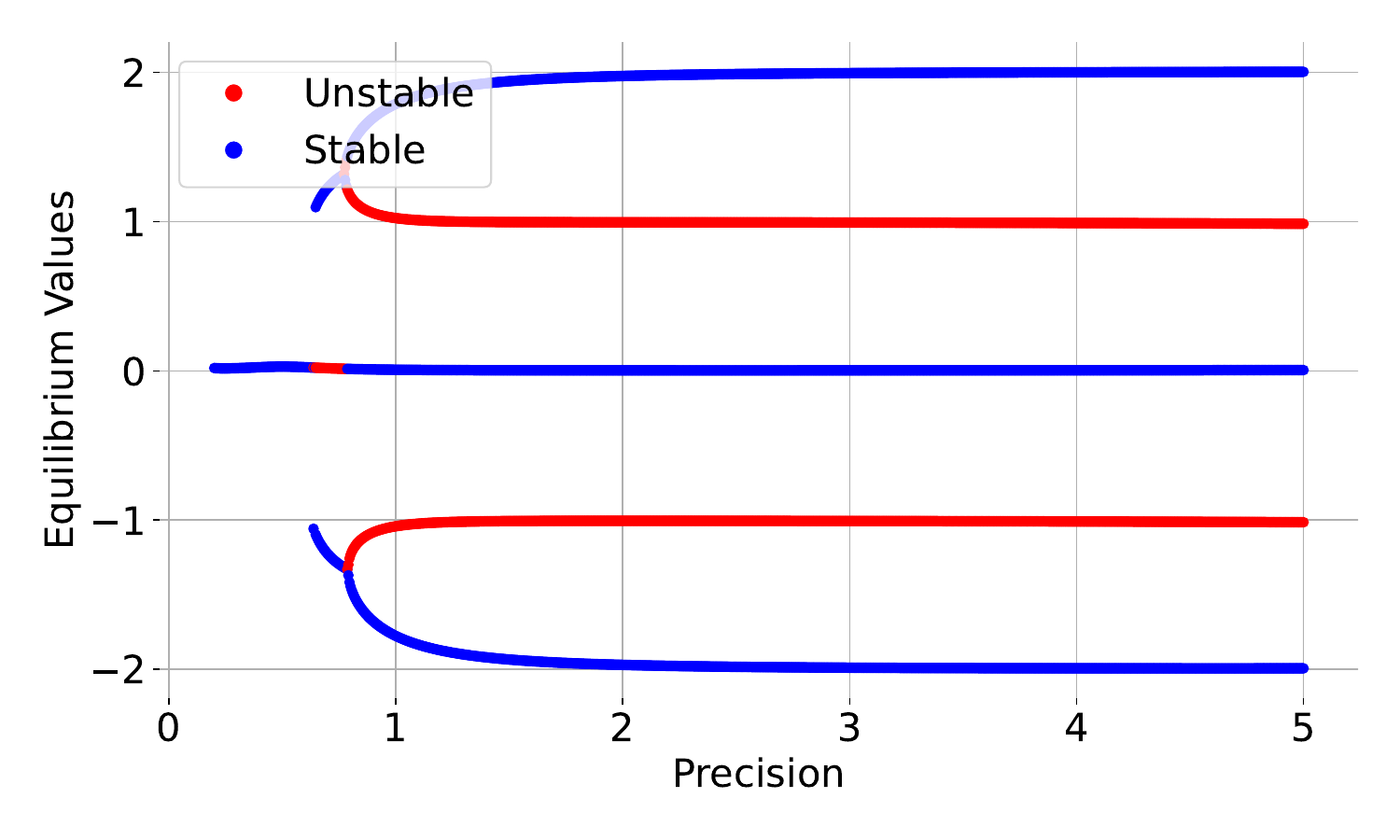}
  \end{subfigure}
  \caption{Bifurcation graphs for RBTS in the setting of a three (left, with means (-2, 0, 2)) and four-component (right, with means (-3,-1,1,3)) Gaussian mixture.
  Note we observe the same stability and number of equilibria as in DG, except that bifurcation occurs sooner.
  \label{fig:extra-rbts}
}
\end{subfigure}
\end{figure}

\section{Omitted Results for Correlated Agreement} \label{app:ca}

\subsection{Equilibrium characterization}

We give the following general characterization of equilibria for Correlated Agreement with $m \geq 3$ reports below, alongside a generalized monotonicity condition. 
The proof follows in the same way as our binary-report results, equating utilities of adjacent reports to each other at their corresponding thresholds. 
Throughoutm we denote $P_k(\tau; x) = \Pr[\tau_{k-1} < X' \leq \tau_k \mid X = x]$ for any $k \in \{1,2,\dots,m\}$.

We begin by noting that by Equation~\ref{eq:dg-payment}, the interim expected utility for playing strategy $\sigma_i$ depends on $\sigma_j$ both directly in the bonus for agreement and indirectly through $\pi_L$:
\begin{equation} \label{eq:dg-util}
	U_i(\sigma_i, \sigma, x) = \E_{x' \sim \beta(x)} \ones[\sigma_i(x) = \sigma(x')]  - \pi_{\sigma_i(x)}.
\end{equation}

Thus the interim utility of playing $\sigma_i(x) = k$ when receiving signal $x$ against symmetric threshold strategy $\tau$ is 
\begin{equation} \label{eq:ca-utilities}
	U_i(\sigma_i, \sigma^{\tau}, x) = \sum_{\ell \in [m]} P_{\ell}(\tau; x) \Delta(\ell, k) - \pi_k.
\end{equation}
It follows that $\tau$ is a symmetric Bayes-Nash threshold equilibrium when
\begin{equation} \label{eq:ca-k-threshold-equilibrium}
	\forall k \in \{1,\dots,m\}, \; \forall x \in (\tau_{k-1}, \tau_k], \; k \in \arg\max_{k \in \{1, \dots, m\}} \sum_{\ell} P_{\ell}(\tau; x) \Delta(\ell, k) - \pi_k.
\end{equation}

\begin{condition}{($k$-sequential monotonicity)} \label{cond:seq-envelope-ca}
	We say a signal distribution satisfies $k$-sequential monotonicity at a threshold $\tau$ if the upper envelope $\max_k \sum_{\ell} P_k(\tau; x) \Delta(\ell, k) - \pi_k$ is composed of the functions $k=1,2,\dots,m$ appearing sequentially in increasing order of $x$.
\end{condition}

\begin{theorem} \label{thm:ca-report-mapping}
	For a fixed $m$ and score matrix $\Delta$, let finite threshold $\tau \in \reals^{m-1}$ be given and $P_k(\tau;x)$ be continuous over $x$ for each $k = \{0,1,\dots,m\}$.
  Let 
  \[
	S(\tau; \hat \tau) \coloneqq
	\begin{bmatrix}
		\sum_{\ell} P_{\ell}(\tau; \hat \tau_1) \left( \Delta(\ell, 1) - \Delta(\ell, 2) \right) + \pi_2 - \pi_1 \\
		\sum_{\ell} P_{\ell}(\tau; \hat \tau_2) \left( \Delta(\ell, 2) - \Delta(\ell, 3) \right) + \pi_3 - \pi_2 \\ 
		\dots \\
    \sum_{\ell} P_{\ell}(\tau; \hat \tau_{m-1}) \left( \Delta(\ell, m-1) - \Delta(\ell, m) \right) + \pi_m - \pi_{m-1},
	\end{bmatrix};
	\]
  in an overload of notation, let $S(\tau) = S(\tau; \tau)$. 
  Then if $\tau$ is an equilibrium under CA, $S(\tau) = \bm{0}$.
	Conversely, if Condition~\ref{cond:seq-envelope-ca} holds at $\tau$, then if $S(\tau) = \bm{0}$, $\tau$ is an equilibrium.
\end{theorem}

\begin{proof}
    We first prove necessity.
    Assume $\tau$ is a threshold equilibrium, and fix $k \in \{1,\dots,m-1\}$.
    Then since $P_k(\tau; x)$ is continuous over $x$ for each $k$, note the function 
    \[f_k(x) = \sum_{\ell} P_{\ell}(\tau; x) \Delta(\ell, k) - \pi_k \]
    is also continuous. 
    It follows $\lim_{x \to \tau_k^-} f_k(x) - f_{k+1}(x) = \lim_{x \to \tau_k^+} f_k(x) - f_{k+1}(x)$.
    But by Equation~\ref{eq:ca-k-threshold-equilibrium}, $\lim_{x \to \tau_k^-} f_k(x) - f_{k+1}(x) \geq 0$ and $\lim_{x \to \tau_k^+} f_k(x) - f_{k+1}(x) \leq 0,$ so we must have $f_k(x) - f_{k+1}(x) = 0.$

  For sufficiency, assume $S(\tau) = \bm{0}.$
	Since each $P_k(\tau; x)$ is continuous, the global maximum $\max_k f_k(x)$ is also continuous and switches from $k$ to $k+1$ when  $f_k(x) = f_{k+1}(x)$. 
	Because $\tau$ solves $S(\tau; \tau) = \bm{0}$, it must coincide with the transition points of the upper envelope. 
	(If it did not, the sequence of maximizing functions would be out of order, violating Condition~\ref{cond:seq-envelope-ca}.)
\end{proof}

We note by the form of utilities indicated by Equation~\ref{eq:ca-utilities} that under Condition~\ref{cond:seq-envelope-ca}, it immediately holds a best response to a threshold strategy $\tau \in \reals^{m-1}$ is also a threshold strategy $\hat \tau \in \reals^{m-1}$.

\subsection{Tridiagonal matrix with three reports} \label{subsec:tridiagonal-m-3}

In this section, we demonstrate theoretically that attempting using CA for $m=3$ reports with the tridiagonal matrix $\Delta_T$ leads to no finite threshold equilibria outside the collapse to a binary threshold of $0$.
Intuitively, note here that $S[2,1] = S[2,2] = S[2,3]$.
That is, if an agent reports the middle report `2' they are scored against all other reports and the expected utility when receiving signal $x$ is $\sum_{\ell \in [m]} P_{\ell}(\tau; x) - \Pr[\tau_{\ell} < X \leq \tau_{\ell+1}] = 0.$
Thus agents are incentivized to report `1' or `3' instead, effectively pushing thresholds toward zero so that the tridiagonal matrix is not useful. 

\begin{lemma} \label{lem:ca-m-3}
  Let $m = 3$ and consider CA with score matrix $\Delta_T$.
  Then under the Gaussian model, an equilibrium $\tau \in \reals^2$ must take the symmetric form $(-x, x)$ for some $x > 0$.
\end{lemma}

\begin{proof}
  Consider any non-symmetric vector $\tau = (\tau_1, \tau_2)$, i.e. $\tau_1 < \tau_2$ and $\tau_1 \neq -\tau_2$. 
  Then for $\tau$ to be an equilibrium, it is necessary by Theorem~\ref{thm:ca-report-mapping} that (1) $\Pr[X' > \tau_2 \mid X = \tau_1] = \Pr[X' > \tau_2]$ and (2) $\Pr[X' \leq \tau_1 \mid X = \tau_2] = \Pr[X' \leq \tau_1]$. 
  In other words, for $\sigma^2 = a^2 + b^2$ we must have 
  \begin{align*}
    \Phi \left( \frac{\tau_2 - \rho \tau_1}{ \sigma \sqrt{1 - \rho^2}}\right) &= \Phi\left(\frac{\tau_2}{\sigma} \right), \\
    \Phi \left( \frac{\tau_1 - \rho \tau_2}{ \sigma \sqrt{1 - \rho^2}}\right) &= \Phi\left(\frac{\tau_1}{\sigma} \right).
  \end{align*}

  Since $\Phi$ is strictly increasing, we can equate the arguments of each side to each other; simplifying with $k = \sqrt{1 - \rho^2}$, we end up with
    \begin{align*}
    \rho \tau_1 &= \tau_2 (1-k), \\
    \rho \tau_2 &= \tau_1 (1-k).
  \end{align*}

  If we multiply these equations together, we have $\rho^2 \tau_1 \tau_2 = \tau_1 \tau_2 (1 - k)^2$.
  If $\tau_1 = 0$, the only solution to this equation is $\tau_2 = 0$, and vice versa if $\tau_2 = 0$.
  Thus WLOG assume $\tau_1, \tau_2 \neq 0$; then we have $\rho^2 = (1 - \sqrt{1 - \rho^2})^2$. 
  Letting $y = \sqrt{1 - \rho^2}$, this simplifies to 
  $y ( y - 1 ) = 0$. 
  The only solutions to this equation are $ y = 0, 1$: if $y = 0$, this implies $\rho = 1$, which contradicts that $b \neq 0$. 
  If $y = 1$, this implies $\rho = 0$, which contradicts that $a \neq 0$. 
  Thus the statement holds.

\end{proof}

\begin{theorem}
  Let $m = 3$ and consider CA with score matrix $\Delta_T$.
  Then under the Gaussian model, there does not exist a stable threshold equilibrium outside the collapsed threshold at $0$.
\end{theorem}

\begin{proof}
  By Lemma~\ref{lem:ca-m-3}, a nontrivial threshold equilibrium $\tau \in \reals^2$ must take the form $\tau = (-x, x)$ for any $x > 0$. 
  Consider any such $\tau$.
  Plugging in $\Delta_T$ to Theorem~\ref{thm:ca-report-mapping}, $\tau$ is an equilibrium if and only if $P_3(\tau; -x) = \pi_3$ and $P_1(\tau; x) = \pi_1$.
  We must then have $f(x) = \Pr[X' > x \mid X = -x] - \Pr[X' > x] = 0$. 
  However, note for any $x$ that 
  \[f(x) =  \Phi\left( - c_1 x \right) - \Phi\left( - c_2 x \right)\]
  for $\sigma = \sqrt{a^2 + b^2}$, $c_1 = \frac{1}{\sigma} \frac{1 + \rho} {\sqrt{1 - \rho^2}}$, and $c_2 = \frac {1} {\sigma}$. 
  Since the Normal CDF is strictly increasing, the statement only holds when $c_1 = c_2$, which occurs when $\rho = 0$. 
  Since $a > 0$, this is impossible; the statement follows. 
\end{proof}

\subsection{Simulations for larger report spaces} \label{subsec:larger-m}

In this section, we include our numerical calculations of equilibria for CA when $m > 4$ to verify that the patterns observed in the main text continue.
In Figure~\ref{fig:larger-m-ca}, we observe that the individual signal precision required for a nontrivial threshold equilibrium to emerge under the diagonal score matrix increases roughly quadratically over $m$.

Meanwhile, we see in Figure~\ref{fig:larger-m-tri-ca} that for any $m > 3$, a nontrivial threshold equilibrium of dimension $m-1$ emerges under the tridiagonal score matrix.
As mentioned in the main text, we note a tradeoff also emerges for larger $m$ in this setting: as the precision approaches $0$, the outer thresholds move toward $\pm \infty$.
This happens because agents with extreme signals have relatively flat posterior distributions. 
Since payments reward positive correlation with neighboring bins, these agents expect higher payoffs from reporting bins closer to the center, where each bin has neighbors on both sides, over reporting edge bins, which have fewer correlated neighbors.
Compared to the points of precision where equilibria collapse under $\Delta_I$, however, these precision values remain extremely low across $m$.
Thus the tridiagonal matrix is still more useful for a larger region of lower precision values.

Conversely, in the limit of large precision values, pairs of adjacent thresholds collapse toward the center, with an effective equilibrium of dimension $m/2 - 1$ if $m$ is even, and $(m-1)/2$ if $m$ is odd.
This phenomenon occurs because agents are incentivized to push their reports \emph{away} from the center as their posteriors concentrate around their signals. 
However, we note in all our figures and simulations that this collapse only occurs for extremely high and thus unreasonable precision values relative to the base signal's noise (e.g., $1/b^2 > 100$).

\begin{figure}[htbp]
    \centering
    \includegraphics[width=0.4\textwidth]{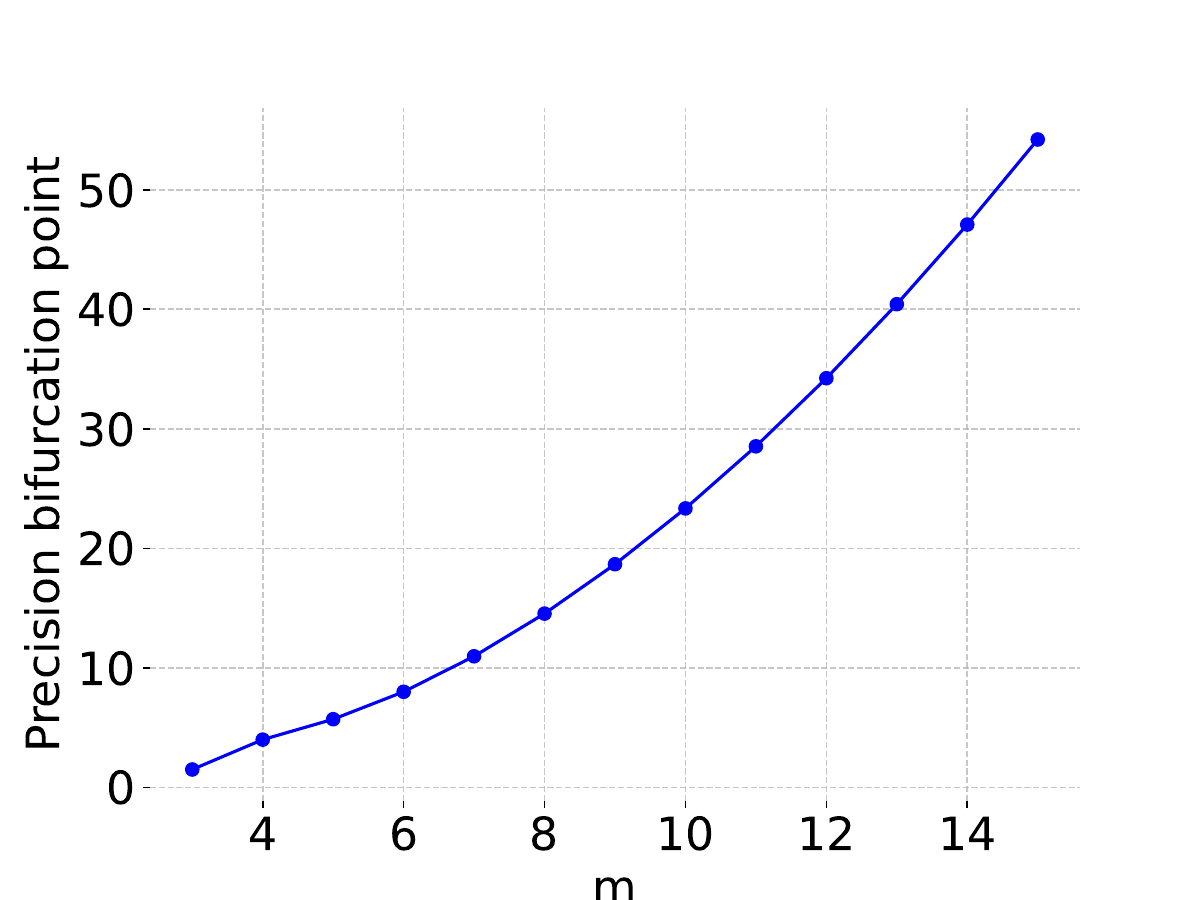}
    \caption{
      A plot of $m$ vs. the bifurcation point of precision $1/b^2$ at which a stable equilibrium $\tau \in \reals^{m-1}$ emerges under CA with diagonal score matrix $\Delta_I$.
      Here we performed a binary search, numerically solving for equilibrium and checking its stability by estimating the Jacobian of the system using finite differences over each coordinate (relative to the ex-ante objective).
      We observe a roughly quadratic increasing relationship.
      If the designer wants more choice in equilibria while using the score matrix $\Delta_I$, then, they must be confident that agents are receiving increasingly low-noise signals.
    }
    \label{fig:larger-m-ca}
\end{figure}

\begin{figure}[htbp]
    \centering
    \begin{subfigure}[t]{0.45\textwidth}
        \includegraphics[width=\textwidth]{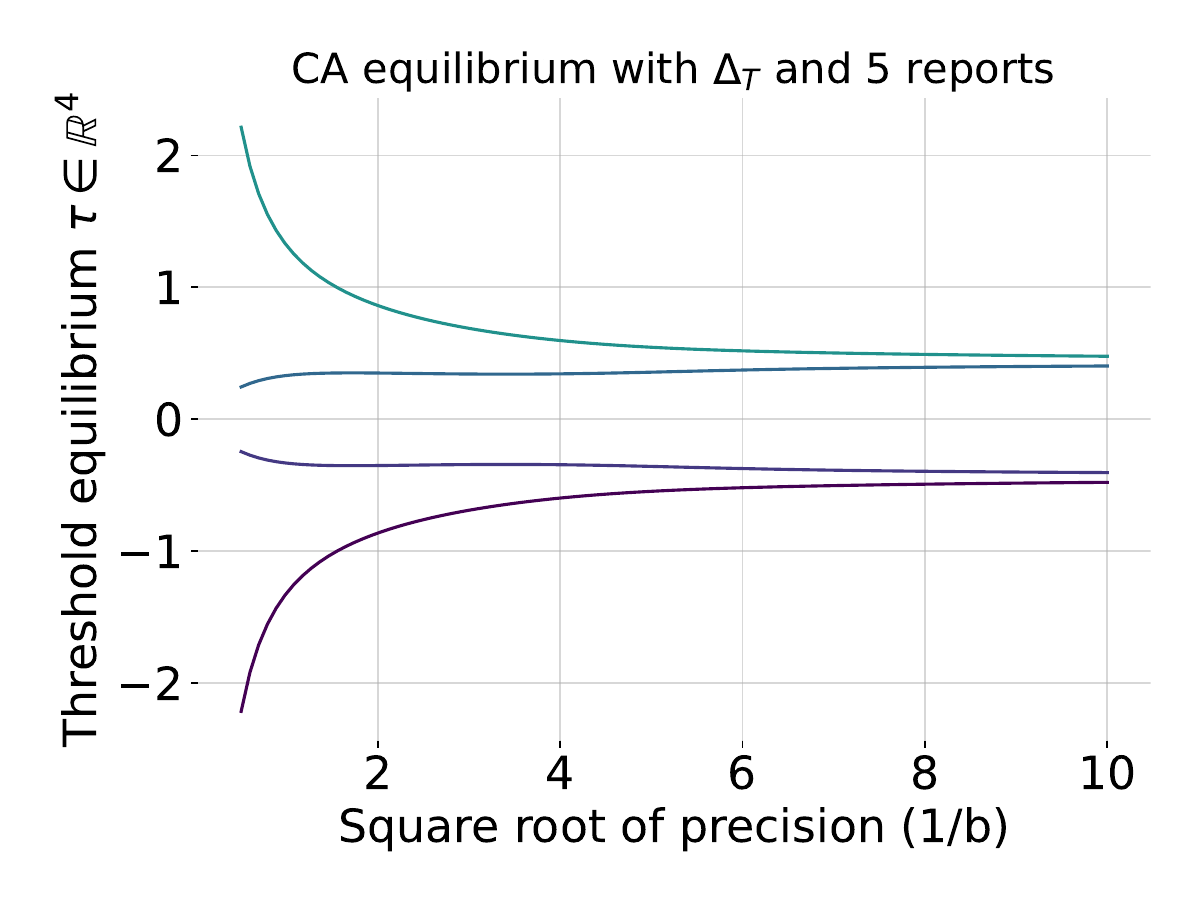}
    \end{subfigure}
    \hfill
    \begin{subfigure}[t]{0.45\textwidth}
        \includegraphics[width=\textwidth]{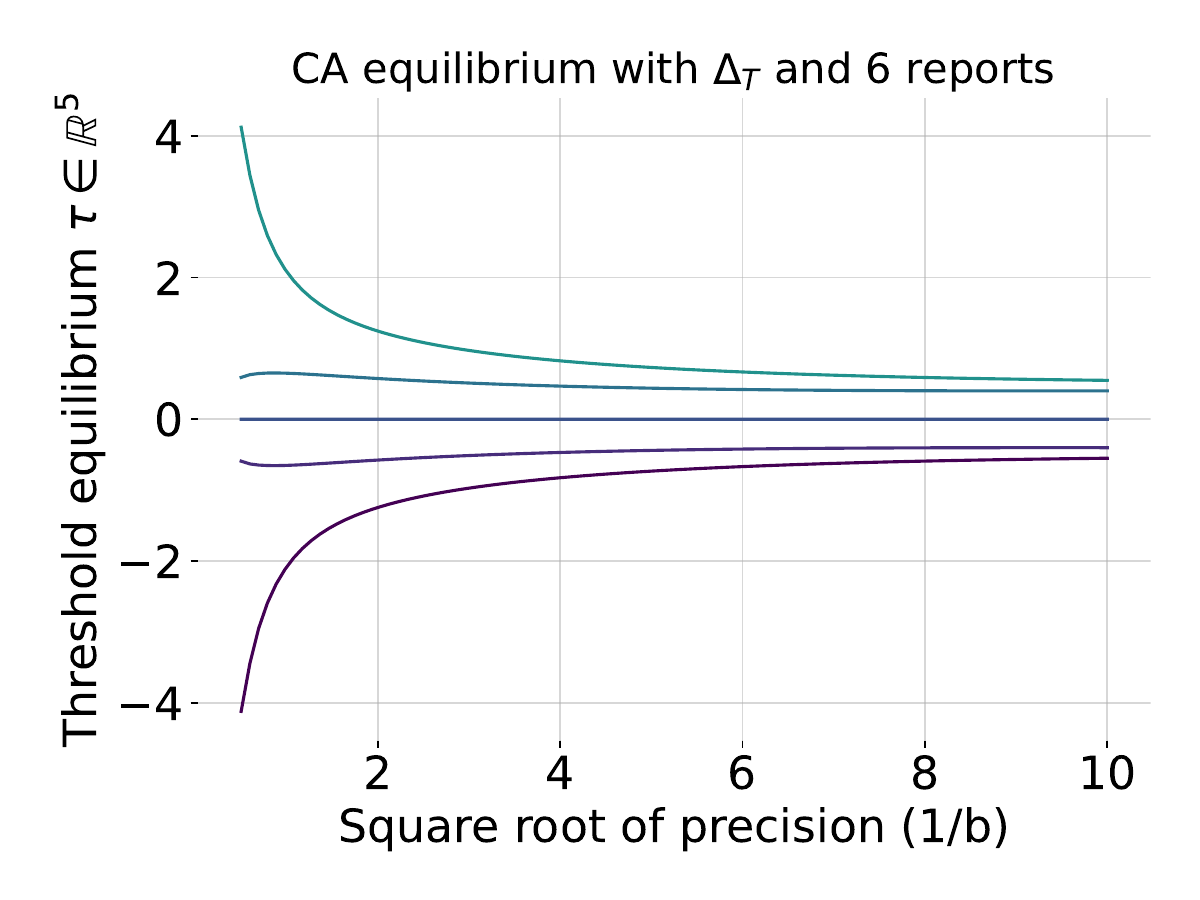}
    \end{subfigure}
    \hfill
    \begin{subfigure}[t]{0.45\textwidth}
        \includegraphics[width=\textwidth]{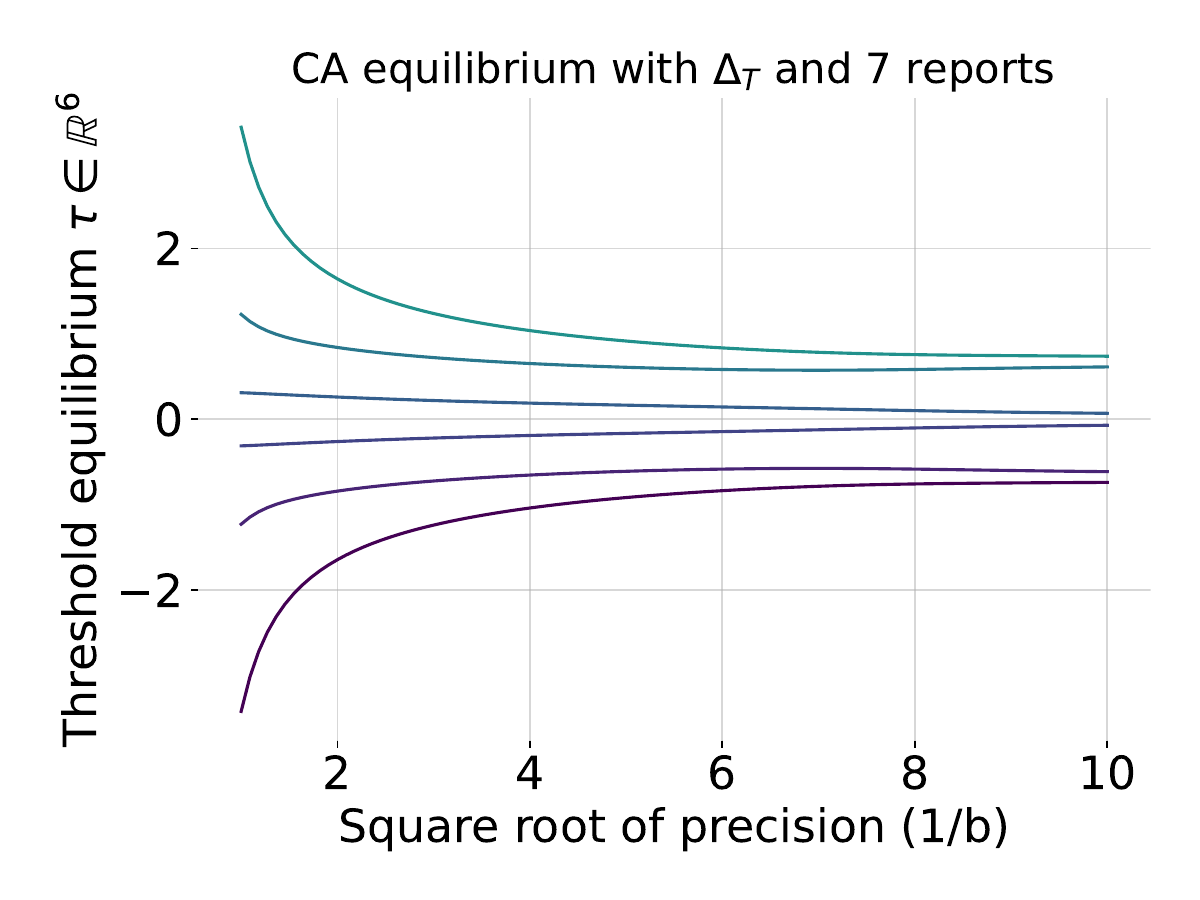}
    \end{subfigure}
    \hfill 
        \begin{subfigure}[t]{0.45\textwidth}
        \includegraphics[width=\textwidth]{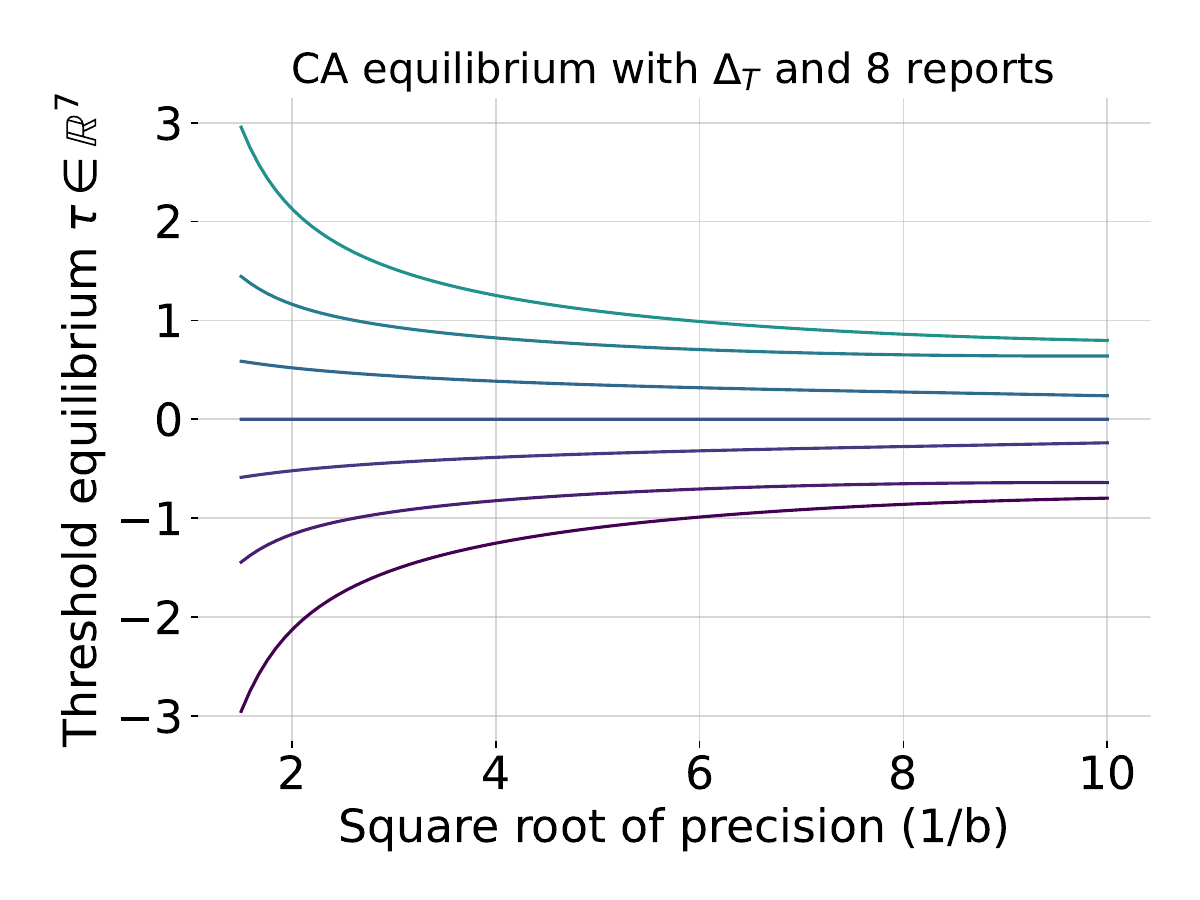}
    \end{subfigure}
    \caption{
      Plots of the unique stable equilibrium of dimension $m-1$ for CA with score matrix $\Delta_T$, for $m=5,6,7,8$, varying the square root of individual signal precision $1/b$.
    }
    \label{fig:larger-m-tri-ca}
\end{figure}

\section{Results for Finite Signal Spaces} \label{app:finite-signal}

\subsection{Output Agreement}
In this section we extend our equilibrium existence and dynamic results under OA to the setting of $m \geq 3$ categorical signals and a binary report space.
Most importantly, our results show that, just as in the real-valued signal case, trivial agreement remains a stable equilibrium; and in many cases, there exists only one other unstable equilibrium corresponding to a threshold between two signal bins.

\paragraph{Equilibrium characterization.}
We consider the setting of a signal space $\CS$ of size $m$. 
Throughout this section, we assume there is some inherent \emph{ordering} to the meaning of these signals, a fact which is already baked into the model of real-valued signals.

\begin{condition} \label{cond:fosc}
    We assume that the set of signals $\CS$ is ordered by strict first-order stochastic dominance. 
    That is, we can index the signals under a strict total order as $1 < 2 < \dots < m$ such that for all $x \in \CS$,
    \begin{equation*}
        \sum_{x' \leq x} \Pr[X' = x' \mid X = 1] > \sum_{x' \leq x} \Pr[X' = x' \mid X = 2] > \dots > \sum_{x' \leq x} \Pr[X' = x' \mid X = m].
    \end{equation*}
    In other words, for any $x, s \in \CS$, let $P(x; s) = \sum_{x' \leq x} \Pr[X' = x' \mid X = s]$. 
    Then for all $x \in \CS$, $P(x; 1) > P(x; 2) \dots > P(x; m)$.
\end{condition}

We observe that we can map our monotonicity conditions for real-valued signals to Condition~\ref{cond:fosc} directly. 

\begin{observation} \label{obs:real-valued-to-discrete}
  Take a joint distribution over real-valued signals $x$ and discretizes it into $m$ bins separated by thresholds $\tau_0, \tau_1, \dots, \tau_m$ (where $\tau_0 = -\infty$, $\tau_m = \infty$).
  Let the function $\Pr[X' \leq \tau_k \mid X = x]$ be monotone decreasing for each $k$ (where both $\tau_k$ and $x$ are real-valued).
  Then in the discrete setting where agents only see which bin $s$ their signal is in, $P(k; s) = \Pr[X' \leq \tau_k \mid \tau_{s-1} < x \leq \tau_s] = \E_{x \mid \tau_{s-1} < x \leq \tau_s} \left[ \Pr\left[X' \leq \tau_k \mid X = x\right] \right]$.
  It immediately holds by monotonicity that Condition~\ref{cond:fosc} is satisfied.
\end{observation}

Because the signal space is finite, ``threshold'' equilibria now correspond to thresholds between two adjacent report bins $k$ and $k+1$.
In cases where e.g. the mode of an underlying real-valued distribution does not lie exactly on such a threshold, it makes sense to study symmetric \emph{mixed} Nash equilibria.
Here we restrict our attention to equilibria that mix over exactly two adjacent thresholds, a reasonable restriction given the sequential ordering of signals under Condition~\ref{cond:fosc}.
For example, if we take the Gaussian model with three signals, based on our real-valued results we expect agents to mix over a threshold between signals $1$ and $2$, and signals $2$ and $3$.

The mixed strategy of each agent is thus now a function $\sigma_i: \CS \to \Delta\left(\twobin\right)$; WLOG, we instead define $\sigma_i: \CS \to [0, 1]$, so that $\sigma_i(x)$ signifies the probability put on report $L$ given signal $x$.
Then we define pure and mixed threshold strategies as follows: 

\begin{definition} \label{def:pure-strategy}
    We define a \emph{pure threshold strategy} above signal $x$ as a strategy $\sigma$ such that $\sigma(x') = 1$ for all $x' \leq x,$ and $\sigma(x') = 0$ otherwise.
\end{definition}

\begin{definition} \label{def:mixed-strategy}
    We define a \emph{mixed threshold strategy} $\sigma$ at signal $x \in \CS$ as a mixed strategy where $0 < \sigma_x(x) < 1$; for any $x' < x,$ $\sigma_x(x') = 1$; and for any $x' > x$, $\sigma_x(x') = 0.$
\end{definition}

Define $\beta(x) = \Pr[X' = \cdot \mid X = x]$ as the posterior distribution over signals conditional on receiving signal $x \in \CS$.
Then agent $i$'s interim expected utility after seeing signal $x$ under OA against mixed strategy $\sigma$ is
\begin{align*}
    U_i(\sigma_i, \sigma, x) &= \E_{x' \sim \beta(x)} \E_{\substack{r_i \sim \sigma_i(x) \\ r_j \sim \sigma(x')}} [ \ones[r_i = r_j]] \\
    &= \sum_{x' \in \CS} \Pr[X' = x' \mid X = x] \left( \sigma_i(x) \sigma(x') + (1 - \sigma_i(x))(1 - \sigma(x')) \right).
\end{align*}

Note in particular that 
\begin{align*}
  U_i(L, \sigma, x) &= \sum_{x' \in \CS} \Pr[X' = x' \mid X = x] \sigma(x') \\
  U_i(H, \sigma, x) &= \sum_{x' \in \CS} \Pr[X' = x' \mid X = x] (1 - \sigma(x')). 
\end{align*}

Now, consider a symmetric threshold equilibrium $\sigma$ at $x$, so that $0 < \sigma(x) < 1,$ i.e. both $L$ and $H$ are in the support of the mixed strategy. 
Then by the principles of indifference, we must have
\begin{align}
    U_i(L, \sigma, x) &= U_i(H, \sigma, x) \nonumber \\
    \sum_{x' \in \CS} \Pr[X' = x' \mid X = x] \sigma(x') &= \sum_{x' \in \CS} \Pr[X' = x' \mid X = x] (1 - \sigma(x')) \nonumber \\
   P(x - 1; x) + \Pr[X' = x \mid X = x] \sigma(x) &= 1/2. \label{eq:mixed-br}
\end{align}

Meanwhile, for $x' < x$, $\sigma(x) = 1$ (i.e. playing $L$ as a function of signal $x$) must be a best response.
In other words, $U_i(L, \sigma, x') > U_i(H, \sigma, x')$, or
\begin{equation} \label{eq:L-br}
    P(x - 1; x') + \Pr[X = x \mid X' = x'] \sigma(x) > 1/2,
\end{equation}
Similarly, for $x' > x$, we must have $U_i(L, \sigma, x') < U_i(H, \sigma, x')$, or
\begin{equation} \label{eq:H-br}
    P(x - 1; x') + \Pr[X = x \mid X' = x'] \sigma(x) < 1/2.
\end{equation}

Under a pure threshold equilibrium above $x$, Equations~\ref{eq:L-br} and~\ref{eq:H-br} hold for $x' \leq x$ and $x' > x$ respectively.
Using this characterization, we can immediately show that uninformative equilibria still exist.

\begin{proposition}
    A symmetric equilibrium exists where $\sigma(x) = 1$ for all $x \in \CS$ (all agents report $L$), and where $\sigma(x) = 0$ for all $x \in \CS$ (all agents report $H$).
\end{proposition}

\begin{proof}
    Assume $\sigma(x) = 1$ for all $x \in \CS$.
    Then $U_i(L, \sigma, x) = \sum_{x' \in \CS} \Pr[X' = x' \mid X = x] = 1 > 0 = U_i(H, \sigma_j, x)$, so that $\sigma$ is an equilibrium. 
    The same argument holds for when $\sigma(x) = 0$ for all $x \in \CS$. 
\end{proof}

For mathematical convenience, we can also define a mixed threshold by extending the discrete set of signals $\{1,2,\dots,m\}$ to a continuous interval.
Specifially, let $\tau \in [0, m]$ indicate a mixed threshold strategy at $x = \lceil \tau \rceil$, i.e. $\tau - \lfloor \tau \rfloor = \sigma(x)$ is the probability of reporting $L$ when receiving signal $x$. 
If $\tau \in \{0,1,2,\dots,m\}$, then $\tau$ represents a pure threshold strategy above $x = \tau$.
In particular, $\tau = 0$ represents the uninformative equilibrium of always reporting $H$, and similarly $\tau = m$ represents the uninformative equilibrium of always reporting $L$. 
We denote such a symmetric mixed (or pure) strategy as $\sigma_{\tau}$. 
For convenience, we let $\Pr[X' \leq 0 \mid X = x] = 0$ for any $x \in \CS$.
Now, we define $\PC(\tau; x): \{0,1,\dots,m\} \to \reals$ for a fixed threshold $\tau \in [0, m]$ as
\begin{equation}
    \PC(\tau; x) =
    \begin{cases}
    (\lceil \tau \rceil  - \tau) \Pr[X' \leq \lfloor \tau \rfloor \mid X = x] + (\tau - \lfloor \tau \rfloor ) \Pr[X' \leq \lceil \tau \rceil \mid X = x] & \tau \notin \{0, 1, \dots, m\}, \\
    \Pr[X' \leq \tau \mid X = x], \tau \in \{0, 1, 2, \dots, m\},
    \end{cases}
\end{equation}

It immediately follows that $U_i(L, \sigma_{\tau}, x) = \PC(\tau; x)$.
Moreover, $\PC(\tau; x)$ is continuous over $\tau$ and monotone decreasing over $x$ by Condition~\ref{cond:fosc}.
Using this continuous extension, then, we can characterize mixed equilibria using the same tools as our previous real-valued analysis. 
To do so, we define the function $G(\tau) = \PC(\tau; \lceil \tau \rceil).$

\begin{proposition} \label{prop:cont-existence-oa}
    Let Condition~\ref{cond:fosc} hold, and define the function $G(\tau) = \PC(\tau; \lceil \tau \rceil).$
    Then there is a mixed equilibrium at $\tau \in (0,m)$ if and only if $G(\tau) = 1/2$.
\end{proposition}

\begin{proof}
    For necessity, we note that any mixed equilibrium $\tau$ must satisfy $U_i(L, \sigma_{\tau}, \lceil \tau \rceil) = G(\tau) = 1/2$ by indifference. 
    For sufficiency, assume $G(\tau) = 1/2,$ so that $U_i(L, \sigma_{\tau}, \lceil \tau \rceil) = 1/2$.
    Then note for any signal $x < \lceil \tau \rceil$, by Condition~\ref{cond:fosc} $U_i(L, \sigma_{\tau}, x) = \PC(\tau; x) > \PC(\tau; \lceil \tau \rceil) = 1/2$.
    Meanwhile, for any signal $x > \lceil \tau \rceil$, $U_i(L, \sigma_{\tau}, x) = \PC(\tau; x) < \PC(\tau; \lceil \tau \rceil ) = 1/2$.
    By definition, then, $\tau$ is a mixed symmetric equilibrium.
\end{proof}

Characterization of pure equilibria follows in a similar way:

\begin{proposition} \label{prop:oa-existence-pure}
    Let Condition~\ref{cond:fosc} hold.
    For OA, there is a pure threshold equilibrium above $\tau = x$ for $x \in \{1,\dots,m-1\}$ if and only if $P(x; x) > 1/2$ and $P(x; x+1) < 1/2$.
\end{proposition}

\begin{proof}
    For sufficiency, assume $P(x; x) > 1/2$ and $P(x; x + 1) < 1/2.$ 
    Then by Condition~\ref{cond:fosc}, for any $x' \leq x,$ $P(x;x') \geq P(x;x) > 1/2$, and for any $x' > x,$ $P(x;x') \leq P(x;x-1) < 1/2$.
    By Definition~\ref{def:pure-strategy}, then, $\tau$ is a pure strategy equilibrium.
    For necessity, assume $x$ is a threshold equilibrium.
    Then by definition, for any $x' \leq x,$ $P(x;x') > 1/2$, and for any $x' > x$ (including $x+1$), $P(x;x') < 1/2$. 
\end{proof}

\paragraph{Dynamics.}
We can now study the same form of best response dynamics to refine our set of equilibria. 
The designer sets some initial $\tau(0) \in [0,m]$, and at each continuous step a small fraction of agents best respond with a threshold $\hat \tau \in \text{BR}(\tau)$, where $\text{BR}(\tau)$ is now a \emph{set} of best responses.
(Note if $x$ and $x+1$ are both best responses, so is every $\tau \in [x, x+1]$.)
We thus consider the dynamics of the differential inclusion system $\dot \tau = \tau - \text{BR}(\tau)$.
To prove statements of stability, we consider any threshold sufficiently close to an equilibrium at $\tau^*$ and show any best response in $\text{BR}(\tau)$ follows a direction toward (or against) $\tau^*$.

\begin{proposition}
    Let Condition~\ref{cond:fosc} hold.
    The best response set $\text{BR}(\tau)$ to a threshold strategy $\tau$ corresponds to either a set of threshold strategies $\hat \tau \in [x, x+1]$, or a pure threshold strategy above $x$, for some $x \in \{0,1,\dots,m\}.$
\end{proposition}

\begin{proof}
    Assume other agents are playing according to a (potentially mixed) threshold strategy $\tau$, and consider the utility of agent $i$.
    Then $\PC(\tau; x)$ is strictly monotone decreasing over the order of $x$ under Condition~\ref{cond:fosc}.
    It follows there is a unique value $x \in \{0,1,\dots,m\}$ such that for $x' \leq x$, $U_i(L, \sigma_{\tau}, x') = \PC(\tau; x') \geq 1/2$, and for $x' > x$, $\PC(\tau; x') < 1/2$;
    if $U_i(L, \sigma_{\tau}, x) > 1/2$, $x$ is a unique pure-strategy best response, while if $U_i(L, \sigma_{\tau}, x) = 1/2$, any $\tau \in [x-1, x]$ is a best response threshold strategy.
\end{proof}

\begin{theorem} \label{thm:oa-mixed-unstable}
    Let Condition~\ref{cond:fosc} hold, and let $\tau^*$ be a mixed equilibrium under OA (i.e. $\tau \in (0,m) \backslash \{1, 2, \dots, m\}$).
    Then $\tau^*$ is unstable.
\end{theorem}

\begin{proof}
    Take current threshold $\tau = \tau(t)$ in a sufficiently small neighborhood around $\tau^*$ such that $\tau$ is not an integer (i.e., $\tau$ is a mixed strategy), and consider an agent $i$ receiving signal $x = \lceil \tau \rceil$.
    There are three cases to consider:
    \begin{enumerate}
        \item $G(\tau) = \frac{1}{2}.$
        Then by Proposition~\ref{prop:cont-existence-oa}, $\tau$ is an equilibrium. 
        \item $G(\tau) > \frac{1}{2}.$
        Then $U_i(L, \sigma_{\tau}, \lceil \tau \rceil) > \frac{1}{2}$, so it follows the best pure response of agent $i$ is to report $L$.
        By Condition~\ref{cond:fosc}, it follows that for any $\hat \tau \in \text{BR}(\tau)$, $\hat \tau > \tau$.
        Thus $\dot \tau = \hat \tau - \tau > 0.$
        \item $G(\tau) < \frac{1}{2},$ so that $U_i(L, \sigma_{\tau}, \lceil \tau \rceil) < \frac{1}{2}$.
        By Condition~\ref{cond:fosc}, it follows that for any $\hat \tau \in \text{BR}(\tau)$, $\hat \tau < \tau$.
        Thus $\dot \tau = \hat \tau - \tau < 0.$
    \end{enumerate}
    It follows that if $G(\tau) = \PC(\tau; \lceil \tau \rceil)$ is strictly increasing at equilibrium $\tau^*$, then $\tau^*$ is \emph{unstable}. 
    But note that for $\tau \notin \{0,1,\dots,m\},$
    $G'(\tau) = P(\lceil \tau \rceil; \lceil \tau \rceil) - P(\lfloor \tau \rfloor; \lceil \tau \rceil) > 0$, meaning $\PC(\tau; \lceil \tau \rceil)$ is strictly increasing. 
\end{proof}

\begin{theorem}
    Let Condition~\ref{cond:fosc} hold, and let $\tau^*$ be a pure symmetric equilibrium under OA. 
    Then $\tau^*$ is stable.
\end{theorem}

\begin{proof}
    Let $\tau^*$ be a pure equilibrium, so that by Proposition~\ref{prop:oa-existence-pure}, $P(\tau^*; \tau^*) > 1/2$ and $P(\tau^*; \tau^* + 1) < 1/2$.
    Now consider agents best responding to some $\tau < \tau^*$ sufficiently close to $\tau^*$, so that $\lceil \tau \rceil = \tau^*$. 
    Then 
    \begin{align*}
        U_i(L, \sigma_{\tau}, \tau^*) &= \PC(\tau; \tau^*) \\
        &= (\tau^*  - \tau) P(\lfloor \tau \rfloor ; \tau^*) + (\tau - \lfloor \tau \rfloor ) P(\tau^*; \tau^*) \\
        &\geq (\tau - \lfloor \tau \rfloor) P(\tau^*; \tau^*),
    \end{align*}

    where the inequality holds since $ P(\lfloor \tau \rfloor ; \tau^*) \geq 0$.
    
    Now, note for $\tau > \frac{1}{2P(\tau^*; \tau^*)} + \lfloor \tau \rfloor$, $U_i(L, \sigma_{\tau}, \tau^*) > 1/2$.
    Since $P(\tau^*; \tau^*) > 1/2$, we have $\frac{1}{2P(\tau^*; \tau^*)} < 1$, so there exists some nontrivial $\tau' < \tau^*$ such that for all $\tau \in [\tau', \tau^*)$, $ U_i(L, \sigma_{\tau}, \tau^*) > \frac{1}{2}$.
    This means that $L$ is a better response when one is best responding to $\tau$, so in the dynamic update of the threshold, $\hat \tau > \tau$ and $\dot \tau = \hat \tau - \tau > 0$.
    A similar argument shows that for $\tau > \tau^*$ sufficiently close to $\tau^*,$ $ U_i(L, \sigma_{\tau}, \tau^* + 1) < \frac{1}{2}$, so that $\dot \tau = \hat \tau - \tau < 0$.
    Thus $\tau^*$ is stable.
    
\end{proof}
Using similar logic, we can show that all uninformative equilibria are stable.

\begin{theorem}
  Let Condition~\ref{cond:fosc} hold.
  Then uninformative equilibria are stable. 
\end{theorem}

\begin{proof}
  Let $\tau^* = 0$, and take some $\tau = \epsilon < 1$.
  Consider an agent who receives signal $1 = \lceil \tau \rceil$ best responding to $\tau$. 
    Then 
    \begin{align*}
        U_i(L, \sigma_{\tau}, 1) &= \tau P(1; 1).
    \end{align*}
    It follows for sufficiently small $\epsilon$ that $U_i(L, \sigma_{\tau}, 1) < 1/2$, so that agents receiving signal $1$ (and by Condition~\ref{cond:fosc}, all higher signals) will always report $H$, and therefore follow a threshold of $\tau = 0$.
    Thus $\dot \tau = \hat \tau - \tau < 0$, and the uninformative equilibrium at $\tau = 0$ is stable. 
    Similar logic follows if $\tau = m$.
\end{proof}

It immediately follows by topology of the dynamics that between any two adjacent pure equilibria, there must be a mixed equilibrium.

\begin{corollary}
  Let Condition~\ref{cond:fosc} hold.
  Then between any two adjacent pure threshold equilibria, there exists a mixed threshold equilibrium.
\end{corollary}

\paragraph{Gaussian model.}

We extend the real-valued Gaussian model to a \emph{discretized Gaussian model} in the following way: for parameters $a,b$, we map the $m$ quartiles of the marginal distribution to $m$ equal-probability discrete bins separated by $m-1$ thresholds, with each bin corresponding to a discrete signal.
This induces a discrete joint distribution.
As the main text proves, we know for any real-valued threshold $\tau$ that $\Pr[X' \leq \tau \mid X = x]$ is monotone decreasing.
By Observation~\ref{obs:real-valued-to-discrete}, then, Condition~\ref{cond:fosc} holds.
Thus we can apply the previous section's results to understand what mixed and pure equilibria exist.

In general, based on our results in the previous section, the following pattern emerges: if there is a (stable) pure equilibrium at some signal $x \in \CS$, it is surrounded by (unstable) mixed equilibria.
We find that unless individual agent noise is sufficiently low, if $m$ is odd there is a \emph{single} unstable, mixed equilibrium at $\tau = m/2$, so that dynamics inexorably lead to uninformative consensus.
If $m$ is even, there is a single stable equilibrium at $\tau = m/2$, with its basin of attraction shrinking as precision decreases.
As individual signal precision grows, we do find that more stable, pure equilibria emerge beyond $0$, but only at larger precision values (see Figure~\ref{fig:bifurcation-finite-oa} for visualization and discussion).
Regardless, uninformative equilibria remain stable, so that agents will inevitably drift toward blind agreement unless the designer is able to set initial thresholds close enough to the stable, pure equilibria if they exist.

\begin{figure}[htbp]
    \centering
    \begin{subfigure}[t]{0.32\textwidth}
        \includegraphics[width=\textwidth]{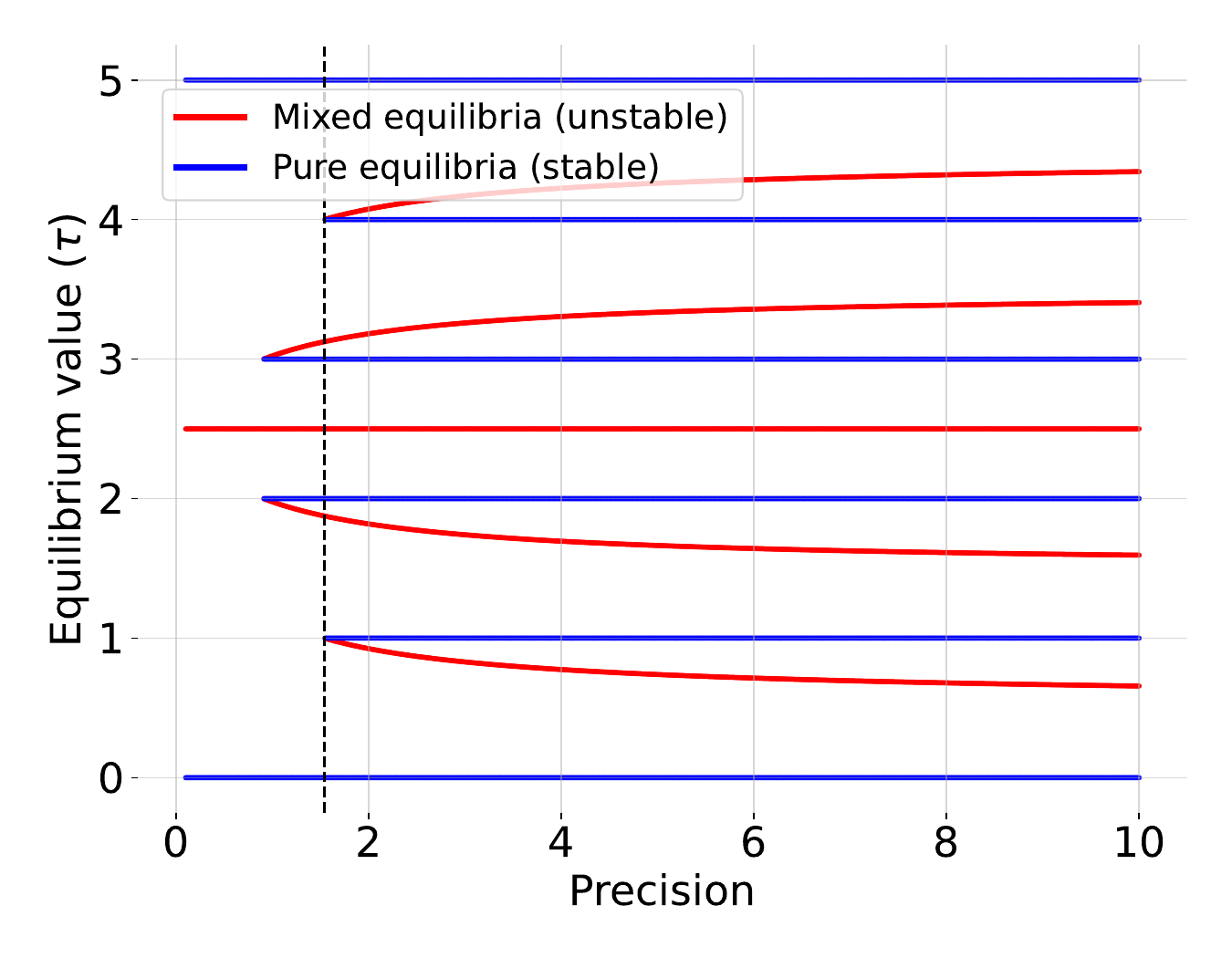}
    \end{subfigure}
    \hfill
    \begin{subfigure}[t]{0.32\textwidth}
        \includegraphics[width=\textwidth]{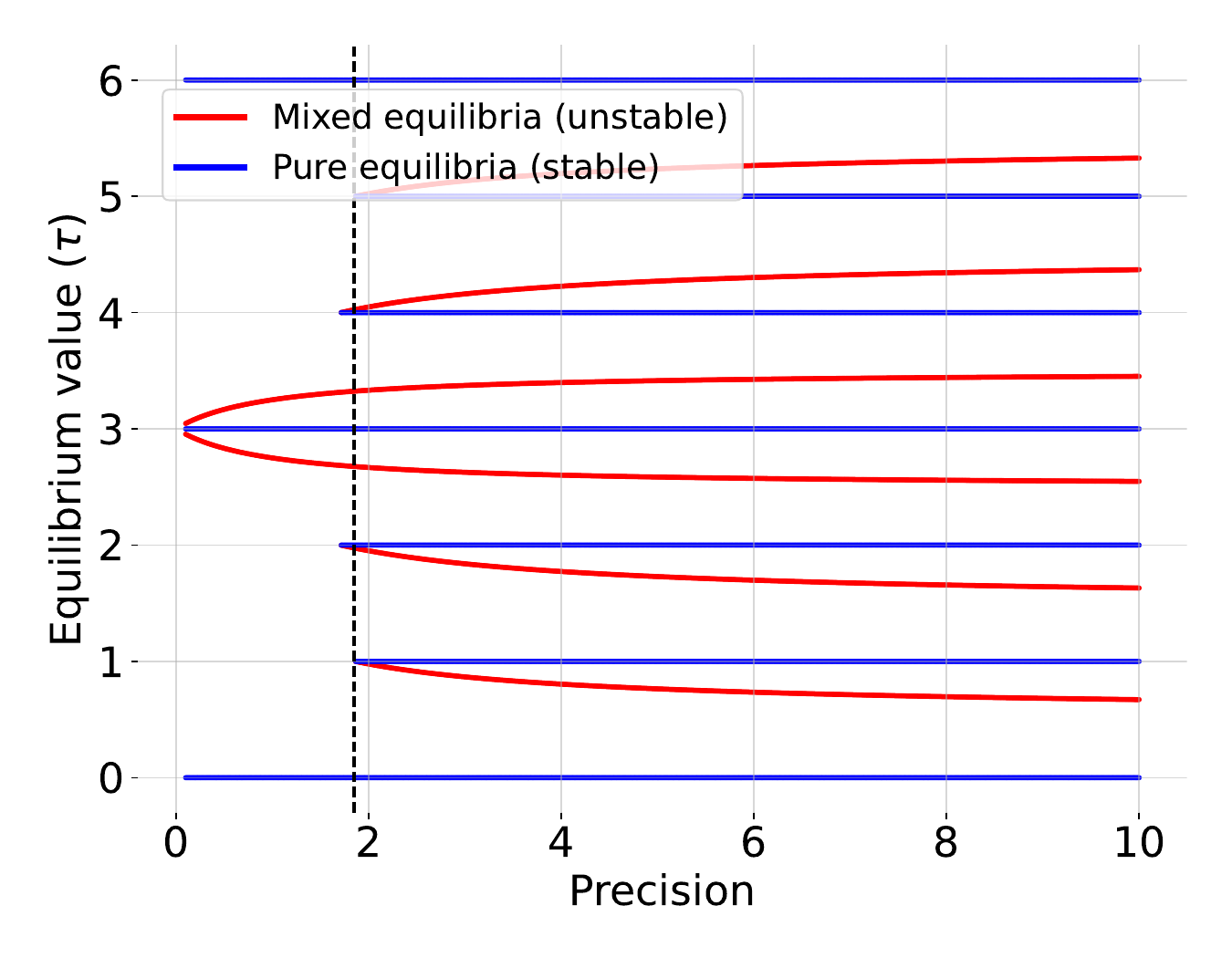}
    \end{subfigure}
    \hfill
    \begin{subfigure}[t]{0.32\textwidth}
        \includegraphics[width=\textwidth]{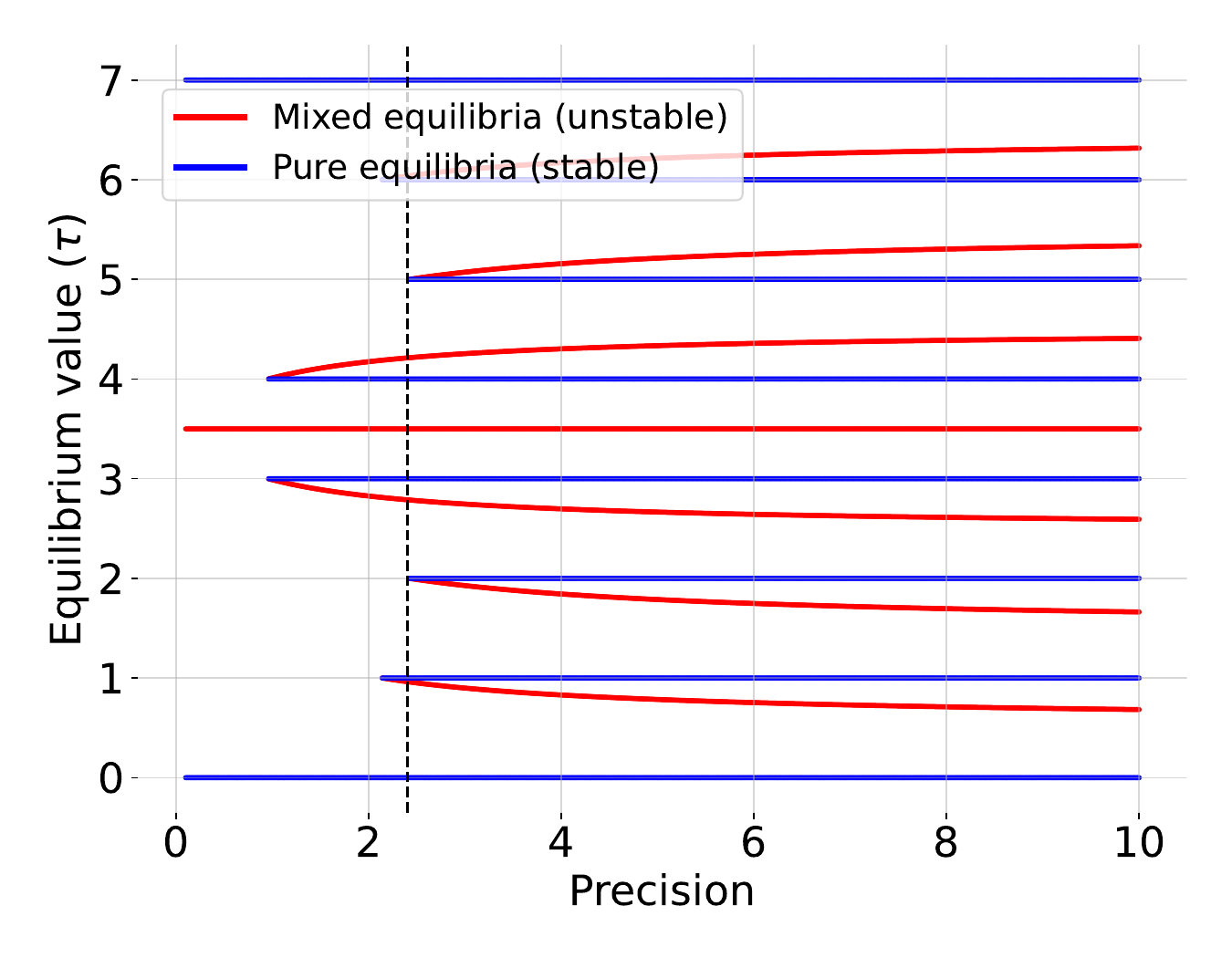}
    \end{subfigure}

    \caption{Bifurcation diagram of equilibria under OA for (left to right) $m=5,6$ and $7$ finite ordered signals. 
    As our theory shows, the red lines indicate unstable mixed equilibria, while the blue indicate stable pure equilibria. 
    We always have stable equilibria at $\tau=0$ and $\tau=m$, representing blind agreement (always reporting $H$ or $L$).
    The dotted black lines indicate the value of individual signal precision $(1/b^2)$ at which there are $m-1$ possible (nontrivial) stable equilibria.
    }
    \label{fig:bifurcation-finite-oa}
\end{figure}

\subsection{Dasgupta-Ghosh}

We now study the same model of threshold equilibria and dynamics, but under the DG mechanism. 

\paragraph{Equilibrium characterization.}
We define DG in the same way as in the main text, but note here that the prior penalty term $\pi_{r_i}$ for report $r_i$ is now $\pi_{r_i} = \Pr[X = r_i]$.
Agent $i$'s interim expected utilities for reporting $L$ or $H$ after seeing signal $x$ under DG against symmetric strategy $\sigma$ are then
\begin{align*}
    U_i(L, \sigma, x) &= \sum_{x' \in \CS } (\Pr[X' = x' \mid X = x] - \Pr[X = x']) \sigma(x'), \\
    U_i(H, \sigma, x) &= \sum_{x' \in \CS} (\Pr[X' = x' \mid X = x] - \Pr[X = x']) (1 - \sigma(x')).
\end{align*}

Consider a symmetric threshold equilibrium $\sigma$ at $x$ so that $0 < \sigma(x) < 1,$ i.e. both $L$ and $H$ are in the support of the symmetric mixed strategy $\sigma$.
Then by the principles of indifference, we must have
\begin{align}
    U_i(L, \sigma, x) &= U_i(H, \sigma, x) \nonumber \\
    2 \left(\sum_{x' \in \CS} \Pr[X' = x' \mid X = x] - \Pr[X = x'] \right) \sigma(x') &= 0 \nonumber \\
    P(x-1; x) + \Pr[X' = x \mid X = x] \sigma(x) &= \Pr[X \leq x-1] + \Pr[X = x] \sigma(x). \label{eq:mixed-br-dg}
\end{align}

Meanwhile, for $x' < x$, $\sigma(x) = 1$ and $U_i(L, \sigma, x') > U_i(H, \sigma, x')$ or
\begin{equation} \label{eq:L-br-dg}
    P(x-1; x')  + \Pr[X = x \mid X' = x']\sigma(x') > \Pr[X \leq x - 1]  + \Pr[X = x]\sigma(x');
\end{equation}
and  for $x' > x$, $\sigma(x) = 0$ and $U_i(L, \sigma, x') < U_i(H, \sigma, x')$ or
\begin{equation}  \label{eq:H-br-dg}
    P(x-1; x')  + \Pr[X = x \mid X' = x']\sigma(x') < \Pr[X \leq x - 1]  + \Pr[X = x]\sigma(x').
\end{equation}
Under a pure threshold equilibrium above $x$, Equations~\ref{eq:L-br-dg} and~\ref{eq:H-br-dg} hold for $x' \leq x$ and $x' > x$ respectively.
Using these equations, we immediately observe that uninformative equilibria still exist under DG.

\begin{proposition}
    A symmetric equilibrium exists where $\sigma(x) = 1$ for all $x \in \CS$ (all agents report $L$), and where $\sigma(x) = 0$ for all $x \in \CS$ (all agents report $H$).
\end{proposition}

\begin{proof}
    Assume $\sigma(x) = 1$ for all $x \in \CS$.
    Then $U_i(L, \sigma, x) = \sum_{x' \in \CS} \Pr[X' = x' \mid X = x] - \Pr[X = x'] = 0 = U_i(H, \sigma_j, x)$, so that $\sigma$ is an equilibrium. 
    The same argument holds when $\sigma(x) = 0$ for all $x \in \CS$. 
\end{proof}

Now, we leverage the same continuous definition of threshold equilibria $\tau \in [0,m]$ and function $\PC(\tau; x)$ to characterize behavior. 
We also define an analagous, continuous definition of the prior $F$, with
\begin{equation}
    F(\tau) =
    \begin{cases}
    (\lceil \tau \rceil  - \tau) \Pr[X' \leq \lfloor \tau \rfloor] + (\tau - \lfloor \tau \rfloor ) \Pr[X' \leq \lceil \tau \rceil], & \tau \in \notin \{0,1,2, \dots,m\}, \\
    \Pr[X' \leq \tau], \tau \in \{0, 1, 2, \dots, m\}.
    \end{cases}
\end{equation}
It immediately follows by definition that $U_i(L, \sigma_{\tau}, x) = \PC(\tau; x) - F(\tau)$; again, $\PC(\tau; x) - F(\tau)$ is strictly monotone decreasing over the order of $x$, and continuous over $\tau$. 
As in our analysis of OA, we can then characterize pure and mixed threshold equilibria.

\begin{proposition} \label{prop:cont-existence-dg}
    Let Condition~\ref{cond:fosc} hold.
    Then there is a mixed equilibrium at $\tau$ if and only if $G(\tau) = F(\tau)$.
\end{proposition}

\begin{proof}
    For necessity, we note that any mixed equilibrium $\tau$ must satisfy $G(\tau) = F(\tau)$ by the indifference principle in Equation~\ref{eq:mixed-br-dg}.

    For sufficiency, assume $G(\tau) = F(\tau),$ so that $U_i(L, \sigma_{\tau}, \lceil \tau \rceil) = 0$.
    Then note for any signal $x < \lceil \tau \rceil$, by Condition~\ref{cond:fosc} $U_i(L, \sigma_{\tau}, x) = \PC(\tau; x) - F(\tau) > \PC(\tau; \lceil \tau \rceil) - F(\tau) = 0$.
    Meanwhile, for any signal $x > \lceil \tau \rceil$, $U_i(L, \sigma_{\tau}, x) = \PC(\tau; x) - F(\tau) < \PC(\tau; \lceil \tau \rceil ) - F(\tau) = 0$.
\end{proof}

\begin{proposition} \label{prop:dg-existence-pure}
    Let Condition~\ref{cond:fosc} hold.
    Then a pure threshold at $x$ is an equilibrium if and only if $P(x; x) > F(x)$ and $P(x; x+1) < F(x).$
\end{proposition}

\begin{proof}
    For sufficiency, assume $P(x; x) > \Pr[X \leq x]$ and $P(x; x + 1) <  \Pr[X \leq x].$ 
    Then by Condition~\ref{cond:fosc}, for any $x' \leq x,$ $P(x;x') \geq P(x;x) >  \Pr[X \leq x]$, and for any $x' > x,$ $P(x;x') \leq P(x;x-1) <  \Pr[X \leq x]$.
    For necessity, assume $x$ is a threshold equilibrium.
    Then by definition, for any $x' \leq x,$ $P(x;x') >  \Pr[X \leq x]$, and for any $x' > x$ (including $x+1$), $P(x;x') <  \Pr[X \leq x]$. 
\end{proof}

\paragraph{Dynamics.}
We now have the tools to study the same dynamics model as in the previous section, with $\dot \tau = \text{BR}(\tau) - \tau$.
First we show that pure equilibria are stable. 

\begin{theorem}
    Let Condition~\ref{cond:fosc} hold, and let $\tau^*$ be a symmetric pure equilibrium under DG. 
    Then $\tau^*$ is stable.
\end{theorem}

\begin{proof}
    By Proposition~\ref{prop:dg-existence-pure}, $P(\tau^*; \tau^*) > \Pr[x' \leq \tau^*]$ and $\Pr[\tau^*; \tau^* + 1] <  \Pr[x' \leq \tau^*]$.
    Now consider agents best responding to some $\tau < \tau^*$ sufficiently close to $\tau^*$, so that $\lceil \tau \rceil = \tau^*$. 
    Then 
    \begin{align*}
        U_i(L, \sigma_{\tau}, \tau^*) &= \PC(\tau; \tau^*) - F(\tau) \\
        &= (\tau^*  - \tau) (P(\lfloor \tau \rfloor; \tau^*) - \Pr[X \leq \lfloor \tau \rfloor]) + (\tau - \lfloor \tau \rfloor ) (P(\tau^*; \tau^*) - \Pr[X \leq \tau^*]).
    \end{align*}
    
    By Proposition~\ref{prop:dg-existence-pure}, $P(\tau^*; \tau^*) > \Pr[X \leq \tau^*]$; so we have $ U_i(L, \sigma_{\tau}, \tau^*) = (1 - \lambda) X + \lambda Y$ for some $Y > 0$.
    It immediately follows there exists some $0 < \lambda < 1$ such that $U_i(L, \sigma_{\tau}, \tau^*) > 0$, or equivalently some $\tau' < \tau^*$ such that for all $\tau \in [\tau', \tau^*)$, $ U_i(L, \sigma_{\tau}, \tau^*) > 0$.
  
    Thus $L$ is a better response when one is best responding to $\tau$, so in the dynamic update of the threshold, $\hat \tau > \tau$. 
    It follows $\dot \tau = \hat \tau - \tau > 0$.
    A similar argument shows that for $\tau > \tau^*$ sufficiently close to $\tau^*,$ $ U_i(L, \sigma_{\tau}, \tau^* + 1) < 0$, so that $\dot \tau = \hat \tau - \tau < 0$.
    Thus $\tau^*$ is stable.
\end{proof}

Meanwhile, to analyze the stability of mixed and uninformative equilibria, we impose the following technical condition on the signal structure.

\begin{condition} \label{cond:positive-corr}
    We assume signals indicate \emph{positive correlation}: for each signal $x \in \CS$, $\Pr[X' = x \mid X = x] > \Pr[X = x].$
\end{condition}
This condition is reasonable in discrete settings where observing a signal consistently strengthens an agent's belief that a peer will receive that signal as well.
Under such settings of local positive correlation, we can show both mixed strategy equilibria and uninformative equilibria are unstable.

\begin{theorem}
    Let Conditions~\ref{cond:fosc} and~\ref{cond:positive-corr} hold, and let $\tau^*$ be a mixed strategy under DG.
    Then $\tau^*$ is unstable. 
\end{theorem}

\begin{proof}
    Take equilibrium point $\tau,$ and consider an agent receiving signal $x = \lceil \tau \rceil$.
    There are three cases to consider:
    \begin{enumerate}
        \item $G(\tau) = F(\tau).$
        Then by Proposition~\ref{prop:cont-existence-dg}, $x$ is an equilibrium. 
        \item $G(\tau) > F(\tau).$
        Then the best response at $x$ is to report $L$, so that for any $\hat \tau \in \text{BR}(\tau)$, $\hat \tau > \tau$. 
        Thus $\dot \tau = \hat \tau - \tau > 0.$
        \item $G(\tau) < F(\tau).$
        Then the best response at $x$ is to report $H$, so that for any $\hat \tau \in \text{BR}(\tau)$, $\hat \tau < \tau$. 
        Thus $\dot \tau = \hat \tau - \tau < 0.$
    \end{enumerate}
    It follows that if $G(\tau) - F(\tau)$ is strictly increasing over $\tau$, then $\tau$ is \emph{unstable}. 
    This is equivalent to 
    \[\Pr[\lceil \tau \rceil; \lceil \tau \rceil] - \Pr[\lfloor \tau \rfloor ; \lceil \tau \rceil] > \Pr[X \leq \lceil \tau \rceil] - \Pr[X \leq \lfloor \tau \rfloor],\]
    or
    \[\Pr[X' = \lceil \tau \rceil \mid X = \lceil \tau \rceil] > \Pr[X = \lceil \tau \rceil], \]
    which holds by Condition~\ref{cond:positive-corr}.
\end{proof}

\begin{theorem}
  Let Conditions~\ref{cond:fosc} and~\ref{cond:positive-corr} hold.
  Then uninformative equilibria are unstable. 
\end{theorem}

\begin{proof}
  Let $\tau^* = 0$, and take some $\tau = \epsilon < 1$.
  Consider an agent who receives signal $1 = \lceil \tau \rceil$ best responding to $\tau$. 
    Then 
    \begin{align*}
        U_i(L, \sigma_{\tau}, 1) &= \PC(\tau; 1) - F(\tau) \\
        &= \tau (\Pr[X' \leq 1 \mid 1] - \Pr[X \leq 1]) \\
        &> 0 \qquad \text{(by Condition~\ref{cond:positive-corr}).}
    \end{align*}
    Thus agents receiving signal $1$ (and all higher signals) will prefer reporting $L$, and therefore follow a threshold at or above $\tau = 1$.
    Thus $\dot \tau = \hat \tau - \tau > 0$, and the uninformative equilibrium at $\tau = 0$ is unstable. 
    Similar logic follows if $\tau = m$.
\end{proof}

The following statement then follows immediately from topology of the dynamics:

\begin{corollary}
  Let Conditions~\ref{cond:fosc} and~\ref{cond:positive-corr} hold.
  Between any two adjacent pure threshold equilibria, there exists a mixed threshold equilibrium; and adjacent to each of the uninformative equilibria is a pure equilibrium.
\end{corollary}

\paragraph{Gaussian model.}

We again consider the discretized Gaussian model, now under DG.
First we note that Condition~\ref{cond:positive-corr} is satisfied in this case. 

\begin{proposition}
  Under the discretized Gaussian model, for any parameters $a,b > 0$, Condition~\ref{cond:positive-corr} holds.
\end{proposition}

\begin{proof}
  Condition~\ref{cond:positive-corr} holds if, for any $\tau_1, \tau_2 \in \reals$, $\tau_1 < \tau_2$ and $R = [\tau_1, \tau_2]$,
  $\Pr[(X, X') \in R \times R] > \Pr[X \in R]^2$. 
  Consider the joint probability as a function of correlation coefficient $\rho = \frac{a^2}{a^2 + b^2}$,  $H(\rho) = \Pr[(X, X') \in R \times R]$.
  Continuity of $H$ for $0 < \rho < 1$ immediately follows by continuity of the Gaussian.
  It then suffices for us to show that $H$ is strictly increasing at any $\rho > 0$, since for $\rho = 0$ we have $\Pr[(X, X') \in R \times R] = \Pr[X \in R]^2$.

  Let $\sigma^2 = a^2 + b^2$.
  Then careful analysis of the bivariate Normal~\citep[Equation 4]{drezner1990computation} allows us to express the derivative as $\frac{\partial H}{\partial \rho} = \phi_2(h_1, h_1; \rho) + \phi_2(h_2, h_2; \rho) - 2\phi_2(h_1, h_2; \rho)$ for $\phi_2(x, x'; \rho)$ the standard bivariate Normal distribution, $h_1 = \tau_1 / \sigma$, and $h_2 = \tau_2 / \sigma$.
  Let $A = \phi_2(h_1, h_1; \rho)$, $B = \phi_2(h_2, h_2; \rho)$ and $C = \phi_2(h_1, h_2; \rho)$.
  Because the bivariate Normal distribution satisfies Total Positivity of Order 2 (TP2) for any $0 < \rho < 1$, we know $A B > C^2$.
  Meanwhile, by the AM-GM inequality, $A + B \geq 2 \sqrt{A B}$.
  Put together, then, $A + B > 2 C$, and the statement holds.
\end{proof}

In general, based on our results in the previous section, the following pattern emerges: the uninformative equilibria remain unstable, so there exists at least one stable equilibrium; and if there are two (stable) pure equilibria at some signal $x \in \CS$, there is a mixed equilibrium between them.

Similar to OA with finite signal spaces, we find that unless individual agent noise is sufficiently low, if $m$ is even there is still a \emph{single} equilibrium at $\tau = m/2$; and if $m$ is odd, there are only two pure equilibria at $\tau = \lfloor m/2 \rfloor$ and $\lceil m/2 \rceil$.
As individual signal precision grows, we do find that more stable, pure equilibria emerge beyond $0$, but only at large precision values.
In particular, the bifurcation point at which the number of stable equilibria reaches its maximum amount for a fixed $m$ \emph{increases} over $m$, so that we require higher and higher precision of individual signals in order to achieve more flexibility. 
(See Figure~\ref{fig:bifurcation-finite-dg} for visualization.)

\begin{figure}[htbp]
    \centering
    \begin{subfigure}[t]{0.32\textwidth}
        \includegraphics[width=\textwidth]{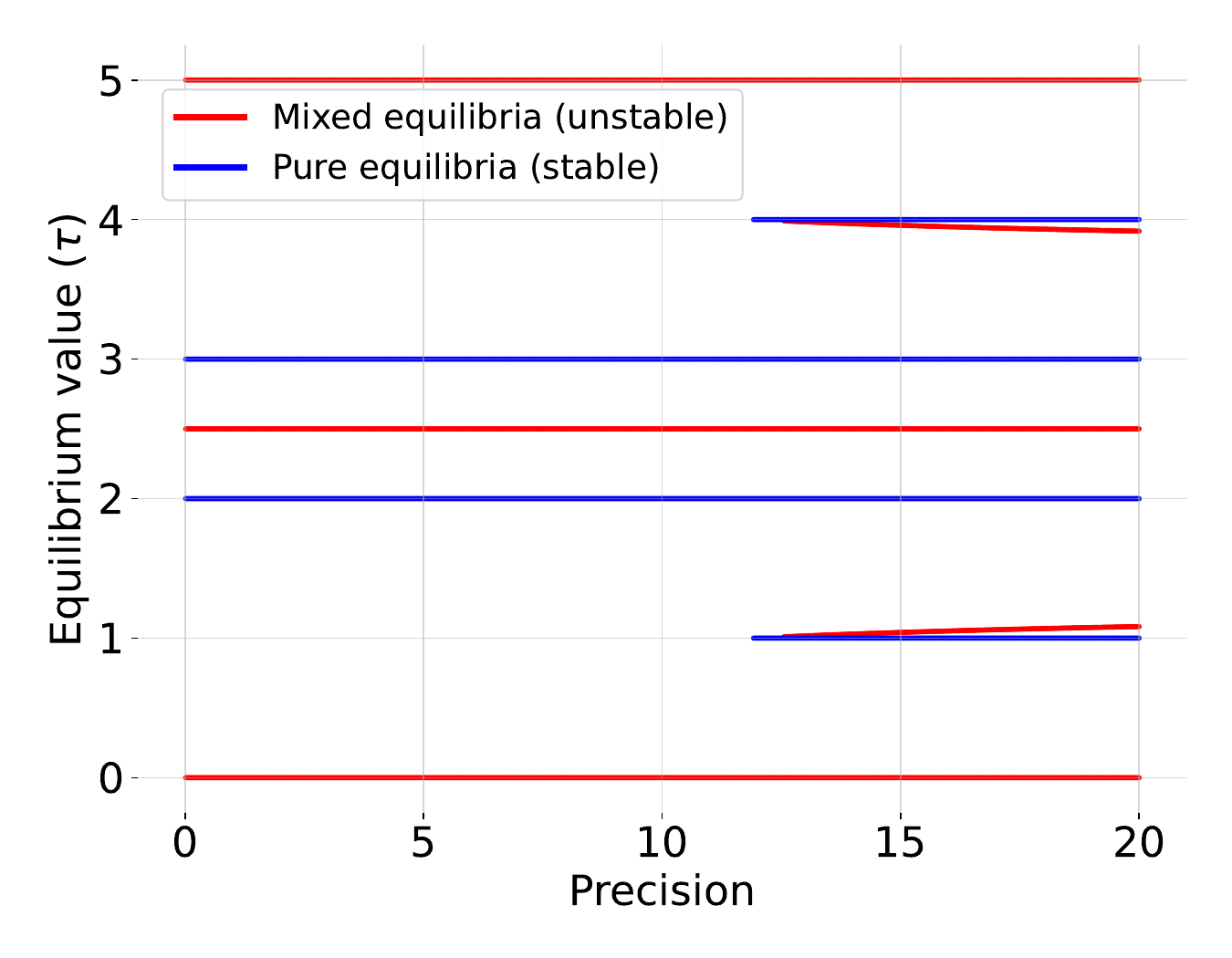}
    \end{subfigure}
    \hfill
    \begin{subfigure}[t]{0.32\textwidth}
        \includegraphics[width=\textwidth]{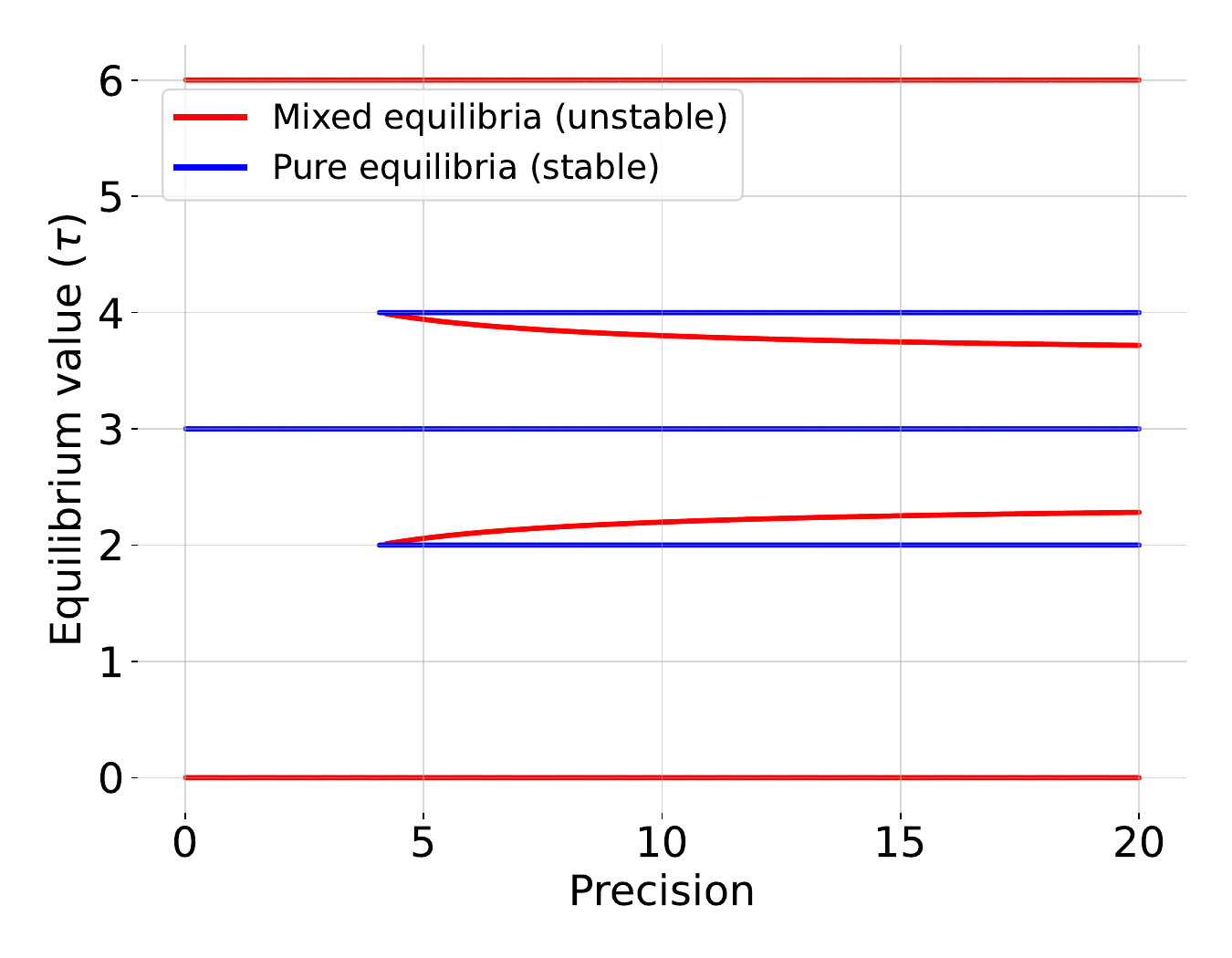}
    \end{subfigure}
    \hfill
    \begin{subfigure}[t]{0.32\textwidth}
        \includegraphics[width=\textwidth]{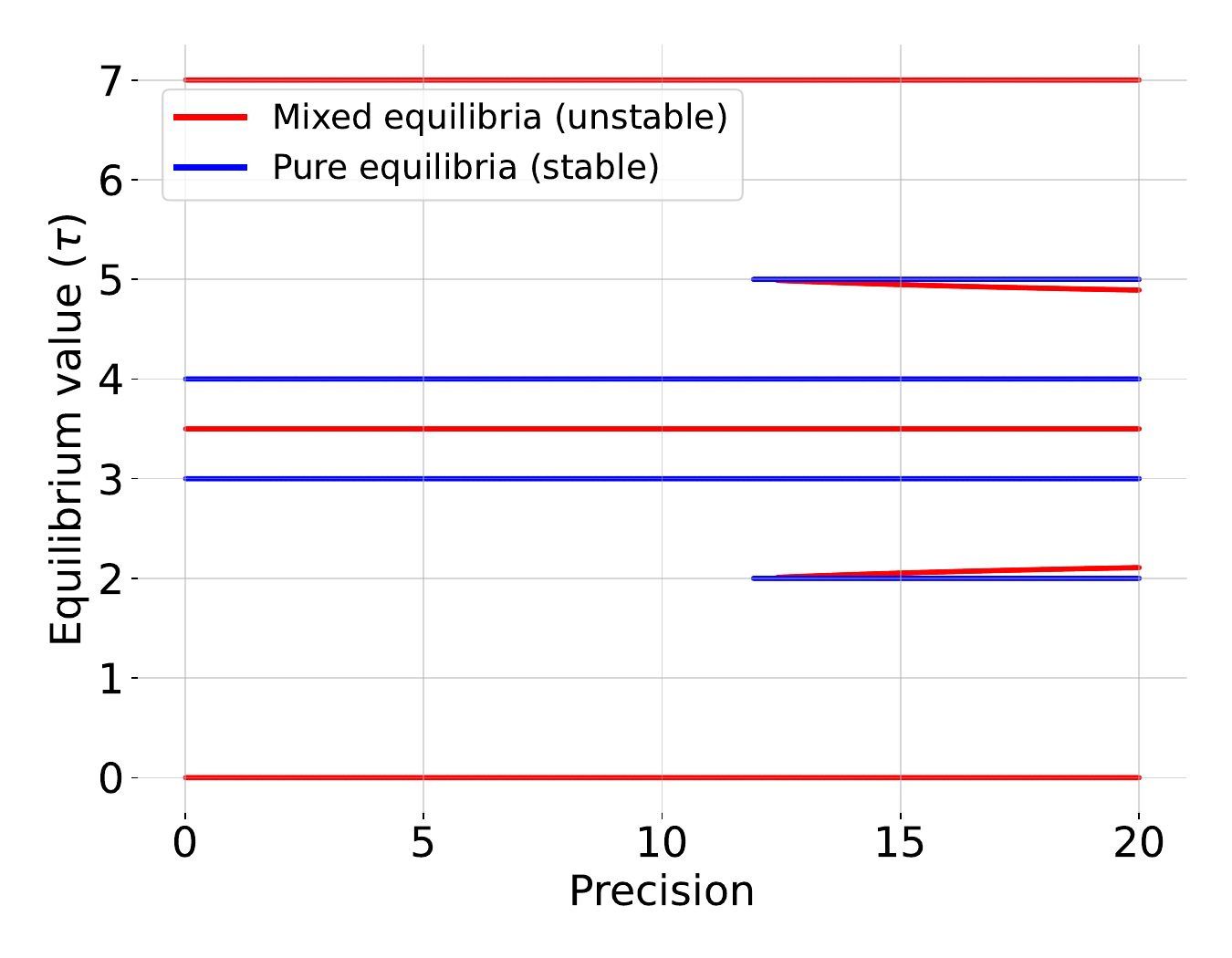}
    \end{subfigure}

    \caption{Bifurcation diagram of equilibria under DG for (left to right) $m=5,6$ and $7$ finite ordered signals. 
    As our theory shows, the red lines indicate unstable mixed equilibria, while the blue indicate stable pure equilibria. 
    We note that within reasonable precision values, the point at which there are $m-1$ possible stable equilibria becomes increasingly large as $m$ grows. 
    Thus the space of regimes where total flexibility is achievable (i.e., there is a stable threshold between any choice of bins) decreases as the signal space becomes finer. 
    }
    \label{fig:bifurcation-finite-dg}
\end{figure}

\end{document}